\documentclass{article}

\usepackage[main, final]{neurips_2026}

\usepackage[utf8]{inputenc} % allow utf-8 input
\usepackage[T1]{fontenc}    % use 8-bit T1 fonts
\usepackage{hyperref}       % hyperlinks
\usepackage{url}            % simple URL typesetting
\usepackage{booktabs}       % professional-quality tables
\usepackage{amsfonts}       % blackboard math symbols
\usepackage{nicefrac}       % compact symbols for 1/2, etc.
\usepackage{microtype}      % microtypography
\usepackage{xcolor}         % colors

\usepackage{amsthm}
 \usepackage{amsmath}
\usepackage{amssymb}
\usepackage{algorithm}
\usepackage{wrapfig}
\usepackage{parskip}
\usepackage[noend]{algpseudocode}
\usepackage{orcidlink}
\usepackage{lineno}
\usepackage{tikzsymbols}
\usepackage{tikz}
\usepackage{mathtools}
\usepackage{thmtools}
\usepackage{thm-restate}
\usepackage[capitalise]{cleveref}
\usepackage{xspace}
\usepackage[new]{old-arrows}
\usepackage{float}
\usepackage{wrapfig}
\usepackage{caption}
\usepackage{subcaption}
\usepackage{graphicx}
\usepackage{mathabx}
\usepackage{stmaryrd}
\usepackage{wasysym}
\usepackage{geometry}
\usepackage{pgfplots}
\pgfplotsset{compat=1.18}
\usepackage{graphicx}
\usepackage{tikz}
\usetikzlibrary{automata,positioning,arrows.meta}
\usepackage{array}
\usepackage{longtable}
\usepackage{multirow}
\usepackage{enumitem}
\newtheorem{lemma}{Lemma}
\newtheorem{theorem}{Theorem}
\newtheorem{definition}{Definition}

\newtheorem{property}{Property}
\newtheorem{remark}{Remark}
\newtheorem{claim}{Claim}
\newtheorem{assumption}{Assumption}
\renewcommand{\epsilon}{\varepsilon}
\renewcommand{\phi}{\varphi}
\usepackage[most]{tcolorbox}
\newtcolorbox{softbox}[1][Definition]{ % Default title is "Definition"
  colframe=blue!20!white,      % The border line (light blue)
  colback=blue!5!white,        % The body background (very pale blue)
  colbacktitle=blue!15!white,  % The TITLE background (slightly darker pale blue)
  
  coltitle=black,              % Title text is black
  coltext=black,               % Body text is black
  
  fonttitle=\bfseries,         % Bold title
  arc=5mm,                     % Rounded corners
  boxrule=0.5pt,               % Thin border
  toptitle=1mm, bottomtitle=1mm, % Padding around the title
  
  title={#1}
}
\newtcolorbox{simplebox}{
  colback=red!5!white,
  colframe=red!20!white,
  coltext=black,
  arc=3mm,
  boxrule=0.5pt,
  left=1.5mm, right=1.5mm,
  top=1mm, bottom=1mm
}
\usepackage[framemethod=TikZ]{mdframed}

\newcounter{quest}[section]

\setlist[itemize]{leftmargin=*}

\newcommand{\size}{\operatorname{size}}

\newcommand{\etr}{\exists\Rb}

\newcommand{\Nb}{\mathbb{N}}
\newcommand{\Rb}{\mathbb{R}}

\newcommand{\Qb}{\mathbb{Q}}

\newcommand{\NP}{\mathsf{NP}}

\newcommand{\PSpace}{\mathsf{PSPACE}}

\newcommand{\Ocomplexity}{\mathcal{O}}

\usepackage{xargs} 
\newcommand{\states}{\mathcal{S}}
\newcommand{\trans}{P}
\newcommand{\cost}{c}
\newcommand{\R}{\mathbb{R}}

\newcommand{\val}{\boldsymbol{v}}
\newcommand{\optval}{\boldsymbol{v}^{\star}}
\newcommand{\discount}{\gamma}
\newcommand{\discountfactor}{\discount}
\newcommand{\discountedpayoff}{\textit{Disc}}
\newcommand{\mc}{M}
\newcommand{\rmc}{\mathcal{M}}
\newcommand{\uncert}{\mathcal{P}}
\newcommand{\uncertaintyset}{\uncert}
\newcommand{\envpol}{\tau}
\newcommand{\actions}{\mathcal{A}}

\newcommand{\rmdp}{\mathcal{R}}

\newcommand{\agentpol}{\sigma}
\newcommand{\radius}{\delta}
\newcommand{\envpolopt}{\envpol^{\star}}
\newcommand{\agentpolopt}{\agentpol^{\star}}
\newcommand{\valopt}{\val^{\star}}
\newcommand{\I}{\mathbb{I}}
\newcommand{\identitymatrix}{\I}
\newcommand{\Bellman}{\mathcal{T}}
\newcommand{\valvector}{\val}
\newcommand{\costvector}{\mathbf{c}}
\newcommand{\pvector}{\mathfrak{p}}
\DeclareMathOperator*{\argmin}{argmin}
\DeclareMathOperator*{\argmax}{argmax}
\newcommand{\nominal}{\hat{\pvector}}

\newcommand{\uvector}{\boldsymbol{u}}
\newcommand{\vvector}{\boldsymbol{v}}
\newcommand{\infnorm}[1]{\left\| #1 \right\|_\infty}

\newcommand{\States}{\states}
\newcommand{\lone}{L_1}
\newcommand{\linf}{L_\infty}
\renewcommand{\succ}{\textit{Succ}}

\newcommand{\Threshold}{\alpha}
\newcommand{\absorbing}{\bot}
\newcommand{\Lp}{L_p}
\newcommand{\defas}{\coloneqq}
\newcommand{\discrmdp}{\mathsf{DiscRMDP}}
\newcommand{\discrmc}{\mathsf{DiscRMC}}

\title{On the Complexity of Discounted Robust MDPs with $L_p$ Uncertainty Sets}

\author{%
  Ali Asadi\\
  Institute of Science and Technology Austria \\
  Klosterneuburg, Austria \\
  \texttt{ali.asadi@ista.ac.at} \\
  \And
  Krishnendu Chatterjee\\
  Institute of Science and Technology Austria \\
  Klosterneuburg, Austria \\
  \texttt{krishnendu.chatterjee@ista.ac.at} \\
  \And
  Alipasha Montaseri\\
  EPFL \\
  Lausanne, Switzerland \\
  \texttt{alipasha.montaseri@epfl.ch} \\
  \And
  Ali Shafiee\\
  Institute of Science and Technology Austria \\
  Klosterneuburg, Austria \\
  \texttt{ali.shafiee@ista.ac.at} \\
}

\begin{document}

\maketitle
{\let\thefootnote\relax\footnotetext{Authors are listed in alphabetical order.}}

\begin{abstract}
    A basic model in sequential decision making is the Markov decision process (MDP), which is extended to Robust MDPs (RMDPs) by allowing uncertainty in transition probabilities and optimizing against the worst-case transition probabilities from the uncertainty sets. 
    The class of $(s,a)$-rectangular RMDPs with $L_p$ uncertainty sets provides a flexible and expressive model for such problems. We study this class of RMDPs with a discounted-sum cost criterion and a constant discount factor. 
    The existence of an efficient algorithm for this class is a fundamental theoretical question in optimization and sequential decision making.
    Previous results only establish a strongly polynomial-time algorithm for $L_\infty$ uncertainty sets.
    In this work, our main results are as follows:
   (a)~we show that for any compact uncertainty set, the policy iteration algorithm for RMDPs is strongly polynomial with oracle access to solutions of Robust Markov chains (RMCs);
   (b)~we present strongly polynomial-time bounds on the policy iteration algorithm for RMCs with $L_1$ and $L_\infty$ uncertainty sets; and 
   (c)~we establish hardness results for RMCs with $L_p$ uncertainty sets for integer $p$ satisfying $1<p<\infty$.
   Finally, motivated by our theoretical bounds, we present experimental results showing how fast policy iteration converges for RMDPs with $L_1$ and $L_\infty$ uncertainty sets.
\end{abstract}

\section{Introduction}

\paragraph{Robust Markov Decision Processes.}
In a Markov decision process (MDP), an agent interacts with a stochastic environment of finite state space, where, given the current state and the choice of an action of the agent, the interaction with the environment updates the state
probabilistically~\cite{puterman94}.
This is the basic model in sequential decision making, where the common assumption is that the transition function is known.
In many practical scenarios, the transition function is not known, but rather  
constructed from data, which results in a level of uncertainty. 
This leads to the study of \emph{Robust Markov decision processes} (RMDPs)~\citep{nilim2005robust,iyengar2005robust}. 
In RMDPs, the assumption of MDPs is relaxed by considering that there is an
uncertainty set containing the original transition function, which is 
chosen by an adversary. 
We consider $(s,a)$-rectangular RMDPs where the uncertainty sets are separate for each state-action pair. This class has been the standard in the literature~\cite{iyengar2005robust,nilim2005robust}.

\paragraph{Discounted-sum costs and values.} The classical objective in MDPs and RMDPs for the agent is to minimize the discounted-sum cost criterion, where each state is assigned a cost and given a discount factor; the discounted cost of a run (infinite trajectory) is the discounted sum of the costs of states appearing in the run.
The value of an RMDP is the minimum expected discounted cost that can be guaranteed by the agent against all possible choices of the transition function from the uncertainty sets by the adversary. 

\paragraph{Exact value computation.}
While approximation is important, exact value computation is relevant
for two reasons.
First, it is a fundamental algorithmic problem.
For discounted MDPs with a fixed discount factor,
\cite{ye2011simplex} gives a strongly polynomial algorithm
for computing the exact value.
Second, value iteration easily gives additive approximation
guarantees for discounted MDPs and RMDPs.
Multiplicative approximation can be harder.
For example, in the stochastic games studied in~\cite{CMSS23},
a multiplicative approximation can be used to compute the exact value.
Exact computation provides a stronger guarantee than approximation.

\paragraph{$L_p$ uncertainty sets and motivation.} 
In this work, we focus on constant discount factors and primarily on $L_p$
uncertainty sets, for $p \in \Nb \cup \{\infty\}$. 
The main motivation to study $L_p$ uncertainty sets is as follows: for learning transition probabilities given the experimental samples, estimates of the uncertainty sets can be obtained using 
concentration bounds, which are then naturally approximated with $L_p$ uncertainty sets~\cite{weissman2003inequalities,DBLP:conf/aistats/BehzadianRPH21,DBLP:journals/ai/GivanLD00,iyengar2005robust}. Thus, the study of RMDPs with various $L_p$ uncertainty sets is a fundamental algorithmic problem that we address in this work.

\paragraph{Previous computational results.} Linear programming solves MDPs with discounted cost in polynomial time~\citep{d1963probabilistic,Derman1972FiniteStateMDP}. 
The existence of a strongly polynomial-time algorithm where the number of arithmetic operations is polynomial independent of the bit length is a fundamental theoretical question.
The seminal work of~\cite{ye2011simplex} presents the first strongly-polynomial time algorithm for constant discount factors for MDPs, and~\cite{hansen2013strategy} extends it to
stochastic games. Recently, a strongly-polynomial time bound was established for RMDPs with $L_\infty$ uncertainty sets by~\cite{asadi2026strongly}. While previous results exist for $L_\infty$ uncertainty sets, a comprehensive algorithmic study of RMDPs with $L_p$ uncertainty sets is missing, which is the focus of this work.

% \paragraph{Open problem.} Existing analyses of robust value and policy iteration for $(s, a)$-rectangular RMDPs with discounted-sum payoffs are developed for approximation schemes. Moreover, the best known time complexity for robust policy iteration is only polynomial-time (in the input encoding), and does not yield a strongly polynomial bound. Informally, an algorithm is strongly polynomial if its number of arithmetic operations is bounded by a polynomial in the problem dimensions, e.g., numbers of states, actions, and discount factor, and does not depend on the numerical magnitudes or bit-length of the input data, e.g., rewards and uncertainty sets. This leaves open whether robust policy iteration admits a strongly polynomial runtime guarantee for $\linf$ uncertainty sets.

\paragraph{Our contributions.} Our main contributions are as follows.
Robust Markov chains (RMCs) are a special class of RMDPs where the agent does not have the choice of actions. 
First, for any compact uncertainty sets, we bound the number of
policy iteration steps and calls to an RMC solver by a polynomial
in the numbers of states and actions, for a fixed discount factor
(\Cref{theorem:oracle_access}).
This gives a strongly polynomial algorithm for RMDPs whenever
the RMC subproblems admit a strongly polynomial algorithm. 
Second, for $L_1$ and $L_\infty$ uncertainty sets, we show that RMCs can be solved in strongly polynomial time. While this result exists for $L_\infty$
sets~\cite{asadi2026strongly}, our analysis improves the previous result by a factor of the number of states, and for $L_1$ sets, our result is novel (\Cref{thm:strongly_poly}).
Third, we establish hardness results for RMCs with $L_p$ uncertainty sets for $1<p<\infty$, and also present upper bounds for $L_p$ uncertainty sets (\Cref{thm:lp-comp-class}).
Finally, motivated by these theoretical bounds, we complement our analysis with experimental results showing how fast policy iteration converges for RMDPs with $L_1$ and $L_\infty$ uncertainty sets.

\paragraph{Technical contributions.} The main technical contributions are threefold. First, we present a
policy iteration algorithm for RMDPs that provides a reduction from RMDPs to RMCs. Second, we introduce a cumulative
transition-probability potential with respect to the optimal distribution chosen by the adversary, which eliminates suboptimal regions of the uncertainty sets instead of their exponentially many vertices. With this potential, we obtain a new strongly-polynomial bound for $L_1$ sets and an improved bound for $L_\infty$ sets. To the best of our knowledge, the previous potential function presented in~\cite{asadi2026strongly} does not extend to $L_1$ sets. Third, for other $L_p$ uncertainty sets
with integer $1 <p <\infty$, we prove hardness via a reduction from $p$-ROOT-SUM, based on a new equivalence between $p$-ROOT-SUM and $\frac{p}{p-1}$-ROOT-SUM, and give
an upper bound in the existential theory of the reals by encoding the adversary's optimization problem via KKT conditions.

\paragraph{Related works.} The exact value computation for RMDPs with various $L_p$ uncertainty sets is the focus of this work. Several related results have been established for the value approximation for more general models. For discounted-cost RMDPs, the value iteration algorithm is established concurrently in~\cite{nilim2005robust,iyengar2005robust}. Moreover, policy iteration is also introduced in~\cite{iyengar2005robust} for value approximation. Another variant of policy iteration was established in~\cite{DBLP:journals/informs/KaufmanS13} along with convergence guarantees. It is noteworthy that none of the above results present polynomial or strongly-polynomial bounds for the exact value computation.
RMDPs have also been studied with other objectives.
Robust control under temporal logic specifications was studied
in~\cite{wolff2012robust}.
For long-run average reward, undiscounted total reward, and stochastic
shortest path objectives, \cite{meggendorfer2025solving} gives algorithms
with precision guarantees.
For reachability and parity objectives, \cite{asadi2026qualitative}
studies whether the agent can ensure the objective with probability one.
For parity objectives with linearly defined uncertainty sets,
\cite{asadi2026quantitative} studies exact value computation.
Our work focuses on exact discounted-sum values.

\section{Preliminaries}

\paragraph{Notation.} Let $S$ be a set. We denote the set of all probability measures over $S$ by $\Delta(S)$, and its power set by $2^S$.
We use $[n]$ to represent the set of all natural numbers from 1 to $n$. The sets of natural numbers and non-negative integer numbers are denoted by $\Nb$ and $\Nb_0$, respectively. For vectors $\uvector$ and $\valvector$, we write $\valvector \succcurlyeq \uvector$ to indicate that every component of $\valvector$ is greater than or equal to the corresponding component of $\uvector$. Moreover, we use $\text{size}(n) = \lceil \log_2(|n|+1) \rceil + 1$, and $\size(\frac{p}{q})=\size(p)+\size(q)$ as the bit-size of integer $n$ and rational number $\frac{p}{q}$.
For a set $X$, we write $X^*$ for the set of all finite sequences
over $X$, including the empty sequence, and $X^+$ for the set
of all non-empty finite sequences.

\paragraph{Robust Markov Decision Processes.} A Robust Markov Decision Process (RMDP) is a tuple $\rmdp = (\states, \cost, \actions, \succ, \uncert)$, where the state space $\states = \{1, \dots, |\states|\}$ is a finite set of size $n$. The action space $\actions$ is a finite set of size $m$. The function $\cost \colon \states \to \R$ gives the cost of visiting each state. 
We use state-dependent costs to simplify notation.
Our results also extend to state-action costs $\cost(s,a)$.
The function $\succ \colon \states \times \actions \to 2^\states$ maps a state and an action to possible next states. The function $\uncert \colon \states \times \actions \to 2^{\Delta(\states)}$ gives a set of possible transition probabilities, where $\uncert(s,a) \subseteq \Delta(\succ(s,a))$. We assume that $\uncert(s, a)$ is a compact non-empty set.

\begin{remark}
    We assume $(s,a)$-rectangular uncertainty, which means that the adversary observes the chosen action before selecting a transition function from the uncertainty set. This assumption is standard in the literature~\cite{iyengar2005robust,nilim2005robust}. Other classes of uncertainty are not the focus of our work.
\end{remark}

\paragraph{Dynamics.} At first, the initial state is $s_0 \in \states$. At each time step $t \ge 0$, the agent first chooses an action $a_t \in \actions$. Then, the adversary selects a transition probability distribution $\pvector_t \in \uncert(s_t, a_t)$. The subsequent state $s_{t+1}$ is drawn probabilistically from the distribution $\pvector_t$. A run is an infinite sequence of states and actions $\omega \defas (s_0, a_0, s_1, \ldots)$ where for $t \ge 0$, there exists a distribution $\pvector_t \in \uncert(s_t, a_t)$ such that $\pvector_t(s_{t+1}) > 0$. A history is a finite sequence of states and actions $h \defas (s_0, a_0, \ldots, s_t)$.
%The set of all runs is denoted by $\Omega$, and the set of all runs where $s_0 = s$ is denoted by $\Omega_s$. 

\paragraph{Policies.} A history-dependent agent policy is denoted as $\agentpol \colon (\states \times \actions)^* \times \states \to \Delta(\actions)$. An agent policy is memoryless if it only depends on the current state, i.e., $\agentpol \colon \states \to \Delta(\actions)$ and is pure if it prescribes deterministic actions, i.e., $\agentpol \colon (\states \times \actions)^* \times \states \to \actions$. We say an agent policy is positional if it is pure and memoryless. A history-dependent adversary policy is denoted as $\envpol \colon (\states \times \actions)^+ \to \Delta(\Delta(\states))$ such that for every finite sequence of states and actions $h = (s_0, a_0, \ldots, s_t, a_t)$, we have $\envpol(h) \in \Delta(\uncert(s_t, a_t))$. The notions of memoryless, pure, and positional adversary policies are defined analogously.

\paragraph{Probability Measures.} For a history $h$, the cone of $h$ is the set of runs where $h$ is their prefix. Given an agent policy $\agentpol$, an adversary policy $\envpol$, and an initial state $s_0$, the
unique probability measure over Borel sets of runs is denoted by $\mathbb{P}_{s_0}^{\agentpol, \envpol}(.)$, which is defined by Carathéodory's extension theorem by extending the natural definition over cones of runs~\cite{billingsley2012ProbabilityMeasurea}. The expectation with respect to this probability measure is denoted by $\mathbb{E}_{s_0}^{\agentpol, \envpol}(.)$.

\paragraph{Discounted Cost Criterion.} A criterion 
is a measurable function that assigns a real number to every run. For a run $\omega = (s_0, a_0, \ldots)$, the discounted cost criterion is defined as $\discountedpayoff(\omega) \defas \sum_{t \in \Nb_0} \gamma^t \cost(s_t)$. The objective of the agent is to minimize the expected discounted cost criterion, while the adversary wants to maximize it.

\begin{remark}
    For the discounted cost criterion, positional policies for both agent and adversary are as powerful as history-dependent policies~\cite{iyengar2005robust, nilim2005robust}. Therefore, we restrict our attention to positional policies in the rest of the paper.
\end{remark}

\paragraph{Values.} The discounted value of a state $s$ is defined as
\begin{equation}
    \optval_\rmdp(s) \defas \inf_{\agentpol} \sup_{\envpol} \mathbb{E}_{s}^{\agentpol, \envpol} ( \discountedpayoff(\omega) )\,.
\end{equation}

We refer to $\optval_\rmdp$ as the optimal value vector.
The superscript $\star$ denotes optimality and should not be confused
with the Kleene star $X^*$.
We drop the $\rmdp$ label when it is clear from the context.

\paragraph{Optimal Policies.} An agent policy $\agentpol^{\star}$ is optimal if it guarantees the value for every state $s \in \states$, i.e., 
\[
    \sup_{\envpol} \mathbb{E}_{s}^{\agentpol^{\star}, \envpol} ( \discountedpayoff(\omega) ) = \optval_\rmdp(s)\,.
\]
The notion of optimal adversary policy is defined analogously.

\paragraph{Robust MCs.} A Robust Markov Chain (RMC) is a special case of an RMDP that has one possible action for the agent; hence, we can drop the action set. An RMC is $\rmc = (\states, \cost, \succ, \uncert)$. The value of RMC is denoted by $\optval_\rmc(s)$. We drop the $\rmc$ label when it is clear from the context.

% \paragraph{Markov Decision Processes.} A Markov Decision Process (MDP) is a special case of an RMDP where for every state $s$ and action $a$, the uncertainty set $\uncert(s, a)$ is a singleton. An MDP is $M = (\states, \cost, \actions, \trans)$, where $\trans \colon \states \times \actions \to \Delta(\states)$ gives the fixed transition probabilities. The value of an MDP is denoted by $\val_\mdp(s)$. We drop the $\mdp$ label when it is clear.

\paragraph{Markov Chains.} A Markov Chain (MC) is a special case of an RMC where for every state $s$, $\uncert(s)$ is a singleton. A MC is $\mc = (\states, \cost, \trans)$, where $\trans \colon \states \to \Delta(\states)$ are the fixed transition probabilities. The value of an MC is denoted by $\optval_\mc(s)$. We drop the $\mc$ label when it is clear from the context.

\paragraph{Problem Statement.} We consider the following computational problems.

\begin{tcolorbox}
    \textbf{Problem $\discrmdp$.} Given an RMDP $\rmdp$ with a constant discount factor $\discount \in [0, 1)$, compute the value vector $\optval_\rmdp \in \R^{|\states|}$, optimal agent policy $\agentpolopt \colon \states \to \actions$, and optimal adversary policy $\envpolopt \colon \states \times \actions \to \Delta(\states)$.
\end{tcolorbox}

\begin{tcolorbox}
   \textbf{Problem $\discrmc$.} Given an RMC $\rmc$ with a constant discount factor $\discount \in [0, 1)$, compute the value vector $\optval_\rmc \in \R^{|\states|}$ and optimal adversary policy $\envpolopt \colon \states \to \Delta(\states)$.
\end{tcolorbox}

The decision version of the above problems is, given an initial state $s$ and a threshold $\lambda$, decide whether $\optval(s) \geq \lambda$.

\section{Overview of Results}

In this section, we give an overview of our results.
First, we recall some complexity classes on which our results rely and then state our main theoretical contributions.

\subsection{Complexity Definitions}

\paragraph{Strongly Polynomial Algorithms.} An algorithm is strongly polynomial if its number of arithmetic
operations is polynomial in the input dimensions, independent of
the bit-length of the numerical data, and the bit-length of its
intermediate values is polynomial in the input's bit-length. Recall that an algorithm is said to run in polynomial time if its overall execution time is bounded by a polynomial in the input's bit-length. Note that strongly-polynomial time implies polynomial time, but not vice versa. Thus, strongly-polynomial time gives stronger efficiency guarantees.
%Consequently, the runtime of a standard polynomial-time algorithm can still depend on the bit-size of the input coefficients. A strongly polynomial algorithm differs in that its total number of arithmetic operations is completely independent of the numerical values of these coefficients, and any remaining dependence on the input coefficients is strictly limited to their inherent contribution to the encoding length. 

\paragraph{$\mathbf{p}$-ROOT-SUM Complexity.} The $p$-ROOT-SUM problem asks whether, given a sequence of positive integers $a_i$ (for $i \in [n]$) and an integer threshold $\Threshold$, the following inequality holds:
$$\sum_{i=1}^n a_i^{\frac{1}{p}} \geq \Threshold$$
For $p=2$, this corresponds to the classic Square-Root-Sum problem, which is a major and long-standing open problem in computational geometry since 1976~\cite{DBLP:conf/stoc/GareyGJ76}.  This problem is known to be in $\PSpace$, but neither polynomial-time algorithms are known nor even NP upper bounds are established~\cite{allender2009complexity}. We define the $p$-ROOT-SUM complexity class to contain all computational problems that are polynomial-time reducible to the $p$-ROOT-SUM problem.

% \paragraph{Existential Theory of the Reals ($\exists\Rb$).} A formula of the \emph{existential theory of the reals}, or ETR for short, is a formula of the form$$\exists x_1 \dots \exists x_k~\phi(x_1, \dots, x_k)$$where $\phi$ is a quantifier-free formula written with the symbols $0$, $1$, $=$, $\leq$, $<$, $+$, $-$, $\times$, $\wedge$, $\Leftrightarrow$, $\neg$, and parentheses, adhering to standard syntactic rules and semantics.

\paragraph{Existential theory of the reals.} 
A sentence in the \emph{existential theory of the reals} (ETR) is of the form $$\exists x_1 \dots \exists x_k~\phi(x_1, \dots, x_k)$$ where $\phi(x_1,\dots,x_k)$ is a \emph{quantifier-free} formula in the language of ordered fields over the reals $\Rb$. More concretely, atomic sub-formulas are $f(x_1,\dots,x_k)=0$ or $f(x_1,\dots,x_k)<0$, where $f \in \Qb[x_1,\dots,x_k]$ is a polynomial over $\Qb$, and arbitrary formulas are built by Boolean connectives, e.g., conjunction, disjunction, and negation.  We denote by $\exists\Rb$ the complexity class of problems that can be reduced in polynomial time to the problem of deciding the validity of an ETR sentence. This class naturally contains $\NP$ and is known to be contained within $\PSpace$ \cite{Can88}.

\subsection{Main Results}
In this section, we state our main results. 
The main algorithm to solve RMDPs and RMCs is the policy iteration algorithm (\Cref{algo:RMDPPI} and \Cref{algo:RMCPI}):
In RMDPs, this algorithm iterates over positional strategies of the agent,
and in RMCs, this algorithm iterates over positional strategies of the adversary.
The key goal is to obtain bounds on the number of iterations. 
Our results are as follows:

\begin{algorithm}[h]
\caption{Policy Iteration For RMDPs (\texttt{RMDP-PI})}  \label{algo:RMDPPI}
\begin{algorithmic}[1]
    \State \textbf{Input:} RMDP $\rmdp = (\states,\cost,\actions,\succ,\uncert)$, discount factor $\discount$
    \State \textbf{Output:} The optimal policy $\agentpol^{\star}$
    
    \State $\agentpol^0 \gets \text{arbitrary agent policy}$
    \State $t \gets 0$
    
    \Repeat
        \State $\envpol^{t} \gets \texttt{RMC-Oracle}(\rmdp^{\agentpol^t}, \discount)$ \Comment{Policy Evaluation: Invoking the RMC Oracle}
        
        \State $\val^t \gets (\I - \discount \trans^{\agentpol^t, \envpol^t})^{-1} \cost$
        
        \State $t \gets t + 1$

        \ForAll{$s \in \states$} \Comment{Policy Improvement}
            \State $\agentpol^{t}(s) \gets \argmin\limits_{a \in \actions} \max\limits_{\pvector \in \uncert(s,a)} \{ \cost(s) + \discount \pvector^\top \val^{t-1} \}$
        \EndFor
    \Until{$\agentpol^t = \agentpol^{t-1}$}
    
    \State \Return $\agentpol^t$
\end{algorithmic}
\smallskip
{\footnotesize
\textit{Tie-breaking Rule:} Choose the minimizing action with the
smallest index.\par}
\end{algorithm}

\begin{algorithm}[ht] 
\caption{Policy Iteration For RMCs (\texttt{RMC-PI})} \label{algo:RMCPI}
\begin{algorithmic}[1]
    \State \textbf{Input:} RMC $\rmc = (\states,\cost,\succ,\uncert)$, discount factor $\discount$
    \State \textbf{Output:} The optimal adversary policy $\envpol^{\star}$
    
    \State $\envpol^0 \gets \text{arbitrary adversary policy}$
    \State $t \gets 0$
    
    \Repeat
        \State $\val^t \gets (\identitymatrix - \discount \trans^{\envpol^t})^{-1} \cost$ \Comment{Policy Evaluation}
        
        \State $t \gets t + 1$
        
        \ForAll{$s \in \states$} \Comment{Policy Improvement}
            \State $\envpol^{t}(s) \gets \argmax\limits_{\pvector \in \uncert(s)} \{ \cost(s) + \discount \pvector^\top \val^{t - 1} \}$
        \EndFor
    \Until{$\envpol^t = \envpol^{t-1}$}
    
    \State \Return $\envpol^t$
\end{algorithmic}
\smallskip
{\footnotesize
\textit{Tie-breaking Rule:} For $L_1$ and $L_\infty$ sets, use
Algorithms~\ref{alg:lone_max} and~\ref{alg:linf_max}, respectively,
with ties between equal-value states broken by increasing
state index.\par}
\end{algorithm}

\begin{theorem}\label{theorem:oracle_access}
    The policy iteration algorithm for Problem~$\discrmdp$,
    Algorithm~\ref{algo:RMDPPI}, requires
    $\mathcal{O}(nm \frac{\log(1-\discountfactor)}{\log\discountfactor})$
    iterations and solutions of Problem~$\discrmc$
    for any compact uncertainty sets.
    In particular, if Problem~$\discrmc$ admits a strongly
    polynomial-time algorithm, then Problem~$\discrmdp$ also
    admits a strongly polynomial-time algorithm.
\end{theorem}

Our following results for RMCs are related to $L_p$ uncertainty sets, which are defined as follows:
\begin{assumption}
    The uncertainty sets are defined as $L_p$ balls, for $p \in \Nb \cup \{\infty\}$, around a nominal transition function $\nominal_s \in \Delta(\succ(s))$. For RMCs, this means $\uncert(s) = \{\pvector \in \Delta(\succ(s)) ~|~ \lVert \pvector - \nominal_s \rVert_p \leq \radius(s) \}$ for a nominal transition $\nominal_s$ and radius $\radius(s)$. This definition extends naturally to RMDPs via $\nominal_{s,a}$ and $\radius(s, a)$ for each state-action pair. 
\end{assumption}

\begin{theorem}\label{thm:strongly_poly}
    The policy iteration algorithm for Problem~$\discrmc$, Algorithm~\ref{algo:RMCPI}, for $L_1$ and $L_{\infty}$ uncertainty sets requires $\mathcal{O}(n^3 \log n \log \left( \frac{1 - \discountfactor}{n} \right)/\log \discountfactor)$ (i.e., strongly polynomial for constant discount factors) iterations.
\end{theorem}

\begin{theorem}
\label{thm:lp-comp-class}
    For any integer $p > 1$, the following statements hold:
\begin{enumerate}
    \item \label{item:p-root-sum-hardness} The decision version of Problem~$\discrmc$ with $L_p$ uncertainty sets is $p$-ROOT-SUM-hard;
    \item \label{item:etr-easyness} The decision version of Problem~$\discrmc$ with $L_p$ uncertainty sets is in $\exists\R$.
\end{enumerate}
\end{theorem}

\begin{table}[t]
\centering
\small
\setlength{\tabcolsep}{12pt}
\begin{tabular}{@{}lll@{}}
\toprule
Uncertainty sets & PI iterations & Reference \\
\midrule
\multicolumn{3}{@{}l}{\textit{RMDPs, \Cref{algo:RMDPPI} with an RMC solver}} \\
\addlinespace[2pt]
\quad $L_\infty$ & $O(nm)$ & \cite{asadi2026strongly} \\
\quad \textbf{Compact} & $\boldsymbol{O(nm)}$
    & \textbf{\Cref{theorem:oracle_access}} \\
\midrule
\multicolumn{3}{@{}l}{\textit{RMCs, \Cref{algo:RMCPI}}} \\
\addlinespace[2pt]
\quad $L_\infty$ & $O(n^4\log^2 n)$ & \cite{asadi2026strongly} \\
\quad $\boldsymbol{L_\infty}$ & $\boldsymbol{O(n^3\log^2 n)}$
    & \textbf{\Cref{thm:strongly_poly}} \\
\quad $\boldsymbol{L_1}$ & $\boldsymbol{O(n^3\log^2 n)}$
    & \textbf{\Cref{thm:strongly_poly}} \\
\bottomrule
\end{tabular}
\vspace{6pt}
\caption{Comparison of policy iteration bounds for a fixed discount
factor $0<\discountfactor<1$.
Here, $n$ is the number of states and $m$ is the number of actions.
RMDPs use $(s,a)$-rectangular uncertainty sets.
The RMDP rows count outer iterations with access to an RMC solver.
The RMC rows count iterations of the robust policy-evaluation algorithm.
Thus, the RMDP bounds alone are not full running-time bounds.
Our results are shown in bold.}
\label{tab:pi_bounds}
\end{table}

\paragraph{Significance.} We now discuss the significance of our results, which present a comprehensive picture of RMDPs and RMCs with $L_p$ uncertainty sets and a constant discount factor. 
Our first result bounds the number of policy iteration steps
and calls to an RMC solver for arbitrary compact uncertainty sets.
For a fixed discount factor, these bounds are polynomial
in the numbers of states and actions.
The second result presents strongly polynomial time bounds for RMCs with $L_1$ and $L_\infty$ uncertainty sets.
Although the strongly-polynomial-time bounds for the $L_\infty$ case were previously established by~\cite{asadi2026strongly}, our analysis yields a tighter bound by a factor of $n$; and the result for $L_1$ uncertainty sets is novel.
Table~\ref{tab:pi_bounds} compares our policy iteration bounds
with the previous bounds for $L_\infty$ uncertainty sets.
Finally, the first item of the third result shows computational hardness for $L_p$ uncertainty sets with $1<p<\infty$ (i.e., other than $L_1$ and $L_\infty$).
The last item of the third result presents upper bounds for RMCs with $L_p$ uncertainty sets.
Together, these results complete the study of efficient algorithms for RMDPs and RMCs with a constant discount factor with $L_p$ uncertainty sets.
In the following sections, we provide the main proof ideas of all these results.

\section{Overview of \texorpdfstring{\cref{theorem:oracle_access}}{Theorem 1}: Strongly Polynomial PI for RMDPs via RMC Oracles}
In this section, we provide a brief overview of the proof of \cref{theorem:oracle_access}. The theorem states that for any compact uncertainty set and a constant discount factor, the policy iteration algorithm for Problem~$\discrmdp$ (Algorithm~\ref{algo:RMDPPI}) terminates after $\Ocomplexity\left(nm \frac{\log(1-\discountfactor)}{\log \discountfactor}\right)$ iterations given oracle access to the solution of Problem~$\discrmc$. The full proof is presented in Appendix~\ref{app:thm1}. The key insight is that, after fixing an agent policy $\agentpol$, we obtain an RMC, and consequently, the policy-evaluation step can be handled by an RMC oracle call. Therefore, our main task is to show that the number of iterations is polynomial. We organize the proof into four main steps. 
We follow the action-elimination approach of~\cite{ye2011simplex}.
The main difference is the potential function.
For RMDPs, the potential function must account for the adversary's choice of a transition
distribution for each action.
Our potential function extends this approach to compact uncertainty sets.

\paragraph{Robust Bellman Optimality Operator.} We first define the robust Bellman optimality operator as
\[
    (\Bellman \valvector)_s \defas \cost(s) + \discountfactor \min_{a \in \actions} \max_{\pvector \in \uncert(s,a)} \pvector^\top \valvector .
\]
Since $\actions$ is finite and each uncertainty set $\uncert(s,a)$ is compact, the minimum and maximum are attained. The operator $\Bellman$ is a $\discountfactor$-contraction in the infinity norm and therefore has a unique fixed point $\valopt$, which is the optimal value vector. Moreover, for every fixed agent policy $\agentpol$, the induced value vector $\valvector^\agentpol$ satisfies $\Bellman \valvector^\agentpol \preccurlyeq \valvector^\agentpol$, and consequently $\valopt \preccurlyeq \valvector^\agentpol$. Using this contraction property, we show that \texttt{RMDP-PI} generates a monotonically non-increasing sequence of value vectors, i.e., $\valvector^{t+1} \preccurlyeq \valvector^t$. Moreover, the value error decreases exponentially, i.e., 
\[
    \infnorm{\valvector^t - \valopt}
    \leq
    \discountfactor^t \infnorm{\valvector^0 - \valopt},
\]
which gives a convergence rate that is later used to bound the number of iterations.

\paragraph{Potential Function.} We next define a potential function that quantifies the cost of choosing an action $a$ instead of an optimal action at state $s$, where both choices are evaluated against the optimal value vector. Formally,
\[
    f(s,a)
    \defas
    \max_{\pvector \in \uncert(s,a)} \pvector^\top \valopt
    -
    \max_{\pvector \in \uncert(s,\agentpolopt(s))} \pvector^\top \valopt .
\]
Intuitively, $f(s,a)$ measures the continuation cost incurred by choosing $a$ at state $s$ instead of $\agentpolopt(s)$.

For any agent policy $\agentpol$, the potential gives two complementary bounds. First, if $\agentpol(s)=a$, then the value error at state $s$ satisfies
\[
    \valvector^\agentpol_s - \valopt_s \geq \discountfactor f(s,a).
\]
Second, if $\hat{s} \in \argmax_{s \in \states} f(s,\agentpol(s))$, then the value error is bounded by
\[
    \infnorm{\valvector^\agentpol - \valopt}
    \leq
    \frac{\discountfactor}{1-\discountfactor} f(\hat{s},\agentpol(\hat{s})) .
\]
These inequalities establish lower and upper bounds on the difference between an arbitrary policy's value vector and the optimal value vector, and are later used to prove that Algorithm~\ref{algo:RMDPPI} terminates in a polynomial number of iterations.

\paragraph{Elimination of Actions with Maximum Potential.} We combine the lower and upper bounds above with the exponential convergence rate. Let $\agentpol^l$ be a non-optimal policy at iteration $l$, and let $\hat{s} \in \argmax_{s \in \states} f(s,\agentpol^l(s))$. If a later policy $\agentpol^k$ selects the same action at $\hat{s}$ after at least
$L \defas \left\lfloor \log_{\discountfactor}(1-\discountfactor)
\right\rfloor + 1$ iterations, then the lower bound on the potential function contradicts the exponential decay of the value error. Hence, the state-action pair $(\hat{s},\agentpol^l(\hat{s}))$ is permanently removed after $L$ iterations.
There are at most $nm$ state-action pairs. 
Since at least one state-action pair is eliminated every $L$
iterations until an optimal policy is reached, the fixed tie-breaking
rule ensures termination within $\Ocomplexity(nmL)$ iterations.

\paragraph{Reduction to One RMC Oracle Call.} Finally, we show that the policy-improvement step can be implemented using an RMC oracle. For every state-action pair $(s,a)$, the algorithm needs to compute
\[
    \cost(s) + \discountfactor \max_{\pvector \in \uncert(s,a)} \pvector^\top \valvector^{t-1} .
\]
Instead of solving these maximization problems separately, we construct a batch RMC $\rmc_{\mathrm{batch}}$. This RMC has one transient state $z_{s,a}$ for every pair $(s,a) \in \states \times \actions$ and one absorbing state $x_s$ for every original state $s \in \states$. The absorbing state $x_s$ is assigned value $\valvector^{t-1}_s$, and the uncertainty set at $z_{s,a}$ is exactly $\uncert(s,a)$. Therefore, we have $\valvector^{\star}_{\mathrm{batch}}(z_{s,a}) = \cost(s) + \discountfactor \max_{\pvector \in \uncert(s,a)} \pvector^\top \valvector^{t-1}$. 
Thus, all action-improvement quantities can be computed with one
RMC oracle call on an RMC with $\Ocomplexity(nm)$ states.
The construction reuses the original uncertainty sets, with
transitions redirected to the corresponding absorbing states.
Together with the policy-evaluation call, each iteration uses
two RMC oracle calls.
Our bound counts arithmetic operations and oracle calls,
assuming exact arithmetic.
Intermediate values, including the costs of the batch RMC,
may be irrational.
Therefore, this bound alone does not guarantee polynomial
running time in the standard bit model.
However, if the RMC oracle is a strongly polynomial-time
algorithm, then Algorithm~\ref{algo:RMDPPI} is a
strongly polynomial-time algorithm for Problem~$\discrmdp$.

\section{Overview of \texorpdfstring{\cref{thm:strongly_poly}}{Theorem 2}: Complexity of RMC with \texorpdfstring{$L_1, L_\infty$}{L1, L-infinity} Uncertainty Sets}

In this section, we provide a brief overview of the proof of \cref{thm:strongly_poly}. The theorem states that for $\lone$ and $\linf$ uncertainty sets and a constant discount factor, the policy iteration algorithm for Problem~$\discrmc$ (Algorithm~\ref{algo:RMCPI}) terminates in a polynomial number of iterations. Since each iteration runs in strongly polynomial time, this yields a strongly polynomial algorithm under $\lone$ and $\linf$ uncertainty sets with a constant discount factor. Combined with \cref{theorem:oracle_access}, this completes the proof that RMDPs with $\lone$ or $\linf$ uncertainty sets and a constant discount factor are solvable in strongly polynomial time. The full proof is presented in Appendix~\ref{app:thm2}.

The challenge of analyzing policy iteration for RMCs, in contrast to standard MDPs, is that the adversary's action space is the continuous uncertainty set $\uncert(s)$. Although the worst-case distribution lies on the boundary of the $\lone$ or $\linf$ ball, the boundary can have exponentially many vertices (e.g. $L_1$ ball of dimension $n$ has up to $2^n$ vertices). Thus, a naive enumeration over boundary points fails. Our argument bypasses this by exploiting the algebraic structure of the distributions returned by the policy improvement step and combining it with a number-theoretic combinatorial bound. We organize the proof into four main steps.

\paragraph{Policy Improvement and Structural Property.} The policy improvement step of Algorithm~\ref{algo:RMCPI} requires solving $\max_{\pvector \in \uncert(s)} \pvector^\top \valvector$. For $\lone$, the homotopy method of~\cite{ho2021partial} sorts the successor states by descending value and greedily shifts the available budget from the lowest-value tail to the highest-value head. For $\linf$, the algorithm by~\cite{behzadian2021fast} performs an analogous transfer subject to the entry-wise constraint $|\pvector_{s'} - \nominal_{s,s'}| \leq \radius_s$. Both procedures run in linear time after sorting the values, and expose a clean structural property: every coordinate of the resulting distribution is either equal to the nominal $\nominal_{s,s'}$, shifted from it by a fixed value ($\pm\radius_s$ for $\linf$ and $\pm\radius_s/2$ for $\lone$), or saturated at $0$ or $1$, with at most one ``incomplete'' coordinate. This property is used in the convergence analysis.

\paragraph{Potential Function and Sandwich Bounds.} 
Fix a state $s$ and order its successors $s_1, \ldots, s_{|\succ(s)|}$
in descending order of $\valopt$, breaking ties by state index.
Let $F^{\star}(s,i)$ be the cumulative transition probability of the optimal
distribution returned by the corresponding greedy procedure
with this order.
For an adversary policy $\envpol$, define the cumulative transition $F^\envpol(s,i) \defas \sum_{j=1}^{i} \pvector^\envpol_{s,s_j}$. A first observation is that $F^{\star}(s,i) \geq F^\envpol(s,i)$ for every feasible $\envpol$ and every pair of state and index $(s,i)$. We then define the potential function $f_\envpol(s,i)$ which is guaranteed to be non-negative:
\[
    f_\envpol(s,i) \;\defas\; \bigl(F^{\star}(s,i) - F^\envpol(s,i)\bigr)\bigl(\valopt_{s_i} - \valopt_{s_{i+1}}\bigr).
\]
Using the Bellman equations and summation by parts, we obtain a \emph{lower bound} on the per-state value gap, $\valopt_s - \valvector^\envpol_s \geq \discountfactor f_\envpol(s,i)$. A complementary \emph{upper bound} is shown to bound the global value error by the worst-case potential, namely
\[
    \infnorm{\valopt - \valvector^\envpol} \;\leq\; \frac{\discountfactor n}{1-\discountfactor}\, \hat{f}_\envpol, \qquad \text{where } \hat{f}_\envpol \defas \max_{s \in \states,\, i} f_\envpol(s,i).
\]

These inequalities sandwich the global value error by a single scalar quantity associated with the worst state-index pair.

\paragraph{Halving of the Cumulative Gap.} Combining the sandwich bounds with the standard exponential convergence $\infnorm{\valopt - \valvector^t} \leq \discountfactor^{t} \infnorm{\valopt - \valvector^0}$ yields the key halving lemma: if $(s,i)$ maximizes the potential at iteration $t$, then within $L \defas \lceil \log_{\discountfactor}\!\bigl(\tfrac{1-\discountfactor}{2n}\bigr) \rceil$ subsequent iterations the gap $F^{\star}(s,i) - F^\envpol(s,i)$ must shrink by at least a factor of two; otherwise, the lower bound at iteration $l$ and the upper bound at iteration $t$ would contradict the exponential decay rate of policy iteration.

\paragraph{Combinatorial Bound on the Number of Halvings.} The structural property of policy iteration (the first step) implies that for every state $s$ and every policy $\envpol$ visited by the algorithm, the cumulative gap $F^{\star}(s,i) - F^\envpol(s,i)$ is a signed sum of $\Ocomplexity(n)$ numbers (the nominal probabilities $\nominal_{s,j}$ together with $1$ and multiples of $\radius_s$), with constant integer coefficients. A combinatorial lemma~\cite{asadi2026strongly} shows that any such set of signed sums takes $\Ocomplexity(n \log n)$ distinct binary scales $\lfloor \log_2 |y| \rfloor$. Since every halving drops the gap to a smaller binary scale, each of the at most $n^2$ state-index pairs $(s,i)$ can have at most $\Ocomplexity(n \log n)$ halvings. Combined with the $L$-iteration cost per halving from the above step, this bounds the total number of policy iterations by $\Ocomplexity(n^3 \log n \cdot L)$, which gives the polynomial bound of \cref{thm:strongly_poly}. The same argument improves the previously known bound of~\cite{asadi2026strongly} for the $\linf$ case by a factor of $|\States|$, and gives the first polynomial bound in the $\lone$ case.

\section{Overview of \texorpdfstring{\cref{thm:lp-comp-class}}{Theorem 3}: Complexity of RMC with \texorpdfstring{$L_p$}{Lp} Uncertainty Sets}

In this section, we provide a brief overview of the proof for \Cref{thm:lp-comp-class}. This theorem states that for any integer $p > 1$, solving the decision version of the problem $\discrmc$ with $L_p$ uncertainty sets is $p$-ROOT-SUM-hard, and the problem belongs to $\exists\mathbb{R}$ complexity class. 

We break down the proof into three steps: (1)~Establishing the equivalence between the complexity classes $p$-ROOT-SUM and $\frac{p}{p-1}$-ROOT-SUM, (2)~Proving the problem's hardness via a reduction from the $\frac{p}{p-1}$-ROOT-SUM problem to the decision problem, and (3)~Showing the problem belongs to the $\exists \mathbb{R}$ by encoding it as an $\etr$ formula. The full proof of \cref{thm:lp-comp-class} is detailed in Appendix~\ref{app:thm3}.

\paragraph{1. Equivalence of ROOT-SUM Problems.} Initially, we simplify the underlying complexity classes. We prove that the $p$-ROOT-SUM problem is polynomial-time equivalent to the $\frac{p}{p-1}$-ROOT-SUM problem. 
We achieve this through a sequence of reductions. First, we show that the $p$-ROOT-SUM problem is reducible to the $kp$-ROOT-SUM problem for any constant integer $k > 1$. Next, to handle fractional roots, we introduce a ``Greedy $p$-th Power Decomposition'' algorithm, which breaks down any positive integer $n$ into a sum of $p$-th powers using $\mathcal{O}(\log \log n)$ terms in polynomial time.
Using this decomposition, we prove that the $p$-ROOT-SUM problem can be reduced to the $\frac{p}{q}$-ROOT-SUM problem for any integer $q$ where $\gcd(p,q)=1$. Combining these steps shows that the denominator in the fractional root does not change the complexity class, establishing the polynomial-time equivalence between the $p$-ROOT-SUM and $\frac{p}{p-1}$-ROOT-SUM problems.

\paragraph{2. Proving Hardness via Reduction.}
We prove that the decision version of the problem $\discrmc$ is at least as hard as solving the $\frac{p}{p-1}$-ROOT-SUM problem. We do this by constructing a three-layered RMC that forces the RMC's Bellman equations to mimic the ROOT-SUM inequality.
In an RMC, the adversary chooses transition probabilities from an $L_p$ uncertainty set to maximize the expected cost. Formally, this corresponds to finding the maximum inner product between the transition probability vector and the value vector of successor states. According to Hölder's inequality \cite{folland1999real}, maximizing an inner product under an $L_p$ constraint evaluates the dual norm, which is the $L_{\frac{p}{p-1}}$ norm. This dual norm computes the fractional powers for the $\frac{p}{p-1}$-ROOT-SUM problem. By this property, we construct the three-layered RMC as follows:
\begin{itemize}
    \item \textit{Absorbing States:} At the base layer, we create absorbing states with fixed costs, which are equal to the decomposition of the input integers $a_i$ by the ``Greedy $p$-th Power Decomposition'' algorithm. 
    \item \textit{Transient States:} The middle layer has transient states for each $a_i$, which are connected to corresponding absorbing states. Because the adversary picks worst-case transitions from an $L_p$ ball, the value of each transient state is proportional to the $L_{\frac{p}{p-1}}$ norm of the corresponding absorbing state values, which is proportional to $a_i^{\frac{p-1}{p}}$.
    \item \textit{Initial State:} The top layer is a single initial state that transitions uniformly to all transient states. Its value is the average of the transient states' values, which is $\sum_{i=1}^n a_i^{\frac{p-1}{p}}$.
\end{itemize}

By structuring the RMC this way, the adversary's optimal choices compute the fractional roots for us. Finally, deciding whether the value of the initial state meets the threshold $\lambda$ is equal to deciding if the sum meets the threshold $\Threshold$. Thus, the decision problem is $p$-ROOT-SUM-hard.

\paragraph{3. ETR Upper Bound.}
Finally, we prove that the decision version of problem $\discrmc$ is in $\exists\mathbb{R}$. We show this by modeling the main problem (computing exact values and optimal strategy) as a convex optimization problem. Because this optimization problem satisfies Slater's condition, we characterize its optimal solution using the Karush-Kuhn-Tucker (KKT) conditions \cite{boyd2004convex}.
Because $p$ is a constant integer, the derivative of these constraints produces integer powers. This ensures that the optimization problem can be expressed by polynomial equations and inequalities.
To complete the proof, we combine these polynomial KKT conditions with the Bellman equations to create a system of polynomial equations and inequalities. By definition, finding a valid solution to a system of polynomial equations and inequalities over the real numbers is in $\exists\mathbb{R}$. 

\section{Experiments}

\label{sec:experiments}

In this section, we empirically evaluate \texttt{RMDP-PI} (\Cref{algo:RMDPPI}) across five benchmarks under $L_1$ and $L_\infty$ uncertainty sets, motivated by the strongly polynomial theoretical bounds we established for these specific cases. Our primary goal is to observe how many iterations \Cref{algo:RMDPPI} takes to converge on these benchmarks. Additionally, we analyze the total number of iterations \texttt{RMC-PI} (\Cref{algo:RMCPI}) performed over each iteration of \Cref{algo:RMDPPI}.

\paragraph{Implementation Details.}
All experiments are implemented in Python 3.11.14 using exact calculations. Specifically, every transition probability, reward, discount factor, and uncertainty radius is stored and processed using Python's built-in \texttt{fractions.Fraction} class. The linear system that arises during policy evaluation is solved exactly via a Bareiss fraction-free Gaussian elimination procedure implemented directly on top of arbitrary-precision rationals from the \texttt{gmpy2} library. The inner $L_1$ and $L_\infty$ oracles (\Cref{alg:lone_max} and \Cref{alg:linf_max}) are procedures that sort values and shift probability masses, involving only rational subtractions and comparisons. All experiments were executed on a single machine with an Intel Core Ultra 5 225U processor and 16GB of RAM.

\begin{figure}[h]
  \centering
  \includegraphics[width=\linewidth]{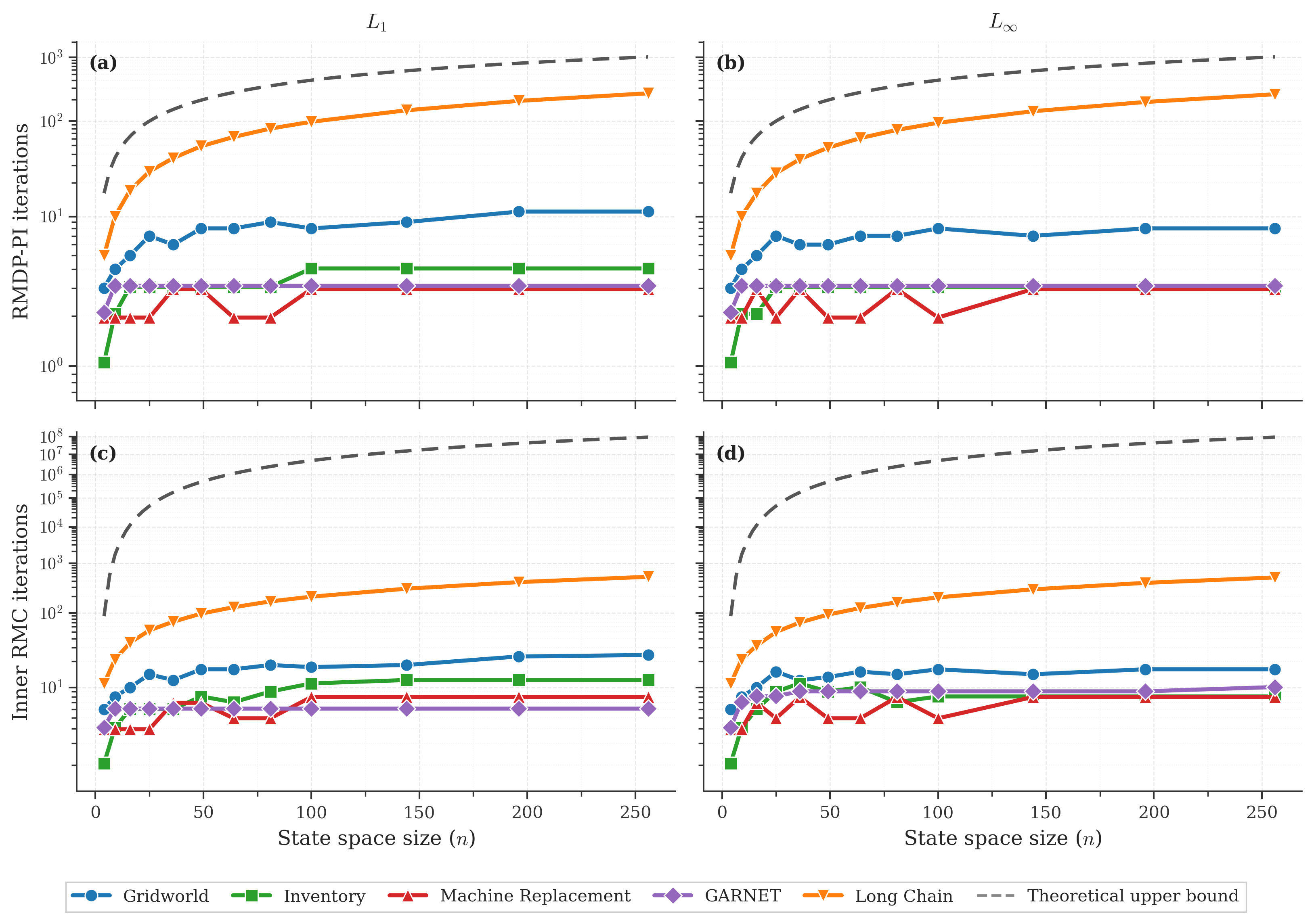}
  \caption{%
     Outer \texttt{RMDP-PI} (a,~b) and total inner \texttt{RMC-PI} (c,~d)
    iterations versus state-space size~$n$ on five benchmarks, under $L_1$
    (left) and $L_\infty$ (right) uncertainty sets, with $\gamma=0.5$,
    $\delta=0.05$. Dashed line: theoretical upper bound.%
  }
  \label{fig:main_results_g05}
\end{figure}

\paragraph{Benchmarks.}
We evaluate \texttt{RMDP-PI} (\Cref{algo:RMDPPI}) on five environments, scaling the state-space size up to $n = 256$. Four are standard: \textbf{Gridworld}~\cite{andrew2018reinforcement}, a stochastic grid-navigation task with an absorbing goal and a trap; \textbf{Inventory Management}~\cite{puterman94}, where the agent chooses order quantities against stochastic demand; \textbf{Machine Replacement}~\cite{wiesemann2013}, in which a degrading machine can be operated, repaired, or replaced; and \textbf{GARNET}~\cite{archibald1995generation}, an unstructured random-MDP baseline. The fifth, \textbf{Long Chain}, is a construction we introduce that forces $\Theta(n)$ outer policy-iteration steps. The complete dynamics for each environment are given in Appendix~\ref{appendix:experiments}.

\paragraph{Results.}
\Cref{fig:main_results_g05} reports the outer \texttt{RMDP-PI} and total inner \texttt{RMC-PI} iteration counts at $\gamma=0.5$, $\delta=0.05$, under both $L_1$ and $L_\infty$ uncertainty sets; similar plots for different values of $\gamma,\delta$ are presented in Appendix~\ref{appendix:experiments} (\Cref{fig:main_results_g9,fig:main_results_g99,fig:main_results_g995}). We show $\gamma=0.5$ in the main text for space reasons.
Across the tested discount factors, we observe similar trends
in the iteration counts, while the theoretical bounds increase
as $\gamma$ approaches $1$.

Across both norms, iteration counts on the four standard benchmarks (Gridworld, Inventory, Machine Replacement, GARNET) grow very slowly with $n$ and stay in the low double digits even at $n=256$, while Long Chain shows the $\Theta(n)$ growth it was designed to produce. All curves sit several orders of magnitude below the worst-case upper bound (dashed grey), and the gap only widens as $\gamma \to 1$.

% \begin{figure}[t]
%   \centering
%   \begin{subfigure}[t]{0.95\linewidth}
%     \includegraphics[width=\linewidth]{figures/g5d05.png}
%     \caption{$\gamma=0.5,\ \delta=0.05$}
%     \label{fig:main_results_a}
%   \end{subfigure}

%   \vspace{0.4em}

%   \begin{subfigure}[t]{0.95\linewidth}
%     \includegraphics[width=\linewidth]{figures/g9d05.png}
%     \caption{$\gamma=0.9,\ \delta=0.05$}
%     \label{fig:main_results_b}
%   \end{subfigure}

%   \caption{Outer \texttt{RMDP-PI} iterations (top row of each sub-figure) and total inner \texttt{RMC-PI} iterations (bottom row) versus state-space size $n$, on the four benchmarks under $L_1$ (left column) and $L_\infty$ (right column).}
%   \label{fig:main_results_low}
% \end{figure}

% \begin{figure}[t]
%   \centering
%   \begin{subfigure}[t]{0.95\linewidth}
%     \includegraphics[width=\linewidth]{figures/g99d01.png}
%     \caption{$\gamma=0.99,\ \delta=0.01$}
%     \label{fig:main_results_c}
%   \end{subfigure}

%   \vspace{0.4em}

%   \begin{subfigure}[t]{0.95\linewidth}
%     \includegraphics[width=\linewidth]{figures/g995d05.png}
%     \caption{$\gamma=0.995,\ \delta=0.05$}
%     \label{fig:main_results_d}
%   \end{subfigure}

%   \caption{Outer \texttt{RMDP-PI} iterations (top row of each sub-figure) and total inner \texttt{RMC-PI} iterations (bottom row) versus state-space size $n$, on the four benchmarks under $L_1$ (left column) and $L_\infty$ (right column).}
%   \label{fig:main_results_high}
% \end{figure}

\section{Conclusion and Limitations}
%\paragraph{Concluding Remarks.} 
Our work presents a comprehensive picture of algorithmic studies of RMDPs and RMCs with a constant discount factor. There are some limitations of this work, which naturally give rise to future directions. First, this work focuses on $L_p$ uncertainty sets, and generalization to other sets is an interesting direction. Second, for general $L_p$ uncertainty sets, there is a gap between the lower and upper bound ($p$-root-sum hardness and ETR upper bound), and a tighter complexity result is another interesting direction. Third, for $L_1$ and $L_\infty$ uncertainty sets, our work gives strongly polynomial upper bounds for policy iteration. However, lower bounds have not yet been studied in the context of RMDPs.

\section*{Acknowledgments}
We thank Negar Babashah for helpful discussions. 
This research was partially supported by Austrian Science Fund (FWF) 10.55776/COE12 and by the ERC CoG 863818 (ForM-SMArt) grant.

\bibliography{bibliography}

@article{iyengar2005robust,
  title={Robust dynamic programming},
  author={Iyengar, Garud N},
  journal={Mathematics of Operations Research},
  volume={30},
  number={2},
  pages={257--280},
  year={2005},
  publisher={INFORMS}
}

@article{nilim2005robust,
  title={Robust control of Markov decision processes with uncertain transition matrices},
  author={Nilim, Arnab and El Ghaoui, Laurent},
  journal={Operations Research},
  volume={53},
  number={5},
  pages={780--798},
  year={2005},
  publisher={INFORMS}
}

@article{ho2021partial,
  title={Partial policy iteration for l1-robust markov decision processes},
  author={Ho, Chin Pang and Petrik, Marek and Wiesemann, Wolfram},
  journal={Journal of Machine Learning Research},
  volume={22},
  number={275},
  pages={1--46},
  year={2021}
}

@article{allender2009complexity,
  title={On the complexity of numerical analysis},
  author={Allender, Eric and B{\"u}rgisser, Peter and Kjeldgaard-Pedersen, Johan and Miltersen, Peter Bro},
  journal={SIAM Journal on Computing},
  volume={38},
  number={5},
  pages={1987--2006},
  year={2009},
  publisher={SIAM}
}

@inproceedings{Can88,
author = {Canny, John},
title = {Some algebraic and geometric computations in {PSPACE}},
year = {1988},
publisher = {Association for Computing Machinery},
doi = {10.1145/62212.62257},
booktitle = {Proceedings of the Twentieth Annual ACM Symposium on Theory of Computing {STOC}},
pages = {460–467},
numpages = {8},
series = {STOC '88}
}

@book{boyd2004convex,
  title={Convex optimization},
  author={Boyd, Stephen and Vandenberghe, Lieven},
  year={2004},
  publisher={Cambridge university press}
}

@article{asadi2026strongly,
  title={Strongly polynomial time complexity of policy iteration for $L_\infty$ robust MDPs},
  author={Asadi, Ali and Chatterjee, Krishnendu and Goharshady, Ehsan and Karrabi, Mehrdad and Montaseri, Alipasha and Pagano, Carlo},
  journal={arXiv preprint arXiv:2601.23229},
  year={2026}
}

@article{behzadian2021fast,
  title={Fast Algorithms for $ L_\infty$-constrained S-rectangular Robust MDPs},
  author={Behzadian, Bahram and Petrik, Marek and Ho, Chin Pang},
  journal={Advances in Neural Information Processing Systems},
  volume={34},
  pages={25982--25992},
  year={2021}
}

@article{ye2011simplex,
  title={The simplex and policy-iteration methods are strongly polynomial for the Markov decision problem with a fixed discount rate},
  author={Ye, Yinyu},
  journal={Mathematics of Operations Research},
  volume={36},
  number={4},
  pages={593--603},
  year={2011},
  publisher={INFORMS}
}

@book{andrew2018reinforcement,
  title={Reinforcement Learning: An Introduction},
  author={Sutton, Richard S. and Barto, Andrew G.},
  edition={Second},
  year={2018},
  publisher={The MIT Press}
}

@book{puterman94,
  author       = {Martin L. Puterman},
  title        = {Markov Decision Processes: Discrete Stochastic Dynamic Programming},
  series       = {Wiley Series in Probability and Statistics},
  publisher    = {Wiley},
  year         = {1994}
}

@article{wiesemann2013,
  title={Robust Markov decision processes},
  author={Wiesemann, Wolfram and Kuhn, Daniel and Rustem, Ber{\c{c}}},
  journal={Mathematics of Operations Research},
  volume={38},
  number={1},
  pages={153--183},
  year={2013},
  publisher={INFORMS}
}

@article{archibald1995generation,
  title={On the generation of markov decision processes},
  author={Archibald, Thomas Welsh and McKinnon, Ken IM and Thomas, Lyn C},
  journal={Journal of the Operational Research Society},
  volume={46},
  number={3},
  pages={354--361},
  year={1995},
  publisher={Taylor \& Francis}
}

@book{billingsley2012ProbabilityMeasurea,
  title = {Probability and Measure},
  author = {Billingsley, Patrick},
  year = {2012},
  publisher = {Wiley}
}

@book{HindrySilverman,
  author = {Hindry, Marc and Silverman, Joseph H.},
  title = {Diophantine Geometry: An Introduction},
  publisher = {Springer},
  year = {2000}
}

@article{weissman2003inequalities,
  title={Inequalities for the l1 deviation of the empirical distribution},
  author={Weissman, Tsachy and Ordentlich, Erik and Seroussi, Gadiel and Verdu, Sergio and Weinberger, Marcelo J},
  journal={Hewlett-Packard Labs, Tech. Rep},
  pages={125},
  year={2003}
}

@inproceedings{DBLP:conf/aistats/BehzadianRPH21,
  author       = {Bahram Behzadian and
                  Reazul Hasan Russel and
                  Marek Petrik and
                  Chin Pang Ho},
  title        = {Optimizing Percentile Criterion using Robust MDPs},
  booktitle    = {{AISTATS}},
  series       = {Proceedings of Machine Learning Research},
  pages        = {1009--1017},
  publisher    = {{PMLR}},
  year         = {2021}
}

@article{DBLP:journals/ai/GivanLD00,
  author       = {Robert Givan and
                  Sonia M. Leach and
                  Thomas L. Dean},
  title        = {Bounded-parameter Markov decision processes},
  journal      = {Artif. Intell.},
  volume       = {122},
  number       = {1-2},
  pages        = {71--109},
  year         = {2000}
}

@book{Derman1972FiniteStateMDP,
  author    = {Cyrus Derman},
  title     = {Finite State Markovian Decision Processes},
  publisher = {Academic Press},
  year      = {1972}
}

@article{d1963probabilistic,
  title={A probabilistic production and inventory problem},
  author={d'Epenoux, Francois},
  journal={Management Science},
  volume={10},
  number={1},
  pages={98--108},
  year={1963},
  publisher={INFORMS}
}

@article{hansen2013strategy,
  title={Strategy iteration is strongly polynomial for 2-player turn-based stochastic games with a constant discount factor},
  author={Hansen, Thomas Dueholm and Miltersen, Peter Bro and Zwick, Uri},
  journal={Journal of the ACM (JACM)},
  volume={60},
  number={1},
  pages={1--16},
  year={2013},
  publisher={ACM New York, NY, USA}
}

@article{DBLP:journals/informs/KaufmanS13,
  author       = {David L. Kaufman and
                  Andrew J. Schaefer},
  title        = {Robust Modified Policy Iteration},
  journal      = {{INFORMS} J. Comput.},
  volume       = {25},
  number       = {3},
  pages        = {396--410},
  year         = {2013}
}

@inproceedings{DBLP:conf/stoc/GareyGJ76,
  author       = {M. R. Garey and
                  Ronald L. Graham and
                  David S. Johnson},
  title        = {Some NP-Complete Geometric Problems},
  booktitle    = {{STOC}},
  pages        = {10--22},
  publisher    = {{ACM}},
  year         = {1976}
}

@book{folland1999real,
  title={Real analysis: modern techniques and their applications},
  author={Folland, Gerald B},
  year={1999},
  publisher={John Wiley \& Sons}
}

@inproceedings{wolff2012robust,
  author    = {Eric M. Wolff and
               Ufuk Topcu and
               Richard M. Murray},
  title     = {Robust Control of Uncertain {Markov Decision Processes}
               with Temporal Logic Specifications},
  booktitle = {51st IEEE Conference on Decision and Control},
  pages     = {3372--3379},
  year      = {2012},
  publisher = {IEEE},
  doi       = {10.1109/CDC.2012.6426174}
}

@inproceedings{meggendorfer2025solving,
  author    = {Tobias Meggendorfer and
               Maximilian Weininger and
               Patrick Wienh{\"o}ft},
  title     = {Solving Robust {Markov Decision Processes}:
               Generic, Reliable, Efficient},
  booktitle = {Proceedings of the AAAI Conference on Artificial Intelligence},
  volume    = {39},
  number    = {25},
  pages     = {26631--26641},
  year      = {2025},
  doi       = {10.1609/aaai.v39i25.34865}
}

@inproceedings{asadi2026qualitative,
  author    = {Ali Asadi and
               Krishnendu Chatterjee and
               Ehsan Kafshdar Goharshady and
               Mehrdad Karrabi and
               Ali Shafiee},
  title     = {Qualitative Analysis of {$\omega$}-Regular Objectives
               on Robust {MDPs}},
  booktitle = {Proceedings of the AAAI Conference on Artificial Intelligence},
  volume    = {40},
  number    = {43},
  pages     = {36137--36145},
  year      = {2026},
  doi       = {10.1609/aaai.v40i43.40931}
}

@misc{asadi2026quantitative,
  author        = {Ali Asadi and
                   Krishnendu Chatterjee and
                   Ehsan Kafshdar Goharshady and
                   Mehrdad Karrabi and
                   Alipasha Montaseri and
                   Ali Shafiee},
  title         = {Quantitative Analysis of {$\omega$}-Regular Robust {MDPs}},
  year          = {2026},
  eprint        = {2608.25968},
  archivePrefix = {arXiv},
  primaryClass  = {cs.AI},
  url           = {https://arxiv.org/abs/2608.25968}
}

@inproceedings{CMSS23,
	author = {Krishnendu Chatterjee and Tobias Meggendorfer and Raimundo Saona and Jakub Svoboda},
	booktitle = {Proceedings of the 2023 Annual ACM-SIAM Symposium on Discrete Algorithms (SODA)},
	doi = {10.1137/1.9781611977554.ch173},
	eprint = {https://epubs.siam.org/doi/pdf/10.1137/1.9781611977554.ch173},
	pages = {4590--4605},
	publisher = {SIAM},
	year = {2023},
	title = {Faster Algorithm for Turn-based Stochastic Games with Bounded Treewidth},
	url = {https://epubs.siam.org/doi/abs/10.1137/1.9781611977554.ch173}}
\bibliographystyle{alpha}

% The \author macro works with any number of authors. There are two commands
% used to separate the names and addresses of multiple authors: \And and \AND.
%
% Using \And between authors leaves it to LaTeX to determine where to break the
% lines. Using \AND forces a line break at that point. So, if LaTeX puts 3 of 4
% authors names on the first line, and the last on the second line, try using
% \AND instead of \And before the third author name.

%%%%%%%%%%%%%%%%%%%%%%%%%%%%%%%%%%%%%%%%%%%%%%%%%%%%%%%%%%%%

\appendix
\raggedbottom

\section{Proof of \texorpdfstring{\cref{theorem:oracle_access}}{Theorem 1}: Strongly Polynomial PI for RMDPs via RMC Oracles} \label{app:thm1}

In this section, we show that policy iteration finds an optimal
policy and value vector using a polynomial number of iterations
and RMC oracle calls, for compact uncertainty sets and a fixed
discount factor.
We count exact arithmetic operations and oracle calls.
The oracle must accept the possibly irrational costs of the
constructed RMCs.
We detail the procedure in Algorithm~\ref{algo:RMDPPI}.

Let $\rmdp = (\states,\cost,\actions,\succ,\uncert)$ denote the input RMDP. We use $\rmdp^\agentpol$ to denote the RMC induced by fixing the strategy $\agentpol$ for the agent.  We first define the Bellman operator and state its properties in Lemma~\ref{lemma:contractive_value_iteration_rmdp}. These properties guarantee the exponential convergence rate of \texttt{RMDP-PI}, as shown in Lemma~\ref{lemma:consecutive_policy_iteration_bound_rmdp}. Next, we introduce a potential function. We use this function in Lemma~\ref{lemma:lowerbound_potential_rmdp} to establish lower and upper bounds on the value error of the policies. Finally, Lemma~\ref{thm:rmdp:iterations} applies these bounds to prove that \texttt{RMDP-PI} terminates in a polynomial number of iterations.

\subsection{Robust Bellman Optimality Operator}
We define the robust Bellman optimality operator $\Bellman$ for the RMDP $\rmdp$ as follows:
\[
    (\Bellman \valvector)_s = \costvector_s + \discountfactor \inf_{a \in \actions} \sup_{\pvector \in \uncertaintyset_{s,a}} \pvector^\top \valvector.
\]
Note that because the action space $\actions$ is finite, the infimum is always achieved. Thus, we can use $\inf_{a \in \actions}$ and $\min_{a \in \actions}$ interchangeably. 
Also, because the uncertainty set $\uncertaintyset_{s,a}$ is compact, and the inner product $\pvector^\top \valvector$ is a continuous function with respect to $\pvector$, the supremum over this set is attainable. Consequently, we can similarly use $\sup_{\pvector \in \uncertaintyset_{s,a}}$ and $\max_{\pvector \in \uncertaintyset_{s,a}}$ interchangeably.

Furthermore, for a fixed agent policy $\agentpol$, its value vector $\valvector^\agentpol$ is the unique fixed point of the policy evaluation operator. For any state $s \in \states$, it satisfies the following Bellman equation:

$$\valvector^\agentpol_s = \costvector_s + \discountfactor \max_{\pvector \in \uncertaintyset_{s,\agentpol(s)}} \pvector^\top \valvector^\agentpol$$

Let $\trans^\agentpol$ denote the transition matrix induced by the agent's policy $\agentpol$ and the adversary's optimal response, where the adversary's transition probabilities at each state $s$ are given by $\argmax_{\pvector \in \uncertaintyset_{s,\agentpol(s)}} \pvector^\top \valvector^\agentpol$. In matrix form, the value vector can be expressed as:

\begin{equation} \label{eq:value_definition_rmdp} 
    \valvector^\agentpol = 
    \costvector + \discountfactor \trans^\agentpol \valvector^\agentpol =
    (\identitymatrix - \discountfactor \trans^\agentpol)^{-1} \costvector 
\end{equation}

The following lemma establishes that $\Bellman$ is a contraction mapping with a unique fixed point. This is a standard result in robust dynamic programming \cite{iyengar2005robust}.

\begin{lemma} \label{lemma:contractive_value_iteration_rmdp} \label{lemma:value_iteration_decreasing_rmdp}
    The following statements hold:
    \begin{itemize}
        \item For all $\vvector, \uvector \in \R^n$, $\infnorm{\Bellman \uvector - \Bellman \vvector} \leq \discountfactor \, \infnorm{\uvector - \vvector}$.
        \item For every agent policy $\agentpol$ with value vector $\valvector^\agentpol$, $\Bellman \valvector^\agentpol \preccurlyeq \valvector^\agentpol$.
        \item There exists a unique $\valopt\in \R^n$ such that $\Bellman\valopt= \valopt$.
    \end{itemize}
\end{lemma}

\begin{proof}
    \begin{itemize}
        \item To prove the contraction property, fix a state $s \in \states$. Using the properties of the supremum and infimum, we can bound the difference as follows:
        \begin{align*}
            (\Bellman \uvector)_s - (\Bellman \valvector)_s
            &= \discountfactor \left( \inf_{a \in \actions} \sup_{\pvector \in \uncertaintyset_{s,a}} \pvector^\top \uvector - \inf_{a \in \actions} \sup_{\pvector \in \uncertaintyset_{s,a}} \pvector^\top \valvector \right) & (\text{Definition of } \Bellman) \\
            &\leq \discountfactor \sup_{a \in \actions} \left( \sup_{\pvector \in \uncertaintyset_{s,a}} \pvector^\top \uvector - \sup_{\pvector \in \uncertaintyset_{s,a}} \pvector^\top \valvector \right) & (\inf f - \inf g \leq \sup(f - g)) \\
            &\leq \discountfactor \sup_{a \in \actions} \sup_{\pvector \in \uncertaintyset_{s,a}} \left( \pvector^\top \uvector - \pvector^\top \valvector \right) & (\sup f - \sup g \leq \sup(f - g)) \\
            &= \discountfactor \sup_{a \in \actions} \sup_{\pvector \in \uncertaintyset_{s,a}} \pvector^\top ( \uvector - \valvector ) & (\text{Linearity}) \\
            &\leq \discountfactor \infnorm{\uvector - \valvector} & (\text{Since } \pvector \text{ is a distribution})
        \end{align*}
        Applying symmetric reasoning by swapping $\uvector$ and $\valvector$ yields $(\Bellman \valvector)_s - (\Bellman \uvector)_s \leq \discountfactor \infnorm{\uvector - \valvector}$. Combining these gives $|(\Bellman \uvector)_s - (\Bellman \valvector)_s| \leq \discountfactor \infnorm{\uvector - \valvector}$. Since this holds for all states $s \in \states$, it establishes $\infnorm{\Bellman \uvector - \Bellman \valvector} \leq \discountfactor \infnorm{\uvector - \valvector}$.
        
        \item To prove the second statement, we evaluate the operator on the value vector $\valvector^{\agentpol}$ for a given policy~$\agentpol$. For any state $s \in \states$, we have:
        \begin{align*}
            (\Bellman \valvector^{\agentpol})_s
            &= \costvector_s + \discountfactor \inf_{a \in \actions} \sup_{\pvector \in \uncertaintyset_{s,a}} \pvector^\top \valvector^{\agentpol} & (\text{Definition of } \Bellman) \\
            &\leq \costvector_s + \discountfactor \sup_{\pvector \in \uncertaintyset_{s,\agentpol(s)}} \pvector^\top \valvector^{\agentpol} & (\text{Restricting infimum to } \agentpol(s)) \\
            &= \valvector^{\agentpol}_s & (\text{Definition of } \valvector^{\agentpol})
        \end{align*}
        This proves $\Bellman \valvector^\agentpol \preccurlyeq \valvector^\agentpol$.

        \item Finally, because $\Bellman$ is a contraction mapping on the complete metric space $\R^n$ equipped with the infinity norm, the Banach fixed-point theorem guarantees the existence of a unique vector $\valopt \in \R^n$ such that $\Bellman\valopt = \valopt$.
    \end{itemize}
\end{proof}

Now, we show that $\valvector^{\star}$ has a lower value in each state than any other value vector induced by an arbitrary agent policy.

\begin{lemma} \label{lemma:optimal_value_lower_bound_rmdp}
    For any arbitrary agent policy $\agentpol$ with value vector $\valvector^\agentpol$, the optimal value vector $\valvector^{\star}$ satisfies $\valvector^{\star} \preccurlyeq \valvector^\agentpol$.
\end{lemma}

\begin{proof}
    We first establish the monotonicity of the robust Bellman operator $\Bellman$. Let $\uvector, \vvector \in \R^{|\states|}$ be two value vectors such that $\uvector \preccurlyeq \vvector$. For any state $s \in \states$:
    \begin{align*}
        (\Bellman \uvector)_s &= \costvector_s + \discountfactor \inf_{a \in \actions} \sup_{\pvector \in \uncertaintyset_{s,a}} \pvector^\top \uvector \\
        &\leq \costvector_s + \discountfactor \inf_{a \in \actions} \sup_{\pvector \in \uncertaintyset_{s,a}} \pvector^\top \vvector & (\text{Since } \uvector \preccurlyeq \vvector \text{ and } \pvector \geq \mathbf{0}) \\
        &= (\Bellman \vvector)_s
    \end{align*}
    Thus, $\Bellman$ is a monotonic operator: $\uvector \preccurlyeq \vvector \implies \Bellman \uvector \preccurlyeq \Bellman \vvector$.
    
    From Lemma~\ref{lemma:contractive_value_iteration_rmdp}, we know that for any agent policy $\agentpol$, the operator evaluated at the policy's value vector yields $\Bellman \valvector^\agentpol \preccurlyeq \valvector^\agentpol$. 
    
    By repeatedly applying the monotonic operator $\Bellman$ to both sides of this inequality, we generate a chain of non-increasing vectors:
    \begin{align*}
        \Bellman^2 \valvector^\agentpol &\preccurlyeq \Bellman \valvector^\agentpol \preccurlyeq \valvector^\agentpol \\
        \Bellman^k \valvector^\agentpol &\preccurlyeq \valvector^\agentpol \quad \text{for any integer } k \geq 1
    \end{align*}
    Because $\Bellman$ is a contraction mapping (Lemma~\ref{lemma:contractive_value_iteration_rmdp}), the sequence $\Bellman^k \valvector^\agentpol$ converges to the unique optimal fixed point $\valvector^{\star}$ as $k \to \infty$. Taking the limit of both sides preserves the non-strict inequality:
    \begin{align*}
        \lim_{k \to \infty} \Bellman^k \valvector^\agentpol &\preccurlyeq \valvector^\agentpol \\
        \valvector^{\star} &\preccurlyeq \valvector^\agentpol
    \end{align*}
    This completes the proof.
\end{proof}

We next show that \cref{algo:RMDPPI} generates a non-increasing sequence of value functions that converges exponentially to the optimal value. This property guarantees monotonic policy improvement at each iteration.

\begin{lemma}\label{lemma:consecutive_policy_iteration_bound_rmdp}
    Let $\agentpol^t$ and $\agentpol^{t+1}$ denote consecutive policies generated by \texttt{RMDP-PI} (Algorithm~\ref{algo:RMDPPI}), and let $\valvector^t$ denote the value function of $\agentpol^t$. For all $t \geq 0$, the following statements hold:
    \begin{itemize}
        \item $\valvector^{t+1} \preccurlyeq  \valvector^{t}$.
        % \item $\valvector^{t+1} \preccurlyeq \Bellman \valvector^t$
        \item $\infnorm{\valvector^t - \valvector^{\star}} \leq \discountfactor^t \infnorm{\valvector^0 - \valvector^{\star}}$.
    \end{itemize} 
\end{lemma}

\begin{proof}
    \begin{itemize}
        \item Let $\trans^t$ be the transition matrix induced by $\agentpol^t$ and an optimal adversarial response to $\agentpol^t$. We first establish that $\valvector^{t+1} \preccurlyeq \valvector^{t}$. We express the difference as follows:
        \begin{align*}
            \valvector^{t} - \valvector^{t + 1}
            & = \valvector^t - (\identitymatrix - \discountfactor \trans^{t+1})^{-1} \costvector & (\text{By Equation \ref{eq:value_definition_rmdp}}) \\
            & = (\identitymatrix - \discountfactor \trans^{t+1})^{-1} (\identitymatrix - \discountfactor \trans^{t+1}) \left[ \valvector^t - (\identitymatrix - \discountfactor \trans^{t+1})^{-1} \costvector \right] & (\text{Multiplying by } \identitymatrix) \\
            & = (\identitymatrix - \discountfactor \trans^{t+1})^{-1} \left[ (\identitymatrix - \discountfactor \trans^{t+1}) \valvector^t - \costvector \right] & (\text{Expanding the product}) \\
            & = (\identitymatrix - \discountfactor \trans^{t+1})^{-1} \left[ \valvector^t - (\costvector + \discountfactor \trans^{t+1} \valvector^t) \right] & (\text{Rearranging terms}) \\
            & \succcurlyeq (\identitymatrix - \discountfactor \trans^{t+1})^{-1} \left[ \valvector^t - \Bellman \valvector^t \right] & (\text{By Definition of $\Bellman$})
        \end{align*}
        The final inequality follows from the definition of $\Bellman$ and the non-negativity of the matrix $(\identitymatrix - \discountfactor \trans^{t+1})^{-1}$, which we establish below. By Lemma \ref{lemma:value_iteration_decreasing_rmdp}, we know that $\valvector^t - \Bellman \valvector^t \succcurlyeq \mathbf{0}$. To evaluate $(\identitymatrix - \discountfactor \trans^{t+1})^{-1}$, we can express it using the Neumann series expansion. Because $\trans^{t+1}$ is a stochastic transition matrix, its spectral radius (maximum magnitude of eigenvalues) is $1$. Multiplying it by the discount factor $\discountfactor \in (0,1)$ bounds the spectral radius of $\discountfactor \trans^{t+1}$ to less than $1$. This shows the inverse exists and can be expressed as the convergent infinite series:
        \[
            (\identitymatrix - \discountfactor \trans^{t+1})^{-1} = \sum_{k=0}^{\infty} (\discountfactor \trans^{t+1})^k
        \]
        Furthermore, because $\discountfactor > 0$ and every entry in the probability matrix $\trans^{t+1}$ is non-negative, every matrix in this infinite sum is also non-negative. Consequently, $(\identitymatrix - \discountfactor \trans^{t+1})^{-1}$ is non-negative, and multiplying it by the non-negative vector $\left[ \valvector^t - \Bellman \valvector^t \right]$ ensures that the final product is non-negative. Therefore, $\valvector^{t} - \valvector^{t+1} \succcurlyeq \mathbf{0}$, which proves $\valvector^{t+1} \preccurlyeq \valvector^{t}$.
        
        \item Next, we show that $\valvector^{t+1} \preccurlyeq \Bellman \valvector^t$. We analyze the difference:
        \begin{align*}
            \Bellman \valvector^t - \valvector^{t+1}
            & \succcurlyeq \costvector + \discountfactor \trans^{t+1} \valvector^t - \valvector^{t+1} & (\text{Definition of } \Bellman \text{ and } \agentpol^{t+1}) \\
            & = (\identitymatrix - \discountfactor \trans^{t+1})\valvector^{t+1} + \discountfactor \trans^{t+1} \valvector^t - \valvector^{t+1} & (\text{Substituting for } \costvector) \\
            & = \valvector^{t+1} - \discountfactor \trans^{t+1} \valvector^{t+1} + \discountfactor \trans^{t+1} \valvector^t - \valvector^{t+1} & (\text{Expanding the first term}) \\
            & = \discountfactor \trans^{t+1} (\valvector^t - \valvector^{t+1}) & (\text{Simplifying})
        \end{align*}
        The first inequality holds because $\agentpol^{t+1}$ is greedy for $\valvector^t$ and $\trans^{t+1}_s \in \uncertaintyset_{s,\agentpol^{t+1}(s)}$. Since $\discountfactor > 0$, $\trans^{t+1} \geq 0$, and $\valvector^t \succcurlyeq \valvector^{t+1}$ (as proved in the previous step), we conclude that $\valvector^{t+1} \preccurlyeq \Bellman \valvector^{t}$. 
        
        Finally, we prove the second statement of the lemma using induction on $t$. The base case for $t=0$ holds trivially because substituting $t=0$ yields $\infnorm{\valvector^0 - \valvector^{\star}} \leq \discountfactor^0 \infnorm{\valvector^0 - \valvector^{\star}}$, which simplifies to $\infnorm{\valvector^0 - \valvector^{\star}} \leq \infnorm{\valvector^0 - \valvector^{\star}}$. For $t > 0$, we have:
        \begin{align*}
            \infnorm{\valvector^t - \valvector^{\star}}
            & \leq \infnorm{\Bellman \valvector^{t-1} - \valvector^{\star}} & ( \valvector^{\star} \preccurlyeq \valvector^t \preccurlyeq \Bellman\valvector^{t-1}) \\
            & = \infnorm{\Bellman \valvector^{t-1} - \Bellman \valvector^{\star}} & (\text{Since } \valvector^{\star} \text{ is the fixed point}) \\
            & \leq \discountfactor \infnorm{\valvector^{t-1} - \valvector^{\star}} & (\text{By contraction in Lemma \ref{lemma:contractive_value_iteration_rmdp}}) \\
            & \leq \discountfactor^t \infnorm{\valvector^0 - \valvector^{\star}} & (\text{By induction})
        \end{align*}
    \end{itemize}
\end{proof}

\subsection{Potential Function}
We next define the potential function. Intuitively, the potential of an action $a$ in state $s$ measures the additional cost incurred by deviating from the optimal policy $\agentpolopt$ and selecting $a$ instead of $\agentpolopt(s)$. 

\begin{definition}[\textbf{Potential function}] \label{def:potential_function_rmdp}
    Given a state $s \in \states$ and an action $a \in \actions$, the potential function $f(s,a)$ is defined as:
    $$f(s,a) := \sup_{\pvector \in \uncertaintyset_{s,a}}{\pvector^\top \valvector^{\star}} - \sup_{\pvector \in \uncertaintyset_{s,\agentpol^{\star}(s)}}{\pvector^\top \valvector^{\star}}.$$
\end{definition}

Next, we use the potential function to establish lower and upper bounds on the difference between an arbitrary policy's value vector and the optimal value.

\begin{lemma} \label{lemma:lowerbound_potential_rmdp} \label{lemma:upperbound_potential_rmdp}
    Let $\agentpol$ be an arbitrary policy for the agent. The following statements hold:
    \begin{itemize}
        \item \textbf{(Lower bound)} If $\agentpol(s) = a$ for some state $s \in \states$, then $\valvector^{\agentpol}_s - \valvector^{\star}_s \geq \discountfactor f(s,a)$.
        \item \textbf{(Upper bound)} Let $\hat{s} = \argmax_{s \in \states} f(s,\agentpol(s))$. Then $\infnorm{\valvector^\agentpol - \valvector^{\star}} \leq \frac{\discountfactor}{1 - \discountfactor} f(\hat{s},\agentpol(\hat{s}))$.
    \end{itemize}
\end{lemma}

\begin{proof}
    \begin{itemize}
        \item \textbf{(Lower bound)} We lower bound the value difference as follows:
    \begin{align*}
        \valvector^\agentpol_s - \valvector^{\star}_s
        &= \left( \costvector_s + \discountfactor \sup_{\pvector \in \uncertaintyset_{s,a}}{ \pvector^\top \valvector^\agentpol} \right)
         - \left( \costvector_s + \discountfactor \sup_{\pvector \in \uncertaintyset_{s,\agentpol^{\star}(s)}}{ \pvector^\top \valvector^{\star}} \right)
        & (\text{By Equation~\ref{eq:value_definition_rmdp}})
        \\
        & = \discountfactor \left(\sup_{\pvector \in \uncertaintyset_{s,a}}{\pvector^\top \valvector^\agentpol} - \sup_{\pvector \in \uncertaintyset_{s,\agentpol^{\star}(s)}}{\pvector^\top \valvector^{\star}}\right) & (\text{Simplifying})
        \\
        &\geq \discountfactor \left(\sup_{\pvector \in \uncertaintyset_{s,a}}{\pvector^\top \valvector^{\star}} - \sup_{\pvector \in \uncertaintyset_{s,\agentpol^{\star}(s)}}{\pvector^\top \valvector^{\star}}\right)
        & (\valvector^\agentpol \succcurlyeq \valvector^{\star} \text{By \cref{lemma:optimal_value_lower_bound_rmdp}}) \\
        &= \discountfactor f(s,a)
        & (\text{By Definition})
    \end{align*}

    \item \textbf{(Upper bound)} Assume $\envpol$ is the adversary's best response to $\agentpol$, and $\envpol^{\star}$ is the adversary's best response to $\agentpol^{\star}$. Let $\trans^{\agentpol,\envpol}$ and $\trans^{\star}$ denote the corresponding transition matrices under these fixed policies. We expand the value difference using Equation~\ref{eq:value_definition_rmdp}:

    \begin{align*}
        \valvector^\agentpol - \valvector^{\star}
        &= (\identitymatrix - \discountfactor\trans^{\agentpol,\envpol})^{-1} \costvector - (\identitymatrix - \discountfactor\trans^{\star})^{-1} \costvector \\
        &= \bigg( (\identitymatrix - \discountfactor\trans^{\agentpol,\envpol})^{-1} - (\identitymatrix - \discountfactor\trans^{\star})^{-1} \bigg) \costvector.
    \end{align*}

    Applying the matrix identity $A^{-1} - B^{-1} = A^{-1}(B - A)B^{-1}$ yields:

    \begin{align*}
        \valvector^\agentpol - \valvector^{\star} &= (\identitymatrix - \discountfactor\trans^{\agentpol,\envpol})^{-1} \bigg( (\identitymatrix - \discountfactor\trans^{\star}) - (\identitymatrix - \discountfactor\trans^{\agentpol,\envpol}) \bigg) (\identitymatrix - \discountfactor\trans^{\star})^{-1} \costvector \\
        &= \discountfactor \, (\identitymatrix - \discountfactor\trans^{\agentpol,\envpol})^{-1} (\trans^{\agentpol,\envpol} - \trans^{\star}) (\identitymatrix - \discountfactor\trans^{\star})^{-1} \costvector & (\text{Simplifying terms}) \\
        &= \discountfactor \, (\identitymatrix - \discountfactor\trans^{\agentpol,\envpol})^{-1} (\trans^{\agentpol,\envpol} - \trans^{\star}) \valvector^{\star} & (\text{Substituting } \valvector^{\star})
    \end{align*}

    Since $\valvector^\agentpol \succcurlyeq \valvector^{\star}$, we know $\valvector^\agentpol - \valvector^{\star} \succcurlyeq \mathbf{0}$. We bound the vector $(\trans^{\agentpol,\envpol} - \trans^{\star}) \valvector^{\star}$ component-wise. For any state $s \in \states$:

    \begin{align*}
        ((\trans^{\agentpol,\envpol} - \trans^{\star}) \valvector^{\star})_s 
        &= \trans^{\agentpol,\envpol}_s \valvector^{\star} - \trans^{\star}_s \valvector^{\star} \\
        &= \left( \argmax_{\pvector \in \uncertaintyset_{s,\agentpol(s)}}{\pvector^\top \valvector^\agentpol} \right)^\top \valvector^{\star} - \sup_{\pvector \in \uncertaintyset_{s,\agentpol^{\star}(s)}}{\pvector^\top \valvector^{\star}} & (\text{Definition of } \trans^{\agentpol,\envpol}, \trans^{\star}) \\
        &\leq \sup_{\pvector \in \uncertaintyset_{s,\agentpol(s)}}{\pvector^\top \valvector^{\star}} - \sup_{\pvector \in \uncertaintyset_{s,\agentpol^{\star}(s)}}{\pvector^\top \valvector^{\star}} & (\text{Taking worst case in l.h.s}) \\
        &= f(s,\agentpol(s)) & (\text{Definition of } f)
    \end{align*}

    Let $\mathbf{f}^\agentpol$ be the vector where $\mathbf{f}^\agentpol_s = f(s,\agentpol(s))$. Thus, $(\trans^{\agentpol,\envpol} - \trans^{\star}) \valvector^{\star} \preccurlyeq \mathbf{f}^\agentpol$. Because the matrix $(\identitymatrix - \discountfactor\trans^{\agentpol,\envpol})^{-1}$ is non-negative, multiplying it preserves the inequality, forming a non-negative bounding chain:

    \begin{align*}
        \mathbf{0} \preccurlyeq \valvector^\agentpol - \valvector^{\star} \preccurlyeq \discountfactor \, (\identitymatrix - \discountfactor\trans^{\agentpol,\envpol})^{-1} \mathbf{f}^\agentpol
    \end{align*}

    Because all vectors are non-negative, taking the infinity norm preserves the inequality:

    \begin{align*}
        \infnorm{\valvector^\agentpol - \valvector^{\star}}
        &\leq \infnorm{\discountfactor \, (\identitymatrix - \discountfactor\trans^{\agentpol,\envpol})^{-1} \mathbf{f}^\agentpol} \\
        &\leq \discountfactor \, \infnorm{(\identitymatrix - \discountfactor\trans^{\agentpol,\envpol})^{-1}} \, \infnorm{\mathbf{f}^\agentpol} & (\text{Submultiplicativity of } \ell_\infty)
    \end{align*}

    We bound the first term using a Neumann series expansion by the same approach as \cref{lemma:consecutive_policy_iteration_bound_rmdp}. Because $\trans^{\agentpol,\envpol}$ is a stochastic matrix, $\infnorm{(\trans^{\agentpol,\envpol})^i} = 1$. The series converges as follows:

    \begin{align*}
        \infnorm{(\identitymatrix - \discountfactor\trans^{\agentpol,\envpol})^{-1}}
        &= \infnorm{\sum_{i = 0}^\infty (\discountfactor \trans^{\agentpol,\envpol})^i}
        \leq \sum_{i = 0}^\infty \discountfactor^i \infnorm{(\trans^{\agentpol,\envpol})^i} = \frac{1}{1 - \discountfactor}.
    \end{align*}

    Finally, we substitute this bound into the main inequality:

    \begin{align*}
        \infnorm{\valvector^\agentpol - \valvector^{\star}}
        &\leq \frac{\discountfactor}{1 - \discountfactor} \, \infnorm{\mathbf{f}^\agentpol} \\
        &= \frac{\discountfactor}{1 - \discountfactor} \cdot \max_{s \in \states} f(s,\agentpol(s)) & (\text{Definition of } \mathbf{f}^\agentpol \text{ and } \ell_\infty) \\
        &= \frac{\discountfactor}{1 - \discountfactor} f(\hat{s},\agentpol(\hat{s})) & (\text{Definition of } \hat{s})
    \end{align*}
    This completes the proof.
    \end{itemize}
\end{proof}

We combine the lower and upper bounds to relate the value errors of two policies. Specifically, we show that if both policies choose the same action at the state with the maximum potential, the value error decreases by at most a factor of $1 - \discountfactor$.

\begin{lemma} \label{lemma:lower_upper_combination_rmdp}
    Let $\agentpol$ and $\agentpol'$ be arbitrary policies such that $\valvector^{\agentpol} \succcurlyeq \valvector^{\agentpol'}$. Let $\hat{s} = \argmax_{s \in \states} f(s,\agentpol(s))$, and assume $\agentpol'(\hat{s}) = \agentpol(\hat{s})$. Then:
    \[
        \infnorm{\valvector^{\agentpol'} - \valvector^{\star}} \geq (1 - \discountfactor) \infnorm{\valvector^{\agentpol} - \valvector^{\star}}.
    \]
\end{lemma}

\begin{proof}
    We bound the value error as follows:
    \begin{align*}
        \infnorm{\valvector^{\agentpol'} - \valvector^{\star}} 
        &\geq \valvector^{\agentpol'}_{\hat{s}} - \valvector^{\star}_{\hat{s}} & (\text{Definition of } \ell_\infty \text{ norm}) \\
        &\geq \discountfactor f(\hat{s},\agentpol'(\hat{s})) & (\text{Lower bound from Lemma~\ref{lemma:lowerbound_potential_rmdp}}) \\
        &= \discountfactor f(\hat{s},\agentpol(\hat{s})) & (\text{Since } \agentpol'(\hat{s}) = \agentpol(\hat{s})) \\
        &\geq (1 - \discountfactor) \infnorm{\valvector^{\agentpol} - \valvector^{\star}} & (\text{Upper bound from Lemma~\ref{lemma:lowerbound_potential_rmdp}})
    \end{align*}
    This completes the proof.
\end{proof}

\subsection{Elimination of Actions with Maximum Potential}

We use the convergence rate found in Lemma~\ref{lemma:consecutive_policy_iteration_bound_rmdp} to show that the action selected in the state with maximum potential cannot repeat in the algorithm for more than $L := \lfloor \log_{\discountfactor}(1-\discountfactor) \rfloor + 1$ consecutive iterations.

\begin{lemma} \label{lemma:action_elimination}
    Let $\agentpol^{0}, \ldots, \agentpol^{T}$ denote the sequence of agent policies generated by \texttt{RMDP-PI}. For any non-optimal policy $\agentpol^l$ in this sequence, let $\hat{s} = \argmax_{s \in \states} f(s,\agentpol^l(s))$. Then, $\agentpol^k(\hat{s}) \neq \agentpol^l(\hat{s})$ for all $k \geq l + L$.
\end{lemma}

\begin{proof} 
    Assume for contradiction that there exists $k \geq l + L$ such that $\agentpol^k(\hat{s}) = \agentpol^l(\hat{s})$. Because the value vectors are non-increasing, $\valvector^{\agentpol^l} \succcurlyeq \valvector^{\agentpol^k}$. Applying Lemma~\ref{lemma:lower_upper_combination_rmdp} yields:
    \[
        \infnorm{\valvector^{\agentpol^k} - \valvector^{\star}} \geq (1 - \discountfactor) \infnorm{\valvector^{\agentpol^l} - \valvector^{\star}}.
    \]
    However, the global convergence rate from Lemma~\ref{lemma:consecutive_policy_iteration_bound_rmdp} dictates:
    \[
        \infnorm{\valvector^{\agentpol^k} - \valvector^{\star}} \leq \discountfactor^{k-l} \infnorm{\valvector^{\agentpol^l} - \valvector^{\star}}.
    \]
    Since $\agentpol^l$ is not optimal, $\infnorm{\valvector^{\agentpol^l}-\valvector^{\star}}>0$.
    Combining the bounds and dividing by this positive term yields:
    \[
        1 - \discountfactor \leq \discountfactor^{k-l}.
    \]

    However, $k-l \geq L$ implies
    $\discountfactor^{k-l} \leq \discountfactor^L < 1-\discountfactor$,
    which is a contradiction.
\end{proof}

Finally, we combine the elimination steps to establish a polynomial upper bound on the total number of iterations. 

\begin{lemma} \label{thm:rmdp:iterations}
    \texttt{RMDP-PI} terminates with an optimal policy in $\mathcal{O}\left(n \cdot m \cdot \frac{\log{(1-\discountfactor)}}{\log{\discountfactor}}\right)$ iterations.
\end{lemma}

\begin{proof}
    There are exactly $n \cdot m$ distinct state-action pairs. By Lemma~\ref{lemma:action_elimination}, if $\agentpol^l$ is non-optimal,
    $\hat{s} \in \argmax_{s \in \states} f(s,\agentpol^l(s))$,
    and $\hat{a}=\agentpol^l(\hat{s})$, the policy will not select $\hat{a}$ at $\hat{s}$ for any iteration $k \geq l + L$. Consequently, the algorithm permanently eliminates at least one state-action pair every $L$ iterations until an optimal policy
    is reached. With fixed tie-breaking, at most two further iterations
    are needed for the stopping test to succeed. Therefore,
    \[
        T \leq n \cdot m \cdot L + 2 = \mathcal{O}(n \cdot m \cdot \log_\discountfactor(1-\discountfactor)) = \mathcal{O}(n \cdot m \cdot \frac{\log(1 - \discountfactor)}{\log(\discountfactor)}).
    \]
\end{proof}

\subsection{Reduction to One RMC Oracle Call} 
We show that the maximization part of the policy improvement step can be computed for all states and actions simultaneously with one call to the RMC oracle with the same size as the input size ($\mathcal{O}(nm)$).

\begin{lemma} \label{lemma:line10_computation}
    The policy improvement step in Line 10 of \cref{algo:RMDPPI} for all states $s \in \states$ can be computed with one call to an RMC oracle on an RMC with $\mathcal{O}(nm)$ states.
\end{lemma}

\begin{proof}
    In Line 10 of \texttt{RMDP-PI}, the algorithm updates the policy by computing:
    $$ \agentpol^{t}(s) \gets \argmin_{a \in \actions} \max_{\pvector \in \uncert(s,a)} \{ \cost(s) + \discountfactor \pvector^\top \val^{t-1} \} $$
    
    To evaluate this inner maximization for all state-action pairs using a single RMC oracle call, we construct a batch RMC $\rmc_{\mathrm{batch}} = (\states_{\mathrm{batch}}, \cost_{\mathrm{batch}}, \succ_{\mathrm{batch}}, \uncert_{\mathrm{batch}})$. We introduce a set of transient states to represent each state-action pair, and a set of absorbing states to represent the original states. Let $Z = \{z_{s,a} \mid s \in \states, a \in \actions\}$ be the transient states, and $X = \{x_s \mid s \in \states\}$ be the absorbing states. We formulate $\rmc_{\mathrm{batch}}$ as follows:

    \begin{itemize}
        \item $\states_{\mathrm{batch}} = Z \cup X$. The total number of states is $n \cdot m + n = \mathcal{O}(nm)$.
        \item $\cost_{\mathrm{batch}}(z_{s,a}) = \cost(s)$ for all $z_{s,a} \in Z$, and for all $x_s \in X$, $\cost_{\mathrm{batch}}(x_s) = (1-\discountfactor)\val^{t-1}(s)$.
        \item $\succ_{\mathrm{batch}}(z_{s,a}) = \{x_{s'} \mid s' \in \succ(s,a)\}$ for all $z_{s,a} \in Z$, and for all $x_s \in X$, $\succ_{\mathrm{batch}}(x_s) = \{x_s\}$ (making the original states absorbing).
        \item $\uncert_{\mathrm{batch}}(z_{s,a}) = \uncert(s,a)$ for all $z_{s,a} \in Z$, and for all $x_s \in X$, $\uncert_{\mathrm{batch}}(x_s)$ is restricted to the deterministic transition to $x_s$.
    \end{itemize}
    
    Because every state $x_s \in X$ deterministically transitions to itself with a cost of $(1-\discountfactor)\val^{t-1}(s)$, its value in $\rmc_{\mathrm{batch}}$ is $\val^{t-1}(s)$. 
    
    Consequently, the value of any transient state $z_{s,a} \in Z$ satisfies the robust Bellman equation:
    $$ \valvector_{\mathrm{batch}}(z_{s,a}) = \cost_{\mathrm{batch}}(z_{s,a}) + \discountfactor \max_{\pvector \in \uncert_{\mathrm{batch}}(z_{s,a})} \pvector^\top \valvector_{\mathrm{batch}} = \cost(s) + \discountfactor \max_{\pvector \in \uncert(s,a)} \pvector^\top \val^{t-1} $$
    
    Therefore, we construct $\rmc_{\mathrm{batch}}$ and query the RMC oracle once to find the value of all states simultaneously. Then, for each state $s \in \states$, we extract the optimal action by evaluating $\agentpol^{t}(s) \gets \argmin_{a \in \actions} \valvector_{\mathrm{batch}}(z_{s,a})$. This shows that the policy improvement step for all states requires only a single RMC oracle call on an RMC of size $\mathcal{O}(nm)$.
\end{proof}

We now combine the iteration bound with the oracle-call bound
to prove \cref{theorem:oracle_access}.

\begin{proof}[Proof of \Cref{theorem:oracle_access}]
    By \cref{thm:rmdp:iterations}, the \texttt{RMDP-PI} algorithm converges to the optimal policy in $\mathcal{O}\left(n \cdot m \cdot \frac{\log(1-\discountfactor)}{\log\discountfactor}\right)$ iterations. 
    
    We count the number of RMC oracle calls required in each iteration. During the policy evaluation step (Line 6 of Algorithm~\ref{algo:RMDPPI}), the algorithm solves the RMC for the current policy, which requires one call to the RMC oracle. Then, during the policy improvement step, \cref{lemma:line10_computation} shows that the algorithm can process all states simultaneously with one additional RMC oracle call, whose size is linear according to the input size ($\mathcal{O}(nm)$). This results in two RMC oracle queries per iteration.
    
    Because the total number of iterations is bounded by $\mathcal{O}\left(n \cdot m \cdot \frac{\log(1-\discountfactor)}{\log\discountfactor}\right)$, and each iteration requires two RMC oracle solutions, the algorithm finds the optimal policy with a polynomial number of calls to the RMC oracle.
\end{proof}

\section{Proof of \texorpdfstring{\cref{thm:strongly_poly}}{Theorem 2}: Complexity of RMCs with \texorpdfstring{$L_1, L_\infty$}{L1, L-infinity} Uncertainty Sets} \label{app:thm2}

The previous section established a polynomial bound on the number
of policy iteration steps and RMC oracle calls for RMDPs with
compact uncertainty sets and a constant discount factor. This section focuses on solving Robust Markov Chains (RMCs) with $L_1, L_\infty$ uncertainty sets.

We demonstrate that RMCs under $L_1$ and $L_\infty$ uncertainty sets can be solved in strongly polynomial time using the policy iteration procedure, \texttt{RMC-PI}, detailed in \cref{algo:RMCPI}. While the strongly polynomial solvability of the $L_\infty$ case was previously established by Asadi et al.\ \cite{asadi2026strongly}, our analysis yields a tighter bound by a factor of $|\states|$.

Before diving into the main proof, we briefly recall some standard results for Markov models that apply regardless of the specific uncertainty set \cite{puterman94,iyengar2005robust}. 

\paragraph{Robust Bellman Optimality Operator.} Given an RMC $\rmc$ and a value vector $\valvector$, the robust Bellman operator $\Bellman$ (for Algorithm \ref{algo:RMCPI}) is defined as:

\begin{equation} \label{Eq:BellmanRMC}
    (\Bellman \valvector)_s = \costvector_s + \discountfactor \sup_{\pvector \in \uncertaintyset_s} \pvector^\top \valvector.
\end{equation}

Note that because $\uncertaintyset_s$ is compact and the inner product $\pvector^\top \valvector$ is a continuous function with respect to $\pvector$, the supremum over this set is attainable. Consequently, we can use $\sup_{\pvector \in \uncertaintyset_{s}}$ and $\max_{\pvector \in \uncertaintyset_{s}}$ interchangeably.

It is a well-known result that policy iteration converges monotonically to the unique fixed point of $\Bellman$ at an exponential rate. We summarize the core contractive and monotonic properties of this operator in the following lemma.

\begin{lemma} \label{lemma:contractive_value_iteration} \label{lemma:bellman} \label{lemma:value_iteration_increasing} \label{corollary:unique_optimal}
    The following statements hold:
    \begin{itemize}
        \item For any $\uvector, \vvector \in \R^n$, $\infnorm{\Bellman \uvector - \Bellman \vvector} \leq \discountfactor \infnorm{\uvector - \vvector}$.
        \item For any adversary policy $\envpol$ with value vector $\val^{\envpol}$, $\Bellman \val^{\envpol} \succcurlyeq \val^{\envpol}$.
        \item There is a unique optimal value vector $\valvector^{\star} \in \R^n$ such that $\Bellman \valvector^{\star} = \valvector^{\star}$.
    \end{itemize}
\end{lemma}

\begin{proof}
    \begin{itemize}
        \item Fix an arbitrary state $s \in \states$. Without loss of generality, assume $(\Bellman \uvector)_s \geq (\Bellman \vvector)_s$. Let $\pvector'$ and $\pvector''$ be the respective maximizers in $\uncertaintyset_s$ for $\uvector$ and $\vvector$. We can bound the difference as follows:
        \begin{align*}
            (\Bellman\uvector)_s - (\Bellman\vvector)_s
            &= \discountfactor ({\pvector'}^\top \uvector - {\pvector''}^\top \vvector) && \text{(Definition of $\Bellman$)} \\
            &\leq \discountfactor ({\pvector'}^\top \uvector - {\pvector'}^\top \vvector) && \text{(Optimality of $\pvector''$ for $\vvector$)} \\
            &= \discountfactor {\pvector'}^\top (\uvector - \vvector) && \text{(Factorization)} \\
            &\leq \discountfactor \infnorm{\uvector - \vvector}. && \text{($\pvector'$ is a probability distribution)}
        \end{align*}
        Since this holds for any state $s$, the contraction property follows.
        
        \item For the second claim, observe that for any state $s \in \states$:
        \begin{align*}
            (\Bellman \valvector^{\envpol})_s 
            &= \costvector_s + \discountfactor \max_{\pvector \in \uncertaintyset_s} \pvector^\top \valvector^{\envpol} && \text{(Definition of $\Bellman$)} \\
            &\geq \costvector_s + \discountfactor ({\pvector^{\envpol}_s})^\top \valvector^{\envpol} && \text{($\pvector^{\envpol}_s$ is feasible in $\uncertaintyset_s$)} \\
            &= \valvector^{\envpol}_s. && \text{(Definition of $\valvector^\envpol$)}
        \end{align*}
        Because this applies to all states simultaneously, we conclude $\Bellman \valvector^{\envpol} \succcurlyeq \valvector^{\envpol}$.
        
        \item Finally, existence and uniqueness follow immediately from applying the Banach fixed-point theorem to our contraction mapping.
    \end{itemize}
\end{proof}

Next, we show that the optimal value vector $\valvector^{\star}$ is an upper bound in each entry for the value vector induced by any arbitrary adversary policy.

\begin{lemma} \label{lemma:optimal_value_upper_bound_rmc}
    For any arbitrary adversary policy $\envpol$ with value vector $\valvector^\envpol$, the optimal value vector $\valvector^{\star}$ satisfies $\valvector^{\star} \succcurlyeq \valvector^\envpol$.
\end{lemma}

\begin{proof}
    We first establish the monotonicity of the robust Bellman operator $\Bellman$ for RMCs. Let $\uvector, \vvector \in \R^{|\states|}$ be two value vectors such that $\uvector \preccurlyeq \vvector$. For any state $s \in \states$:
    \begin{align*}
        (\Bellman \uvector)_s &= \costvector_s + \discountfactor \max_{\pvector \in \uncert(s)} \pvector^\top \uvector \\
        &\leq \costvector_s + \discountfactor \max_{\pvector \in \uncert(s)} \pvector^\top \vvector & (\text{Since } \uvector \preccurlyeq \vvector \text{ and } \pvector \geq \mathbf{0}) \\
        &= (\Bellman \vvector)_s
    \end{align*}
    Thus, $\Bellman$ is a monotonic operator, meaning $\uvector \preccurlyeq \vvector \implies \Bellman \uvector \preccurlyeq \Bellman \vvector$.
    
    By \cref{lemma:bellman}, we know that evaluating the operator at the value vector of any adversary policy $\envpol$ yields $\Bellman \valvector^\envpol \succcurlyeq \valvector^\envpol$. By repeatedly applying the monotonic operator $\Bellman$ to both sides of this inequality, we generate a chain of non-decreasing vectors:
    \begin{align*}
        \valvector^\envpol \preccurlyeq \Bellman \valvector^\envpol &\preccurlyeq \Bellman^2 \valvector^\envpol \\
        \valvector^\envpol &\preccurlyeq \Bellman^k \valvector^\envpol \quad \text{for any integer } k \geq 1
    \end{align*}
    Because $\Bellman$ is a contraction mapping, the sequence $\Bellman^k \valvector^\envpol$ converges to the unique optimal fixed point $\valvector^{\star}$ as $k \to \infty$. Taking the limit of both sides preserves the non-strict inequality:
    \begin{align*}
        \valvector^\envpol &\preccurlyeq \lim_{k \to \infty} \Bellman^k \valvector^\envpol \\
        \valvector^\envpol &\preccurlyeq \valvector^{\star}
    \end{align*}
    This establishes that $\valvector^{\star} \succcurlyeq \valvector^\envpol$, completing the proof.
\end{proof}

By the properties established in Lemmas~\ref{lemma:value_iteration_increasing},~\ref{lemma:optimal_value_upper_bound_rmc}, we now show that the sequence of value vectors produced by the \texttt{RMC-PI} algorithm is non-decreasing and converges at an exponential rate.

\begin{lemma} \label{lemma:consecutive_policy_iteration_bound} \label{lemma:policy_iteration_less_than_value_iteration} \label{lemma:policy_iteration_exponential_upperbound}
    Suppose $\envpol^0, \envpol^1, \dots$ represents the sequence of policies obtained via \texttt{RMC-PI} (Algorithm~\ref{algo:RMCPI}), with $\valvector^t$ denoting the corresponding value vector for policy $\envpol^t$. For any iteration $t$, the following properties hold:
    \begin{itemize}
        \item $\valvector^{t+1} \succcurlyeq \valvector^{t}$
        % \item $\valvector^{t+1} \succcurlyeq \Bellman \valvector^t$
        \item $\infnorm{\valvector^t - \valvector^{\star}} \leq \discountfactor^t \infnorm{\valvector^0 - \valvector^{\star}}$
    \end{itemize}
\end{lemma}

\begin{proof}
    \begin{itemize}
        \item Let $\trans^{t}$ denote the transition matrix corresponding to the policy chosen at iteration $t$. We analyze the difference between consecutive value vectors:
        \begin{align*}
            \valvector^{t+1} - \valvector^{t} 
            &= (\identitymatrix - \discountfactor \trans^{t+1})^{-1} \costvector - \valvector^t && \text{(Definition of policy value)} \\
            &= (\identitymatrix - \discountfactor \trans^{t+1})^{-1} \left[ \costvector - (\identitymatrix - \discountfactor \trans^{t+1}) \valvector^t \right] && \text{(Factoring out the inverse matrix)} \\
            &= (\identitymatrix - \discountfactor \trans^{t+1})^{-1} \left[ \costvector + \discountfactor \trans^{t+1} \valvector^t - \valvector^t \right] && \text{(Expanding the inner terms)} \\
            & \succcurlyeq (\identitymatrix - \discountfactor \trans^{t+1})^{-1} \left[ \Bellman \valvector^t - \valvector^t \right]. && \text{(Definition of $\Bellman$)}
        \end{align*}
        The final inequality follows from the definition of $\Bellman$ and the non-negativity of the matrix $(\identitymatrix - \discountfactor \trans^{t+1})^{-1}$, which we establish below.
        By Lemma~\ref{lemma:value_iteration_increasing}, we know $\Bellman \valvector^t - \valvector^t \succcurlyeq \mathbf{0}$. We can express the inverse matrix using the Neumann series:
        \[
            (\identitymatrix - \discountfactor \trans^{t+1})^{-1} = \sum_{i=0}^\infty (\discountfactor \trans^{t+1})^i.
        \]
        Because $\discountfactor$ and all entries of $\trans^{t+1}$ are non-negative, the resulting sum consists entirely of non-negative terms and converges. Consequently, the matrix $(\identitymatrix - \discountfactor \trans^{t+1})^{-1}$ is component-wise non-negative, which implies $\valvector^{t+1} - \valvector^t \succcurlyeq \mathbf{0}$.
        
        \item We first establish the intermediate bound $\valvector^{t+1} \succcurlyeq \Bellman \valvector^{t}$. We can show this by expanding the difference:
        \begin{align*}
            \valvector^{t+1} - \Bellman \valvector^t 
            &= \valvector^{t+1} - (\costvector + \discountfactor \trans^{t+1} \valvector^t) && \text{(Definition of $\Bellman$ and policy update)} \\
            &= \valvector^{t+1} - \left( (\identitymatrix - \discountfactor \trans^{t+1})\valvector^{t+1} + \discountfactor \trans^{t+1} \valvector^t \right) && \text{(Substituting $\costvector=(\I-\gamma\trans^{t+1})\val^{t+1}$)} \\
            &= \valvector^{t+1} - (\valvector^{t+1} - \discountfactor \trans^{t+1} \valvector^{t+1} + \discountfactor \trans^{t+1} \valvector^t) && \text{(Expanding the expression)} \\
            &= \discountfactor \trans^{t+1} (\valvector^{t+1} - \valvector^t). && \text{(Simplifying)}
        \end{align*}
        Given $\discountfactor \geq 0$, $\trans^{t+1} \geq 0$ entrywise, and our previous result that $\valvector^{t+1} - \valvector^t \succcurlyeq \mathbf{0}$, we conclude $\valvector^{t+1} \succcurlyeq \Bellman \valvector^{t}$.

        Now we prove the second statement of the lemma using induction on $t$. The base case $t=0$ holds trivially. For $t > 0$:
        \begin{align*}
            \infnorm{\valvector^t - \valvector^{\star}} 
            &\leq \infnorm{\Bellman \valvector^{t-1} - \valvector^{\star}} && \text{(Because $\valvector^{\star} \succcurlyeq \valvector^t \succcurlyeq \Bellman \valvector^{t-1}$)} \\
            &= \infnorm{\Bellman \valvector^{t-1} - \Bellman \valvector^{\star}} && \text{(Fixed point property, Lemma~\ref{corollary:unique_optimal})} \\
            &\leq \discountfactor \infnorm{\valvector^{t-1} - \valvector^{\star}} && \text{(Contraction of $\Bellman$, Lemma~\ref{lemma:contractive_value_iteration})} \\
            &\leq \discountfactor^t \infnorm{\valvector^0 - \valvector^{\star}}. && \text{(By induction)}
        \end{align*}
    \end{itemize}
\end{proof}

\subsection{Main Proof of \texorpdfstring{\cref{thm:strongly_poly}}{Theorem 2}}

In this section, we will show that policy iteration is a strongly polynomial-time algorithm for RMCs with $L_1$ or $L_\infty$ uncertainty set and a constant discount factor. This, combined with \cref{theorem:oracle_access} shows that RMDPs with $L_1$ or $L_\infty$ uncertainty sets and a constant discount factor are strongly polynomial.

We prove this by showing that policy iteration converges to the optimal policy in a polynomial number of iterations. To organize our analysis, we first introduce fast linear-time algorithms for the policy improvement step and identify the structural properties of the distributions they produce. Then, we use these properties to bound the convergence rate of our method.

Before we analyze the algorithms, we define the value of a policy in an RMC. For any adversary policy $\envpol$, let $\trans^\envpol$ be its transition matrix and $\pvector^\envpol_s$ be its transition probabilities at state $s$. The value vector $\valvector^\envpol$ satisfies the following Bellman equation for each state $s \in \states$:
\begin{equation} \label{Eq:BellmanRMCPolicy}
    \valvector^\envpol_s = \costvector_s + \discountfactor {\pvector^\envpol_s}^\top \valvector^\envpol
\end{equation}
We can also write this in matrix form to calculate the value vector directly:
\begin{equation} \label{Eq:RMCValue}
    \valvector^\envpol = (\identitymatrix - \discountfactor \trans^\envpol)^{-1} \costvector
\end{equation}

\subsubsection{Policy Improvement and Structural Property.}

Line 9 of Algorithm~\ref{algo:RMCPI} requires solving the following optimization problem:
\begin{equation*}
    \max_{\pvector \in \uncert(s)}{\pvector^\top \valvector}
\end{equation*}
where $\uncert(s) = \left\{ \pvector \in \Delta(\states) \;\middle|\; \|\pvector - \nominal_s\|_p \leq \radius_s \right\}$ for $p \in \{1, \infty\}$ for a given value vector $\valvector$. 
Note that for $p \in \{1, \infty\}$, the uncertainty set $\uncert(s)$ can be defined by linear constraints.

While one can solve this optimization problem using linear programming, a more efficient linear-time algorithm exists for $L_1$ and $L_\infty$. 

For $L_1$, \cite{ho2021partial} presents a linear-time method, which they refer to as the homotopy algorithm. Given an RMC $\rmc$, a state $s \in \states$, and the current value vector $\valvector$, the algorithm first sorts the successor states in descending order of their values, such that $\valvector_{s_1} \geq \valvector_{s_2} \geq \dots \geq \valvector_{s_{|\succ(s)|}}$. It then iteratively shifts transition probability mass from the states with the lowest values to the state with the highest value ($s_1$). Specifically, the algorithm maintains a pointer $hi$ (initially pointing to the state with the lowest value $s_{|\succ(s)|}$) and transfers the maximum permissible probability mass from $\pvector_{s_{hi}}$ to $\pvector_{s_1}$ at each step. It decrements $hi$ once the probability mass at $s_{hi}$ is exhausted. Algorithm~\ref{alg:lone_max} details this procedure. We refer the reader to \cite{ho2021partial} for the formal proof of correctness. 

\begin{algorithm}[ht]
\caption{Best-case distribution under $L_1$ uncertainty sets}
\label{alg:lone_max}
\begin{algorithmic}[1]
    \State \textbf{Input:} State $s$, nominal distribution $\nominal_s$, state values $\val$, budget $\radius(s)$
    \State \textbf{Output:} Distribution $\pvector$ maximizing $\pvector^\top \val$
    
    \State $\pvector \gets \nominal_s$
    
    \State Sort states $s_1, \dots, s_{|\succ(s)|}$ so that $\val(s_1) \geq \val(s_2) \geq \dots \geq \val(s_{|\succ(s)|})$ \Comment{Break ties by increasing state index}
    
    \State $hi \gets |\succ(s)|$ \Comment{Index of lowest-value state}
    
    \While{$\radius(s) > 0 \text{ and } hi > 1$}
        \State $d \gets \min\!\left(\pvector_{s_{hi}}, \frac{\radius(s)}{2}\right)$
        \State $\pvector_{s_{hi}} \gets \pvector_{s_{hi}} - d$
        \State $\pvector_{s_1} \gets \pvector_{s_1} + d$
        \State $\radius(s) \gets \radius(s) - 2d$
        \State $hi \gets hi - 1$
    \EndWhile
    
    \State \Return $\pvector$
\end{algorithmic}
\end{algorithm}

The policies generated by Algorithm~\ref{alg:lone_max} exhibit the following structural property that we use in our analysis.

\begin{property}[Structural Characterization for $L_1$] \label{prop:policy_structure}
    For any policy $\pvector$ generated by Algorithm~\ref{alg:lone_max}, there exist four disjoint sets of states $R_\pvector, N_\pvector, Z_\pvector,$ and $I_\pvector$ such that:
    \begin{itemize}
        \item $|R_\pvector|=1$ (Receiver), and for $s' \in R_\pvector$ we have $\pvector_{s'} = \min(\nominal_{s,s'} + \frac{\radius_s}{2}, 1)$.
        \item For $s' \in N_\pvector$ (Not-changed), $\pvector_{s'} = \nominal_{s,s'}$.
        \item For $s' \in Z_\pvector$ (Zeros), $\pvector_{s'} = 0$.
        \item $|I_\pvector| \leq 1$ (Incomplete), and for $s' \in I_\pvector$, the probability satisfies $\pvector_{s'} = \nominal_{s,s'} -   \frac{\radius_s}{2} + \sum_{s'' \in Z_\pvector}{\nominal_{s,s''}}.$
    \end{itemize}
\end{property}

Property~\ref{prop:policy_structure} follows directly from the execution of Algorithm~\ref{alg:lone_max}. The set $R_\pvector$ contains the highest-value state receiving the shifted probability. Conversely, $Z_\pvector$ consists of states whose probability became zero, while $I_\pvector$ contains at most one state that partially gives its mass before the budget finishes. Finally, $N_\pvector$ consists of the not-changed states that the pointer $hi$ never reaches.

Similarly, for $L_\infty$, Behzadian et al.\ \cite{behzadian2021fast} present a linear algorithm. It starts with the nominal distribution $\nominal_s$. Then, it sorts the successor states from highest value to lowest value. The algorithm uses a two-pointer approach. One pointer ($hi$) starts on the highest-value state, and the other pointer ($lo$) starts on the lowest-value state. It moves probability mass from the lowest-value states to the highest-value states. Each state can only give or receive up to $\radius_s$ mass, and the probability must stay within the $[0, 1]$ limits. The pointers move until they meet. Algorithm~\ref{alg:linf_max} details this procedure.

\begin{algorithm}[ht]
\caption{Best-case distribution under $L_\infty$ uncertainty sets}
\label{alg:linf_max}
\begin{algorithmic}[1]
    \State \textbf{Input:} State $s$, nominal distribution $\nominal_s$, state values $\valvector$, budget $\radius_s$
    \State \textbf{Output:} Distribution $\pvector$ maximizing $\pvector^\top \valvector$
    
    \State Sort states $s_1, \dots, s_{|\succ(s)|}$ so that $\valvector(s_1) \geq \valvector(s_2) \geq \dots \geq \valvector(s_{|\succ(s)|})$ \Comment{Break ties by increasing state index}
    \State $\pvector \gets \nominal_s$
    \State $hi \gets 1$, $lo \gets |\succ(s)|$
    \State $b_{hi} \gets \radius_s$, $b_{lo} \gets \radius_s$
    
    \While{$hi < lo$}
        \State $d_{hi} \gets \min\!\left(b_{hi}, 1 - \pvector_{s_{hi}}\right)$
        \State $d_{lo} \gets \min\!\left(b_{lo}, \pvector_{s_{lo}}\right)$
        \State $t \gets \min\!\left(d_{hi}, d_{lo}\right)$
        
        \State $\pvector_{s_{hi}} \gets \pvector_{s_{hi}} + t$
        \State $\pvector_{s_{lo}} \gets \pvector_{s_{lo}} - t$
        \State $b_{hi} \gets b_{hi} - t$
        \State $b_{lo} \gets b_{lo} - t$
        
        \If{$b_{hi} = 0 \text{ or } \pvector_{s_{hi}} = 1$}
            \State $hi \gets hi + 1$
            \State $b_{hi} \gets \radius_s$
        \Else
            \State $lo \gets lo - 1$
            \State $b_{lo} \gets \radius_s$
        \EndIf
    \EndWhile
    
    \State \Return $\pvector$
\end{algorithmic}
\end{algorithm}

Similar to the $L_1$ case, the policies generated by Algorithm~\ref{alg:linf_max} exhibit a clear structural property.

\begin{property}[Structural Characterization for $L_\infty$] \label{prop:linf_policy_structure}
    For any policy $\pvector$ generated by Algorithm~\ref{alg:linf_max}, there exist four disjoint sets of states $R_\pvector, D_\pvector, Z_\pvector,$ and $I_\pvector$ such that:
    \begin{itemize}
        \item For $s' \in R_\pvector$ (Receivers), $\pvector_{s'} = \min(1, \nominal_{s,s'} + \radius_s)$.
        \item For $s' \in D_\pvector$ (Donors), $\pvector_{s'} = \nominal_{s,s'} - \radius_s$.
        \item For $s' \in Z_\pvector$ (Zeroed donors), $\nominal_{s,s'} \leq \radius_s$ and $\pvector_{s'} = 0$.
        \item $|I_\pvector| \leq 1$ (Incomplete), and for $s' \in I_\pvector$, the probability satisfies $\nominal_{s,s'} - \radius_s < \pvector_{s'} < \nominal_{s,s'} + \radius_s$.
    \end{itemize}
\end{property}

Property~\ref{prop:linf_policy_structure} shows how the probability mass moves. The set $R_\pvector$ receives the mass. This mass is donated by the states in $D_\pvector$ and $Z_\pvector$. Finally, there is at most one state, $I_\pvector$, that only partially gives or receives mass.

\subsubsection{Potential Function and Sandwich Bounds}

We now analyze the convergence rate of our policy iteration method. This analysis differs from the standard RMDP analysis because the action space of the adversary is infinite. Although it is easy to show that the optimal worst-case probabilities must lie on the boundary of the $L_1$ or $L_\infty$ ball, there can still be exponentially many such extreme points.

First, we define the cumulative transition probability for an adversary policy $\envpol$ as follows.

\begin{definition} \label{def:cumulative_transition}
    Fix a state $s \in \states$, and let $s_1, \dots, s_{|\succ(s)|}$ be an enumeration of the successor states such that $\valvector^{\star}_{s_1} \geq \valvector^{\star}_{s_2} \geq \dots \geq \valvector^{\star}_{s_{|\succ(s)|}}$. We break ties by increasing state index.
    Let $\pvector^\envpol_s$ denote the transition probability vector at state $s$ under the adversary policy $\envpol$. 
    For any index $i \in [|\succ(s)|]$, we define the cumulative transition probability $F^\envpol(s,i)$ as:
    \begin{equation*}
        F^\envpol(s,i) := \sum_{j = 1}^{i} \pvector^\envpol_{s,s_j}.
    \end{equation*}

    We choose $\envpol^{\star}$ by applying Algorithm~\ref{alg:lone_max}
    for $L_1$, or Algorithm~\ref{alg:linf_max} for $L_\infty$,
    at each state with value vector $\valvector^{\star}$ and this ordering.
    We write $F^{\star}(s,i)=F^{\envpol^{\star}}(s,i)$ and $\pvector^{\star}_s=\pvector^{\envpol^{\star}}_s$.
\end{definition}

We now have the following lemma.

\begin{lemma} \label{lemma:big_F_positive}
    For any state $s \in \states$, any policy $\envpol$ with $\pvector^\envpol_s \in \uncert(s)$, and any index $i \in [|\succ(s)|]$, we have $F^{\star}(s,i) \geq F^\envpol(s,i)$.
\end{lemma}

\begin{proof}
    Algorithms~\ref{alg:lone_max}, \ref{alg:linf_max} operate on a given value permutation by greedily transferring probability mass from the lowest-value states at the end of the sequence to the highest-value states at the start until the budget saturates. Because the optimal policy evaluates states using the descending permutation of the $\valvector^{\star}$ values, $\pvector^{\star}_s$ shifts the maximum allowable mass as early as possible in this specific order. For each prefix, the procedure maximizes the mass transferred into
    it from the remaining states, subject to the uncertainty constraints.
    Transfers within the prefix do not change its total probability.
    Thus, every prefix sum is maximal, which gives
    $F^{\star}(s,i) \geq F^\envpol(s,i)$ for any feasible $\envpol$.
\end{proof}

\paragraph{Potential Function.}
For a given adversary policy $\envpol$, we define the potential function $f_\envpol(s, i)$ for a state $s \in \states$ and index $i \in [|\succ(s)|-1]$ as:
\begin{equation}
    f_\envpol(s, i) := (F^{\star}(s,i) - F^\envpol(s,i))(\valvector^{\star}_{s_i} - \valvector^{\star}_{s_{i+1}}).
\end{equation}

First, we show that the potential function is non-negative for any adversary policy $\envpol$.

\begin{lemma} \label{lemma:potential_nonnegative_rmc}
    For any adversary policy $\envpol$, state $s \in \states$, and index $i \in [|\succ(s)|-1]$, the potential function is non-negative, i.e., $f_\envpol(s,i) \geq 0$.
\end{lemma}

\begin{proof}
    By Definition~\ref{def:cumulative_transition}, the successor states are enumerated in descending order of their optimal values, meaning $\valvector^{\star}_{s_1} \geq \valvector^{\star}_{s_2} \geq \dots \geq \valvector^{\star}_{s_{|\succ(s)|}}$. Consequently, for any valid index $i$, the difference between adjacent state values is non-negative:
    \begin{align*}
        \valvector^{\star}_{s_i} - \valvector^{\star}_{s_{i+1}} &\geq 0
    \end{align*}
    Furthermore, by Lemma~\ref{lemma:big_F_positive}, the chosen optimal policy maximizes every cumulative transition
    probability. Thus, for all indices $i$:
    \begin{align*}
        F^{\star}(s,i) - F^\envpol(s,i) &\geq 0
    \end{align*}
    The potential function is defined as the product of these two terms:
    \begin{align*}
        f_\envpol(s,i) &= (F^{\star}(s,i) - F^\envpol(s,i))(\valvector^{\star}_{s_i} - \valvector^{\star}_{s_{i+1}})
    \end{align*}
    Thus, it follows that $f_\envpol(s,i) \geq 0$. This completes the proof.
\end{proof}

Now we first prove the following lemma, which lower bounds the value difference in terms of the potential function.

\begin{lemma}[Lower Bound] \label{lem:lower_bound}
    For every policy $\envpol$, state $s \in \states$, and index $i \in [|\succ(s)|-1]$, we have $\valvector^{\star}_s - \valvector^\envpol_s \ge \discountfactor f_\envpol(s, i)$.
\end{lemma}

\begin{proof}

    We start by bounding the value difference using the Bellman equations:
    \begin{align*}
        \valvector^{\star}_s - \valvector^\envpol_s 
        &= (\costvector_s + \discountfactor {\pvector^{\star}_s}^\top \valvector^{\star}) - (\costvector_s + \discountfactor {\pvector^\envpol_s}^\top \valvector^\envpol) && \text{(Equation \ref{Eq:BellmanRMCPolicy})} \\
        &\geq (\costvector_s + \discountfactor {\pvector^{\star}_s}^\top \valvector^{\star}) - (\costvector_s + \discountfactor {\pvector^\envpol_s}^\top \valvector^{\star}) && \text{(Since $\valvector^{\star} \succcurlyeq \valvector^\envpol$ \text{by \cref{lemma:optimal_value_upper_bound_rmc})}} \\
        &= \discountfactor (\pvector^{\star}_s - \pvector^\envpol_s)^\top \valvector^{\star}. && \text{(Canceling $\costvector_s$ and factoring)}
    \end{align*}
    
    To evaluate this dot product, we use summation by parts. Because both $\pvector^{\star}_s$ and $\pvector^\envpol_s$ are valid probability distributions over the successor states, their total cumulative probabilities at the final index are both equal to $1$. Therefore, their cumulative difference at $n$ is zero: $F^{\star}(s,|\succ(s)|) - F^\envpol(s,|\succ(s)|) = 1 - 1 = 0$. Applying summation by parts yields:
    \begin{align*}
        \discountfactor (\pvector^{\star}_s - \pvector^\envpol_s)^\top \valvector^{\star} &= \discountfactor \sum_{j=1}^{|\succ(s)|}(\pvector^{\star}_{s,s_j} - \pvector^\envpol_{s,s_j}) \valvector^{\star}_{s_j} && \text{(Expanding)} \\
        &= \discountfactor \Bigg[ \sum_{i=1}^{|\succ(s)|-1} \left((F^{\star}(s,i) - F^\envpol(s,i))(\valvector^{\star}_{s_i} - \valvector^{\star}_{s_{i+1}})\right) \\
        &\qquad\quad + (F^{\star}(s,|\succ(s)|) - F^\envpol(s,|\succ(s)|))\valvector^{\star}_{s_{|\succ(s)|}} \Bigg] && \text{(Summation by parts)} \\
        &= \discountfactor \sum_{i=1}^{|\succ(s)|-1} (F^{\star}(s,i) - F^\envpol(s,i))(\valvector^{\star}_{s_i} - \valvector^{\star}_{s_{i+1}}) && \text{(Last term is $0$)} \\
        &= \discountfactor \sum_{i=1}^{|\succ(s)|-1} f_\envpol(s, i). && \text{(Definition of $f_\envpol(s,i)$)}
    \end{align*}
    Because each $f_\envpol(s, i)$ is non-negative by \cref{lemma:potential_nonnegative_rmc}, the total sum is bounded below by any single term $\discountfactor f_\envpol(s, i)$.
\end{proof}

Now we show the following lemma, which upper bounds the value difference based on the potential function.

\begin{lemma}[Upper Bound] \label{lem:upper_bound}
    Let $\hat{f}_\envpol = \max_{s \in \states, i} f_\envpol(s, i)$. Then, $\infnorm{\valvector^{\star} - \valvector^\envpol} \le \frac{\discountfactor n}{1-\discountfactor} \hat{f}_\envpol$.
\end{lemma}

\begin{proof}
    First, we relate the value difference to the transition difference:

    \begin{align*}
        \valvector^{\star} - \valvector^\envpol
        &= (\identitymatrix - \discountfactor\trans^{\star})^{-1} \costvector - (\identitymatrix - \discountfactor\trans^\envpol)^{-1} \costvector 
        && \text{(Equation \ref{Eq:RMCValue})} \\
        &= (\identitymatrix - \discountfactor\trans^\envpol)^{-1} \left[ (\identitymatrix - \discountfactor\trans^\envpol) - (\identitymatrix - \discountfactor\trans^{\star}) \right] \notag \\
        &\quad \cdot (\identitymatrix - \discountfactor\trans^{\star})^{-1} \costvector
        && (A^{-1} - B^{-1} = B^{-1}(B - A)A^{-1}) \\
        &= \discountfactor (\identitymatrix - \discountfactor\trans^\envpol)^{-1} (\trans^{\star} - \trans^\envpol) \valvector^{\star}.
        && \text{(Simplifying terms and Equation \ref{Eq:RMCValue})}
    \end{align*}
    
    We take the infinity norm of both sides and apply the Neumann series expansion:
    \begin{align*}
        \infnorm{\valvector^{\star} - \valvector^\envpol} 
        &\leq \discountfactor \infnorm{(\identitymatrix - \discountfactor\trans^\envpol)^{-1}} \infnorm{(\trans^{\star} - \trans^\envpol) \valvector^{\star}} 
        && \text{(Submultiplicativity of $\linf$ norm)} \\
        &\leq \discountfactor \left( \sum_{k=0}^\infty \discountfactor^k \infnorm{(\trans^\envpol)^k} \right) \infnorm{(\trans^{\star} - \trans^\envpol) \valvector^{\star}} 
        && \text{(Neumann series, Triangle inequality)} \\
        &= \frac{\discountfactor}{1 - \discountfactor} \infnorm{(\trans^{\star} - \trans^\envpol) \valvector^{\star}}. 
        && \text{(Since $\trans^\envpol$ is stochastic, $\infnorm{(\trans^\envpol)^t} = 1$)}
    \end{align*}
    Let $s \in \states$ be the state that maximizes the row sum in the remaining infinity norm. Applying summation by parts to this specific row yields:
    \begin{align*}
        \infnorm{(\trans^{\star} - \trans^\envpol) \valvector^{\star}} 
        &= (\pvector^{\star}_s - \pvector^\envpol_s)^\top \valvector^{\star} 
        && \text{(Definition of the infinity norm)} \\
        &= \sum_{i=1}^{|\succ(s)|-1} f_\envpol(s, i) 
        && \text{(Summation by parts as in Lemma~\ref{lem:lower_bound})} \\
        &\leq |\succ(s)| \hat{f}_\envpol 
        && \text{(Since $f_\envpol(s, i) \leq \hat{f}_\envpol$)} \\
        &\leq n \hat{f}_\envpol 
        && \text{(Since $|\succ(s)| \leq n$)}
    \end{align*}
    Substituting this bounded row sum back into our norm inequality completes the proof.
\end{proof}

\subsubsection{Halving of the Cumulative Gap}
We now chain the two previous lemmas to show the following Lemma.

\begin{lemma}\label{lem:value_contraction}
    Consider two policies $\envpol$ and $\envpol'$. Let $s \in \states$ and $i \in [|\succ(s)|-1]$ be the state and index that maximize the potential function $f_\envpol(s, i)$. If $\envpol'$ satisfies the condition:
    \begin{equation}
        F^{\star}(s,i) - F^{\envpol'}(s,i) \ge \frac{1}{2}(F^{\star}(s,i) - F^\envpol(s,i)),
    \end{equation}
    Then we have the following inequality:
    \begin{equation}
        \infnorm{\valvector^{\star} - \valvector^{\envpol'}} \ge \frac{1-\discountfactor}{2n} \infnorm{\valvector^{\star} - \valvector^\envpol}.
    \end{equation}
\end{lemma}

\begin{proof}
    We first show how the lemma's assumption translates to the potential function. By rearranging the assumed inequality and multiplying both sides by the non-negative value $2(\valvector^{\star}_{s_i} - \valvector^{\star}_{s_{i+1}})$, we obtain:
    \begin{align*}
        2(F^{\star}(s,i) - F^{\envpol'}(s,i))(\valvector^{\star}_{s_i} - \valvector^{\star}_{s_{i+1}}) &\ge (F^{\star}(s,i) - F^\envpol(s,i))(\valvector^{\star}_{s_i} - \valvector^{\star}_{s_{i+1}}) \\
        \Rightarrow 2 f_{\envpol'}(s, i) &\ge f_\envpol(s, i). && \text{(Definition of $f$)}
    \end{align*}
    
    The main proof follows by chaining the upper bound for policy $\envpol$ (Lemma~\ref{lem:upper_bound}) with the lower bound for policy $\envpol'$ (Lemma~\ref{lem:lower_bound}). We can structure this as a direct sequence of inequalities:
    \begin{align*}
        \infnorm{\valvector^{\star} - \valvector^\envpol} 
        &\le \frac{\discountfactor n}{1-\discountfactor} f_\envpol(s, i) 
        && \text{(Upper bound for $\envpol$, Lemma~\ref{lem:upper_bound})} \\
        &\le \frac{2\discountfactor n}{1-\discountfactor} f_{\envpol'}(s, i) 
        && \text{(Based on Lemma's assumption)} \\
        &\le \frac{2n}{1-\discountfactor} (\valvector^{\star}_s - \valvector^{\envpol'}_s) 
        && \text{(Lower bound for $\envpol'$, Lemma~\ref{lem:lower_bound})} \\
        &\le \frac{2n}{1-\discountfactor} \infnorm{\valvector^{\star} - \valvector^{\envpol'}}. 
        && \text{(Definition of the $L_\infty$ norm)}
    \end{align*}
    Rearranging the first and last terms of this chain directly yields the claimed bound.
\end{proof}

The following lemma establishes that the difference $F^{\star}(s,i) - F^\envpol(s,i)$ must decrease by at least a factor of two for at least one state-index pair $(s,i)$ within a bounded number of iterations, $L := \lceil \log_{\discountfactor} \left( \frac{1 - \discountfactor}{2n} \right) \rceil$, of Algorithm~\ref{algo:RMCPI}.

\begin{lemma} \label{lemma:subaction_not_repeat}
    Let $\envpol^{0}, \envpol^{1}, \dots$ be the sequence of policies generated by Algorithm \ref{algo:RMCPI}.  
    For a specific policy $\envpol^{t}$, let $(s, i) = \argmax_{s' \in \states, j \in [|\succ(s')|-1]} f_{\envpol^{t}}(s', j)$ be the state and index that maximize the potential in $\envpol^{t}$. Then, for any subsequent policy $\envpol^{l}$ with $l > t + L$, we have:
    \begin{equation*}
        F^{\star}(s,i) - F^{\envpol^{l}}(s,i) \leq \frac{1}{2} \left( F^{\star}(s,i) - F^{\envpol^{t}}(s,i) \right).
    \end{equation*}
\end{lemma}

\begin{proof}
    Let $(s, i)$ be the state and index pair defined in the lemma. Assume, for the sake of contradiction, that there exists some $l > t + L$ such that:
    \begin{equation*}
        F^{\star}(s,i) - F^{\envpol^{l}}(s,i) > \frac{1}{2} \left( F^{\star}(s,i) - F^{\envpol^{t}}(s,i) \right).
    \end{equation*}

    This assumption allows us to invoke Lemma~\ref{lem:value_contraction}, giving us a lower bound on the value error at step $l$. Simultaneously, the standard convergence property of policy iteration gives us an upper bound. We combine these bounds as follows:
    \begin{align*}
        \infnorm{\valvector^{\star} - \valvector^{\envpol^{l}}} 
        &\leq \discountfactor^{l - t} \infnorm{\valvector^{\star} - \valvector^{\envpol^{t}}} 
        && \text{(Upper bound via Lemma~\ref{lemma:policy_iteration_exponential_upperbound})} \\
        \infnorm{\valvector^{\star} - \valvector^{\envpol^{l}}} 
        &\geq \frac{1 - \discountfactor}{2n} \infnorm{\valvector^{\star} - \valvector^{\envpol^{t}}} 
        && \text{(Lower bound via Lemma~\ref{lem:value_contraction} and our assumption)}
    \end{align*}
    
    If $\infnorm{\valvector^{\envpol^t} - \valvector^{\star}} = 0$, then $\envpol^t$ is optimal, and the algorithm terminates at iteration $t$. Because the sequence continues to $l > t$, we must have $\infnorm{\valvector^{\envpol^t} - \valvector^{\star}} > 0$. Combining these two inequalities and dividing by this strictly positive term yields:
    \begin{align*}
        \discountfactor^{l - t} &\geq \frac{1 - \discountfactor}{2n} 
        && \text{(Canceling the positive value error norm)} \\
        l - t &\leq \log_{\discountfactor} \left( \frac{1 - \discountfactor}{2n} \right) 
        && \text{(Taking $\log_\discountfactor$, flipping the inequality since $\discountfactor < 1$)}.
    \end{align*}

    This implies $l - t \leq L$, which directly contradicts our initial premise that $l > t + L$. Thus, the assumption must be false, meaning the gap $F^{\star}(s,i) - F^\envpol(s,i)$ must have decreased by at least half.
\end{proof}

\subsubsection{Combinatorial Bound on the Number of Halvings}
The previous Lemma showed that after $L$ steps of Algorithm \ref{algo:RMCPI}, the difference $F^{\star}(s,i) - F^\envpol(s,i)$ decreases by at least a factor of $2$. Now we need to show that this decrease only happens at most a polynomial number of iterations. We can show that in fact this only happens $\mathcal{O}(n \log n)$ times. We thus have the following Lemma, which was shown in \cite{asadi2026strongly}, and we mention it for completeness.

\begin{lemma}\label{lemma:combinatorial_lemma}
    Let $c$ be a constant positive integer and let $W$ be a finite set of non-negative real numbers. We define the set of all bounded-coefficient subset sums of $W$ as:
    $$
        \mathcal{A}(W) = \left\{ \left| \sum_{w \in W} d_w w \right| \;\middle|\; d_w \in \{-c,\cdots, -1,0, 1,\cdots,c\} \right\}.
    $$
    Additionally, for any set of positive reals $Y$, let the degree $\mathrm{Deg}(Y)$ denote the number of unique Maximum Significant bits (MSBs) of the set $Y$ where numbers are written in binary, formally given by:
    $$
        \mathrm{Deg}(Y) = \left| \left\{ \lfloor \log_2 y \rfloor \;\middle|\; y \in Y \right\} \right|.
    $$
    Then, for any such set $W$, the number of distinct MSBs in $\mathcal{A}(W)$ is:
    $$
        \mathrm{Deg}(\mathcal{A}(W) \setminus \{0\}) = \mathcal{O}(|W| \log |W|).
    $$
\end{lemma}

\begin{proof}
We show this using a proof by \cite{asadi2026strongly}; we present the full proof here for completeness. Throughout this proof, all logarithms are taken in base 2. The strategy proceeds in three steps: (i) we introduce a notion of bounded-coefficient linear independence and show that real numbers with sufficiently separated MSBs satisfy it; (ii) we use Siegel's lemma to extract a small-coefficient linear relation among any $\Theta(|W| \log |W|)$ elements of $\mathcal{A}(W)$; (iii) we combine these via a greedy selection on MSBs to derive a contradiction whenever $\mathrm{Deg}(\mathcal{A}(W) \setminus \{0\})$ exceeds the claimed bound.

\paragraph{Step 1: MSB separation implies linear independence.}
For a positive integer $c'$, we call reals $y_1, \dots, y_N$ \emph{$c'$-linearly independent} if $\sum_{i=1}^N x_i y_i \neq 0$ for every non-zero $\underline{x} \in \mathbb{Z}^N$ with $\|\underline{x}\|_\infty \leq c'$.

\begin{claim}\label{clm:lin-ind}
There exists a positive integer $D = D(c')$ such that for any $N \geq 1$, any non-zero reals $y_1, \dots, y_N$ satisfying
$$
\lfloor \log_2 |y_{i+1}| \rfloor \geq \lfloor \log_2 |y_i| \rfloor + D \quad \text{for } 1 \leq i < N
$$
are $c'$-linearly independent.
\end{claim}

\begin{proof}[Proof of Claim~\ref{clm:lin-ind}]
Choose $D$ large enough that $2^{D-1} > 2c'$ (e.g., $D = \lceil \log_2 c' \rceil + 3$). The MSB-separation hypothesis implies $|y_{i+1}|/|y_i| > 2^{D-1} > 2c'$ for every $i$, so in particular $|y_{i+1}| > 2|y_i|$. Iterating yields $|y_{N-1}| > \sum_{i=1}^{N-2} |y_i|$ via a geometric sum, hence
$$
|y_N| > 2c' |y_{N-1}| > c' \sum_{i=1}^{N-1} |y_i|.
$$
Suppose $\sum_{i=1}^N x_i y_i = 0$ with $\|\underline{x}\|_\infty \leq c'$. If $x_N \neq 0$, then $|x_N| \geq 1$ and
$$
|x_N y_N| \geq |y_N| > c'\sum_{i=1}^{N-1} |y_i| \geq \left| \sum_{i=1}^{N-1} x_i y_i \right|,
$$
contradicting $x_N y_N = -\sum_{i<N} x_i y_i$. Hence $x_N = 0$, and induction on $N$ gives $\underline{x} = \underline{0}$.
\end{proof}

\paragraph{Step 2: Siegel's lemma.}
We use the following standard application of Siegel's lemma (see, e.g., \cite{HindrySilverman}). There exist absolute constants $n_0 \geq 4$ and $C_0 > 0$ such that for every integer $n \geq n_0$, every real $c \geq 1$, and every integer matrix $A$ with $n$ columns, $N$ rows satisfying $n \log_2 n - 1 \leq N \leq n \log_2 n$, and $\|A\|_\infty \leq c$, there exists a non-zero $\underline{x} \in \mathbb{Z}^N$ in the left-kernel of $A$ with
$$
\|\underline{x}\|_\infty \leq C_0 \cdot c.
$$
This follows from the classical bound $\|\underline{x}\|_\infty \leq (\|A\|_\infty N)^{n/(N-n)}$. For $n \geq n_0$ sufficiently large, we have $n/(N-n) \leq 1/(\log_2 n - 1 - 1/n) \leq 1$, which justifies the inequality $(cN)^{n/(N-n)} \leq c \cdot N^{n/(N-n)}$ since $c \geq 1$. The remaining factor $N^{n/(N-n)}$ converges to $2$ as $n \to \infty$ (an absolute constant independent of $c$), and is therefore uniformly bounded above by $C_0$ for all $n \geq n_0$. We may also enlarge $n_0$ to ensure $n_0 \log_2 n_0 - 1 \geq n_0$, so that the row count $N$ strictly exceeds the column count $n$ and the kernel is guaranteed to be non-trivial.

\paragraph{Step 3: Combining the steps.}
Set $c' := \lceil C_0 c \rceil$ and let $D := D(c')$ be the constant from Claim~\ref{clm:lin-ind}; both depend only on $c$. Let $K := \mathrm{Deg}(\mathcal{A}(W) \setminus \{0\})$.

\emph{Base case.} If $|W| < n_0$, then $|\mathcal{A}(W)| \leq (2c+1)^{|W|} < (2c+1)^{n_0}$, which is a constant depending only on $c$. In particular, $K$ is bounded by this constant, and the lemma holds trivially. Henceforth assume $|W| \geq n_0$.

\emph{Inductive case.} We show $K \leq D \cdot |W| \log_2 |W|$. Assume for contradiction that $K > D \cdot |W| \log_2 |W|$. Sort the distinct MSBs of $\mathcal{A}(W) \setminus \{0\}$ as $m_1 < m_2 < \cdots < m_K$. Picking every $D$-th MSB yields $\lceil K / D \rceil > |W| \log_2 |W|$ MSBs with consecutive gaps at least $D$. Truncate to $N := \lfloor |W| \log_2 |W| \rfloor$ such MSBs, so $|W| \log_2 |W| - 1 \leq N \leq |W| \log_2 |W|$, and pick a representative $y_j \in \mathcal{A}(W) \setminus \{0\}$ with $\lfloor \log_2 y_j \rfloor$ equal to the $j$-th selected MSB. By Claim~\ref{clm:lin-ind}, $y_1, \dots, y_N$ are $c'$-linearly independent.

On the other hand, by definition of $\mathcal{A}(W)$, each $y_j$ admits a representation
$$
y_j = \left| \sum_{w \in W} d_{j,w} w \right| \quad \text{with } d_{j,w} \in \{-c, \dots, c\}.
$$
Letting $\sigma_j \in \{-1, +1\}$ be such that $y_j = \sigma_j \sum_w d_{j,w} w$, and setting $\tilde{d}_{j,w} := \sigma_j d_{j,w}$, we still have $|\tilde{d}_{j,w}| \leq c$ and $y_j = \sum_{w \in W} \tilde{d}_{j,w} w$. Let $A \in \mathbb{Z}^{N \times |W|}$ be the matrix with entries $A_{j,w} = \tilde{d}_{j,w}$; then $\|A\|_\infty \leq c$. Since $|W| \geq n_0$ and $|W| \log_2 |W| - 1 \leq N \leq |W| \log_2 |W|$, Step 2 yields a non-zero $\underline{x} \in \mathbb{Z}^N$ with $\|\underline{x}\|_\infty \leq C_0 c \leq c'$ such that $A^\top \underline{x} = 0$, i.e., $\sum_{j=1}^N x_j \tilde{d}_{j,w} = 0$ for every $w \in W$. Consequently,
$$
\sum_{j=1}^N x_j y_j = \sum_{j=1}^N x_j \sum_{w \in W} \tilde{d}_{j,w} w = \sum_{w \in W} w \sum_{j=1}^N x_j \tilde{d}_{j,w} = 0,
$$
contradicting the $c'$-linear independence of $y_1, \dots, y_N$. Therefore $K \leq D \cdot |W| \log_2 |W| = \mathcal{O}(|W| \log |W|)$.
\end{proof}

Finally, we will show that Algorithm \ref{algo:RMCPI} is strongly polynomial for $L_1$ and $L_\infty$ uncertainty sets. We prove \cref{thm:strongly_poly} by showing it separately for $L_1$ and $L_\infty$ uncertainty sets.

\begin{proof}[Proof of \Cref{thm:strongly_poly}~ for $L_1$ uncertainty sets]
    For any state $s \in \states$, Property~\ref{prop:policy_structure} describes the structure of the transition probabilities generated by Algorithm~\ref{alg:lone_max}. Probability mass is shifted to the highest-value state by fully taking mass from a set of states $Z_\pvector$ and partially taking mass from at most one state $I_\pvector$. The mass taken from $I_\pvector$ is the leftover budget: $\frac{\radius_s}{2} - \sum_{z \in Z_\pvector} \nominal_{s,z}$. 
    
    When we calculate the cumulative probability $F^\envpol(s,i)$ over any prefix set $S_i$, the sum simplifies algebraically. If $I_\pvector \in S_i$, the $-\tfrac{\radius_s}{2}$ from the partial term cancels the receiver's $+\tfrac{\radius_s}{2}$, while the remaining $\sum_{z \in Z_\pvector} \nominal_{s,z}$ adds each depleted state's nominal with coefficient $+1$ (regardless of whether $z \in S_i$). This ensures that $F^\envpol(s,i)$ is always a linear combination of the nominal probabilities $\nominal_{s,j}$ and the budget $\frac{\radius_s}{2}$, using only the coefficients $\{-2,-1, 0, 1,2\}$. 
    
    Therefore, the difference $F^{\star}(s,i) - F^\envpol(s,i)$ strictly belongs to the set of unitary signed subset sums $\mathcal{A}(W_s)$, which is defined over the base set $W_s = \left\{ \nominal_{s,j} \mid j \in \states \right\} \cup \left\{ \frac{\radius_s}{2}, 1 \right\}$.

    By Lemma~\ref{lemma:combinatorial_lemma}, the number of distinct MSBs for the non-zero elements in $\mathcal{A}(W_s)$ is bounded by $\mathcal{O}(|W_s| \log |W_s|)$. Because the size of the base set is $|W_s| \leq n + 2$, the total number of possible MSBs is bounded by $\mathcal{O}(n \log n)$.

    From Lemma~\ref{lemma:subaction_not_repeat}, we know that every $L = \lceil \log_{\discountfactor} \left( \frac{1 - \discountfactor}{2n} \right) \rceil$ iterations, the maximum cumulative probability gap $F^{\star}(s,i) - F^\envpol(s,i)$ must decrease by at least half. Since the optimal values are fixed, this halving directly corresponds to dropping to a lower MSB within the set $\mathcal{A}(W_s)$. 
    
    There are at most $n^2$ state-index pairs $(s,i)$, and for any specific pair the gap $F^{\star}(s,i)-F^\envpol(s,i)$ can drop its MSB at most $\mathcal{O}(n\log n)$ times. Combining with the $L$ iterations per drop yields $n^2\cdot\mathcal{O}(n\log n)\cdot L=\mathcal{O}(n^3\log n\cdot L)$.
\end{proof}

\begin{proof}[Proof of \Cref{thm:strongly_poly}~ for $L_\infty$ uncertainty sets]
    For any state $s \in \states$, Property~\ref{prop:linf_policy_structure} describes the structure of the transition probabilities generated by Algorithm~\ref{alg:linf_max}. Probability mass is moved between states under strict $L_\infty$ constraints. As a result, the final probability $\pvector_{s, s'}$ for any successor state is always formed by adding or subtracting the budget $\radius_s$ from the nominal probabilities $\nominal_{s,s'}$, or by capping at the limits $0$ or $1$.

    Therefore, the difference between the optimal probability and the current policy's probability, $\pvector^{\star}_{s,s'} - \pvector^\envpol_{s,s'}$, strictly belongs to the set of signed subset sums $\mathcal{A}(X_s)$ with coefficients in $\{-2,-1,0,1,2\}$. This set is defined over the base set $X_s = \left\{ \nominal_{s,j} \mid j \in \states \right\} \cup \left\{ 1 \right\} \cup \left\{ \radius_s, 2 \cdot \radius_s, \cdots, n \cdot \radius_s \right\}$.

    The multiples $2\radius_s,\dots,n\radius_s$ are included because the at-most-one incomplete coordinate, by conservation $\sum_{s'} \pvector^\envpol_{s,s'} = 1$, contributes a coefficient $|D_\pvector|-|R_\pvector| \in \{-n,\dots,n\}$ on $\radius_s$. Hence the cumulative difference $F^{\star}(s,i)-F^\envpol(s,i)$ may carry a coefficient up to $\pm 2n$ on $\radius_s$, which we encode using at most two base elements $k\cdot\radius_s$ with $k \leq n$, each with a coefficient in $\{-2,\dots,2\}$. This keeps the coefficient bound in $\mathcal{A}(X_s)$ a constant, as Lemma~\ref{lemma:combinatorial_lemma} requires.

    By Lemma~\ref{lemma:combinatorial_lemma}, the number of distinct MSBs for the non-zero elements in $\mathcal{A}(X_s)$ is bounded by $\mathcal{O}(|X_s| \log |X_s|)$. Because the size of the base set is $|X_s| \leq 2n + 1$, the total number of possible MSBs is strictly bounded by $\mathcal{O}(n \log n)$.

    From our convergence analysis, we know that every $L = \lceil \log_{\discountfactor} \left( \frac{1 - \discountfactor}{2n} \right) \rceil$ iterations, the maximum cumulative probability gap $F^{\star}(s,i) - F^\envpol(s,i)$ must decrease by at least half. Since the optimal values remain fixed, this halving directly corresponds to dropping to a lower MSB within the set $\mathcal{A}(X_s)$. 
    
    There are at most $n^2$ possible state-index pairs $(s, i)$ that can dictate this maximum cumulative probability gap. For any specific pair, the associated probability discrepancy can drop its MSB at most $\mathcal{O}(n \log n)$ times. Because it takes at most $L$ steps to force a drop, the total number of policy iteration steps before termination is bounded by $n^2 \cdot \mathcal{O}(n \log n) \cdot L = \mathcal{O}(n^3 \log n \cdot L)$. This represents a polynomial number of iterations, completing the proof.
\end{proof}

\section{Proof of \texorpdfstring{\cref{thm:lp-comp-class}}{Theorem 3}: Complexity of RMC with \texorpdfstring{$L_p$}{Lp} Uncertainty Sets} \label{app:thm3}

In this section, we analyze the complexity of the Problem~$\discrmc$ with $L_p$ uncertainty sets for any constant integer $p > 1$. 
Since the goal of this section is to show the complexity bounds for this problem, we consider the Turing model of computation, where all the numbers in the input are assumed to be rational. 

First, we analyze the $p$-ROOT-SUM complexity class. We prove that this class is polynomially equivalent to the class of $\frac{p}{q}$-ROOT-SUM where $\gcd(p,q)=1$. 
Second, we construct a polynomial-time reduction from the $\frac{p}{p-1}$-ROOT-SUM problem to the decision version of the main problem, which shows that it is $p$-ROOT-SUM-hard.
Finally, we establish an upper bound on the complexity of the problem. We show that it belongs to the class $\etr$. We achieve this by modeling the main problem (computing the exact values and optimal strategy) as an $\etr$ formula. The corresponding decision problem can then be easily solved by adding a threshold constraint to this formula.

\subsection{Equivalence of ROOT-SUM Problems}

Our main objective in this subsection is to prove that the $p$-ROOT-SUM problem and the $\frac{p}{q}$-ROOT-SUM problem (where $\gcd(p,q)=1$) are polynomial-time equivalent.

First, we show that the $p$-ROOT-SUM problem is reducible to the $kp$-ROOT-SUM problem. Next, we introduce a polynomial-time greedy decomposition algorithm that expresses any integer as a sum of $p$-th powers using a logarithmically bounded number of terms. Then, we construct a formal reduction from the $p$-ROOT-SUM problem to the $\frac{p}{q}$-ROOT-SUM problem. Finally, \Cref{lem:p-p/q-complexity} combines these steps to conclude the polynomial-time equivalence of these classes.

In the following lemma, we show that the $kp$-ROOT-SUM problem is harder than $p$-ROOT-SUM problem for any rational number $p$ and any constant integer $k > 1$:

\begin{lemma}\label{lem:kp_reduction}
For any rational number $p$ and any constant integer $k > 1$, the $p$-ROOT-SUM problem is polynomial-time reducible to the $kp$-ROOT-SUM problem.
\end{lemma}

\begin{proof}
Let an instance of the $p$-ROOT-SUM problem be given by a sequence of positive integers $a_1, \dots, a_n$ and an integer threshold $\Threshold$. We need to decide if the following holds:
\[ \sum_{i=1}^n \sqrt[p]{a_i} \ge \Threshold \]

We construct an instance of the $kp$-ROOT-SUM problem. For each index $i$, we define a new value $a'_i = a_i^k$. Because $a_i$ and $k$ are integers, $a'_i$ is also an integer. 

By the properties of exponents, we have:
\[ \sqrt[kp]{a'_i} = (a'_i)^{\frac{1}{kp}} = (a_i^k)^{\frac{1}{kp}} = a_i^{\frac{k}{kp}} = a_i^{\frac{1}{p}} = \sqrt[p]{a_i} 
\Rightarrow
\sum_{i=1}^n \sqrt[kp]{a'_i} = \sum_{i=1}^n \sqrt[p]{a_i} \]

Therefore, the condition $\sum_{i=1}^n \sqrt[kp]{a'_i} \ge \Threshold$ is true if and only if $\sum_{i=1}^n \sqrt[p]{a_i} \ge \Threshold$. 

Because $k$ is a constant, computing $a_i^k$ for all $n$ terms takes polynomial time and has a polynomial number of bits. Therefore, the reduction is polynomial.
\end{proof}

Now we need an efficient way to write integers as a sum of powers of $p$. The following lemma provides a polynomial-time algorithm to express any integer as a sum of $p$-th powers.

\begin{lemma}\label{lem:greedy_powers}
    Let $p$ be a constant positive integer greater than zero.
    Any integer $n$ can be written as a sum $n = \sum_{i=1}^{k} a_i^p$ for some positive integers $a_i$, 
    where the number of terms $k = O(\log \log n)$. Furthermore, the sequence of integers $a_i$ can be computed in time polynomial in the bit-size of $n$.
\end{lemma}

\begin{proof}

We construct the sequence $a_i$ using the greedy approach by \cref{alg:greedy_powers}.

    \begin{algorithm}[ht]
    \caption{Greedy $p$-th Power Decomposition}
    \label{alg:greedy_powers}
    \begin{algorithmic}[1]
    \State \textbf{Input:} Integer $n > 0$, Integer $p > 1$
    \State \textbf{Output:} Sequence $(a_1, a_2, \dots, a_k)$ such that $\sum_{i=1}^k a_i^p = n$
    \State $n_0 \gets n$
    \State $j \gets 0$
    \While{$n_j > 0$}
        \State $a_{j+1} \gets \lfloor n_j^{1/p} \rfloor$ \Comment{Find the largest integer root}
        \State $n_{j+1} \gets n_j - a_{j+1}^p$ \Comment{Update the remainder}
        \State $j \gets j + 1$
    \EndWhile
    \State $k \gets j$
    \State \Return $(a_1, a_2, \dots, a_k)$
    \end{algorithmic}
    \end{algorithm}

    Let $n_0 = n$. For each step $j \ge 0$, if $n_j > 0$, the algorithm chooses the largest integer $a_{j+1} = \lfloor n_j^{1/p} \rfloor$ such that $a_{j+1}^p \le n_j$, and updates the remainder to $n_{j+1} = n_j - a_{j+1}^p$. We bound $n_{j+1}$ as follows:
    \begin{align*}
        n_{j+1} &= n_j - a_{j+1}^p 
        && \text{(Definition of the $n_{j+1}$)} \\
        &< (a_{j+1} + 1)^p - a_{j+1}^p 
        && \text{(Since $(a_{j+1} + 1)^p > n_j \geq a_{j+1}^p$)} \\
        &= \sum_{\ell=0}^{p-1} \binom{p}{\ell} a_{j+1}^\ell 
        && \text{(Binomial expansion)} \\
        &\le \left( \sum_{\ell=0}^{p-1} \binom{p}{\ell} \right) a_{j+1}^{p-1} 
        && \text{(Since $a_{j+1} \ge 1$, we have $a_{j+1}^\ell \le a_{j+1}^{p-1}$)} \\
        &< 2^p a_{j+1}^{p-1} 
        && \text{(Sum of binomial coefficients)} \\
        &\le 2^p n_j^{\frac{p-1}{p}} 
        && \text{(Since $a_{j+1} \le n_j^{\frac{1}{p}}$)} \\
        \Rightarrow \log_2 n_{j+1} &\le p + \frac{p-1}{p} \log_2 n_j.
        && \text{(Properties of logarithms)}
        \\
        \Rightarrow \log_2 n_{j+1} - p^2 &\le \frac{p-1}{p} (\log_2 n_j - p^2).
        && \text{(Subtracting $p^2$ from both sides)}
    \end{align*}

    This shows that $\log_2 n_j - p^2$ decreases geometrically by the constant factor $0 < \frac{p-1}{p} < 1$. Therefore, the number of steps required for $\log_2 n_j$ to drop below the constant $p^2+1$ is $O(\log\log n_j)$.
    
    Once $\log_2 n_j \le p^2+1$, we have $n_j \le 2^{p^2+1}$. Since $n_j$ is an integer that decreases by at least $1$ in each step, the algorithm will reach $0$ in at most $2^{p^2+1}$ steps. Since $p$ is a constant, this takes $O(1)$ steps. Because $n_0 = n$, the total number of terms $k = O(\log \log n) + O(1) = O(\log \log n)$.

    Finally, at each step $j$, computing the largest integer root $a_{j+1} = \lfloor n_j^{1/p} \rfloor$ takes polynomial time with respect to the bit-size of $n_j$. Because there are $O(\log \log n)$ steps, the entire sequence $a_i$ is computed in polynomial time.
\end{proof}

Next, we show that we can change the root from $p$ to $p/q$ without changing the complexity of the problem. The following lemma provides a polynomial-time reduction for this step.

\begin{lemma}\label{lem:pq_reduction}
For any constant positive integers $p$ and $q$ such that $\gcd(p,q)=1$, the $p$-ROOT-SUM problem is polynomial-time reducible to the $\frac{p}{q}$-ROOT-SUM problem.
\end{lemma}

\begin{proof}

Because $\gcd(p,q)=1$, modular arithmetic guarantees there exists a positive integer $k < q$ such that $q$ divides $kp+1$. Because $k < q$, we can treat $k$ as a constant.

Let an instance of the $p$-ROOT-SUM problem be given by a sequence of positive integers $a_1, \dots, a_n$ and an integer threshold $\Threshold$. We want to decide if:
$$ \sum_{i=1}^n a_i^{\frac{1}{p}} \ge \Threshold $$

We multiply both sides of the inequality by $w := \prod_{j=1}^n a_j^k$.
We rewrite the new term $w \cdot a_i^{\frac{1}{p}}$ as follows:

\begin{align*}
        w \cdot a_i^{\frac{1}{p}} &= z_i \cdot a_i^k \cdot a_i^{\frac{1}{p}} 
        && \text{(Factoring out $z_i:=\prod_{j \neq i} a_j^k$ from $w$)} \\
        &= z_i \cdot a_i^{\frac{kp+1}{p}} 
        && \text{(Combining the exponents of $a_i$)} \\
        &= z_i \left( a_i^{\frac{kp+1}{q}} \right)^{\frac{q}{p}} 
        && \text{(Properties of exponents)} \\
        &= z_i y_i^{\frac{q}{p}},
        && \text{(Substituting integer $y_i := a_i^{\frac{kp+1}{q}}$)}
    \end{align*}
    where $y_i$ is an integer because $q$ divides $kp+1$.

    By \cref{lem:greedy_powers}, we can decompose each integer $z_i$ into a sum of $q$-th powers of integers $c_{i,j}$. This gives $z_i = \sum_{j=1}^{m_i} c_{i,j}^q$, where the number of terms $m_i$ is in $O(\log \log z_i)$. We define $b_{i,j} := c_{i,j}^p y_i$ and simplify the main inequality:
    \begin{align*}
        \sum_{i=1}^n w \cdot a_i^{\frac{1}{p}} \ge w\Threshold 
        &\iff \sum_{i=1}^n z_i y_i^{\frac{q}{p}} \ge w\Threshold 
        && \text{(Substituting the term equivalence)} \\
        &\iff \sum_{i=1}^n \left( \sum_{j=1}^{m_i} c_{i,j}^q \right) y_i^{\frac{q}{p}} \ge w\Threshold 
        && \text{(Applying Lemma~\ref{lem:greedy_powers} to decompose $z_i$)} \\
        &\iff \sum_{i=1}^n \sum_{j=1}^{m_i} \left( c_{i,j}^p y_i \right)^{\frac{q}{p}} \ge w\Threshold 
        && \text{(Distributing $y_i^{\frac{q}{p}}$ inside the inner sum)} \\
        &\iff \sum_{i=1}^n \sum_{j=1}^{m_i} b_{i,j}^{\frac{q}{p}} \ge w\Threshold.
        && \text{(Substituting $b_{i,j}$)}
    \end{align*}

    This final inequality is an instance of the $\frac{p}{q}$-ROOT-SUM problem with the threshold $w\Threshold$. 

    Because $k$ is a constant, computing and storing $w, z_i, y_i$ takes polynomial time and memory in the bit-size of $a_i$. Furthermore, the number of new terms is $O(n \log \log z_i)$. By \cref{lem:greedy_powers} $c_{i,j}$ is computed in polynomial time. Finally, $b_{i,j}$ can be computed by the arithmetic operations. Therefore, this reduction is valid and runs in polynomial time.
\end{proof}

By combining the previous reductions, we can establish the equivalence result for this subsection. This shows that the denominator in the fractional root in the $p$-ROOT-SUM problem does not affect the complexity.

\begin{lemma} \label{lem:p-p/q-complexity}
For any constant positive integers $p$ and $q$ such that $\gcd(p,q)=1$, the $p$-ROOT-SUM problem and the $\frac{p}{q}$-ROOT-SUM problem are polynomial-time equivalent. In other words, the denominator does not matter in the complexity.
\end{lemma}

\begin{proof}
By Lemma~\ref{lem:kp_reduction}, the $\frac{p}{q}$-ROOT-SUM problem is polynomial-time reducible to the $p$-ROOT-SUM problem.
Similarly, because $\gcd(p,q)=1$, by \cref{lem:pq_reduction} the $p$-ROOT-SUM problem is polynomial-time reducible to the $\frac{p}{q}$-ROOT-SUM problem. Therefore, they are polynomial-time equivalent.
\end{proof}

\subsection{Proving Hardness via Reduction}

By \Cref{lem:p-p/q-complexity}, the $p$-ROOT-SUM problem is polynomial-time equivalent to the $\frac{p}{p-1}$-ROOT-SUM problem. We now formalize the reduction from the $\frac{p}{p-1}$-ROOT-SUM problem to the decision version of the problem $\discrmc$ with $\Lp$ uncertainty sets. Formally, we do the following reduction:

Given an integer constant $p \ge 2$, positive integers $a_1, \dots, a_n$, and an integer threshold $\Threshold$, the $\frac{p}{p-1}$-ROOT-SUM problem asks to decide if $\sum_{i=1}^n a_i^{\frac{p-1}{p}} \ge \Threshold$. We reduce this problem to deciding if the optimal value $\optval(s_0) \ge \lambda$ in an RMC $\rmc = (\states, \cost, \succ, \uncert)$ with $\Lp$ uncertainty sets and a constant discount factor $\discount \in (0, 1)$.

\textbf{Reduction:}
First, we introduce the following scaled variables to construct the RMC:
\begin{itemize}
    \item $\forall i \in [n]: b_i := 2^p a_i $.
    \item $\forall i \in [n]:x_i := b_i / 2 = 2^{p-1} a_i$.
    \item $K := 2^{p-1}\Threshold$.
\end{itemize}

By Lemma~\ref{lem:greedy_powers}, we write $x_i$ as the sum of $p$-th powers of positive integers: $x_i = \sum_{k=1}^{M_i} u_{i,k}^p$. 
Let $M^* = \max_i M_i$. We pad each sequence $u_{i,k}$ with zeros so each sequence has exactly $M^*$ terms. This preserves the sum.
By \cref{lem:greedy_powers}, we know that $M^* = O(\log \log a_i) = O(\log \log b_i)$.

We construct the RMC $\rmc = (\states, \cost, \succ, \uncert)$ as follows:

\begin{itemize}
\item $\states := \{s_0\} \cup \{s_1, \dots, s_n\} \cup \bigcup_{i=1}^n \{\absorbing_{i,1}^+, \dots, \absorbing_{i,M^*}^+, \absorbing_{i,1}^-, \dots, \absorbing_{i,M^*}^-\}$.
The states consist of an initial state $s_0$, $n$ transient states $s_i$, and $2M^*$ absorbing states for each transient state.

\item $\cost \colon \states \to \R$:
\[
\cost(s) := \begin{cases} 
    (1-\discount)u_{i,k}^{p-1} & \text{if } s = \absorbing_{i,k}^+, \\ 
    -(1-\discount)u_{i,k}^{p-1} & \text{if } s = \absorbing_{i,k}^-, \\ 
    0 & \text{if } s \in \{s_0, s_1, \dots, s_n\}.
\end{cases}
\]

\item $\succ \colon \states \to 2^\states$ and $\uncert \colon \states \to 2^{\Delta(\states)}$. \\
The successor mapping and $\Lp$ uncertainty sets (defined by a nominal probability $\nominal_s$ and radius $\radius_s$) are defined for each state class:
\begin{itemize}
    \item \textit{Initial State:} $\succ(s_0) := \{s_1, \dots, s_n\}$. The nominal distribution $\nominal_{s_0}$ is uniform over the transient states ($\nominal_{s_0}(s_i) = \frac{1}{n}$) with zero uncertainty ($\radius_{s_0} := 0$).
    
    \item \textit{Transient States:} $\succ(s_i) := \{\absorbing_{i,1}^+, \dots, \absorbing_{i,M^*}^+, \absorbing_{i,1}^-, \dots, \absorbing_{i,M^*}^-\}$. From $s_i$ to its $2M^*$ absorbing states, the nominal probability $\nominal_{s_i}$ is uniform ($1/2M^*$). We define the fixed uncertainty radius for these states as $\radius := \frac{1}{2M^*}$, so $\radius_{s_i} := \radius$.
    
    \item \textit{Absorbing States:} $\succ(\absorbing_{i,k}^+) := \{\absorbing_{i,k}^+\}, \succ(\absorbing_{i,k}^-) := \{\absorbing_{i,k}^-\}$. Each absorbing state deterministically transitions to itself, with zero uncertainty ($\radius_{\absorbing_{i,k}^+} := 0, \radius_{\absorbing_{i,k}^-} := 0)$.
\end{itemize}

\item $\lambda := \frac{\discount^2 \radius K}{n}$.
\end{itemize}

Now, we analyze the values of the states in RMC $\rmc$. We first calculate the values for the absorbing states.

\begin{lemma}\label{lem:absorbing_values}
    For the absorbing states $\absorbing_{i,k}^+$ and $\absorbing_{i,k}^-$, their values are $u_{i,k}^{p-1}$ and $-u_{i,k}^{p-1}$, respectively.
\end{lemma}

\begin{proof}
    We first prove the claim for the positive absorbing state $\absorbing_{i,k}^+$. By construction, the absorbing state deterministically transitions to itself because its uncertainty radius is zero. Thus, the Bellman equation simplifies to:
    
    \begin{align*}
        \optval(\absorbing_{i,k}^+) &= \cost(\absorbing_{i,k}^+) + \discount \optval(\absorbing_{i,k}^+)
        && \text{(Bellman equation for absorbing states)} \\
        \implies \optval(\absorbing_{i,k}^+) &= \frac{\cost(\absorbing_{i,k}^+)}{1-\discount}
        && \text{(Rearranging for $\optval$)} \\
        &= \frac{(1-\discount)u_{i,k}^{p-1}}{1-\discount}
        && \text{(Substituting the defined cost)} \\
        &= u_{i,k}^{p-1}.
        && \text{(Canceling the $1-\discount$ terms)}
    \end{align*}

    By symmetry, substituting the cost $\cost(\absorbing_{i,k}^-) = -(1-\discount)u_{i,k}^{p-1}$ for the negative absorbing state into the same derivation yields $\optval(\absorbing_{i,k}^-) = -u_{i,k}^{p-1}$.
\end{proof}

Using the values of the absorbing states, we can determine the value of the transient states. The following lemma calculates this value.

\begin{lemma}\label{lem:holder_lp}
 For the transient state $s_i$, let $\val^{(i)}$ be the vector of values of its absorbing states. Then, $\optval(s_i) = \discount \radius \Vert \val^{(i)} \Vert_{\frac{p}{p-1}}$.
\end{lemma}
\begin{proof}
    The Bellman equation for $s_i$ yields:
\begin{align*}
    \optval(s_i) &= \cost(s_i) + \discount \max_{\envpol \in \uncert(s_i)} \envpol^\top \val^{(i)}
    && \text{(Definition of the Bellman equation)} \\
    &= \discount \max_{\envpol \in \uncert(s_i)} \envpol^\top \val^{(i)}
    && (\cost(s_i)=0) \\
    &= \discount \max_{\substack{d \in \mathbb{R}^{2M^*}\\\Vert d \Vert_p \le \radius\\\mathbf{1}^\top d = 0\\\nominal_{s_i} + d \ge 0}} \ (\nominal_{s_i} + d)^\top \val^{(i)} 
    && \text{($\uncert(s_i)$ is an $\Lp$ uncertainty set)} \\
    &= \discount \max_{\substack{d \in \mathbb{R}^{2M^*}\\\Vert d \Vert_p \le \radius\\\mathbf{1}^\top d = 0\\\nominal_{s_i} + d \ge 0}} \ d^\top \val^{(i)} 
    && \text{($\nominal_{s_i}^\top \val^{(i)} = 0$ as $\nominal_{s_i}$ is uniform and $\optval(\absorbing_{i,k}^+) + \optval(\absorbing_{i,k}^-) = 0$)} 
\end{align*}

Because $\frac{p-1}{p} + \frac{1}{p} = 1$, Hölder's inequality states that the inner product is bounded from above \cite{folland1999real}:
    \begin{align*}
        d^\top \val^{(i)} &\le \Vert d \Vert_p \Vert \val^{(i)} \Vert_{\frac{p}{p-1}} \le \radius \Vert \val^{(i)} \Vert_{\frac{p}{p-1}}.
    \end{align*}

To prove this upper bound is achievable, we construct the optimal perturbation $d^{\star}$:
$$d^{\star}_j = \radius \frac{\operatorname{sign}(\val^{(i)}_j) |\val^{(i)}_j|^{\frac{1}{p-1}}}{\Vert \val^{(i)} \Vert_{\frac{p}{p-1}}^{\frac{1}{p-1}}}$$

This vector satisfies all constraints:
\begin{itemize}
    \item ($\Vert d^{\star} \Vert_p = \radius$): By the definition of the $p$ norm we have:
    $$\Vert d^{\star} \Vert_p^p = \sum_j |d^{\star}_j|^p = \frac{\radius^p}{\Vert \val^{(i)} \Vert_{\frac{p}{p-1}}^{\frac{p}{p-1}}} \sum_j |\val^{(i)}_j|^{\frac{p}{p-1}}$$
    By the definition of the $L_{\frac{p}{p-1}}$ norm, the sum is equal to the denominator. Thus, $\Vert d^{\star} \Vert_p^p = \radius^p$. Therefore, $\Vert d^{\star} \Vert_p = \radius$,
    \item ($\mathbf{1}^\top d^{\star} = 0$): Since $|\val^{(i)}_j|^{\frac{1}{p-1}} = (u_{i,k}^{p-1})^{\frac{1}{p-1}} = u_{i,k}$, and $\val^{(i)}$ is symmetric, the components of $d^{\star}$ form canceling pairs $\pm u_{i,k}$.
    \item $\nominal_{s_i} + d^{\star} \ge 0$: Since $\Vert d^{\star} \Vert_p = \radius$, and $\infnorm{d^{\star}} \leq \Vert d^{\star} \Vert_p$, the maximum negative perturbation is at least $-\radius$. Because $\nominal_{s_i,j} = \frac{1}{2M^*} = \radius$, the minimum possible probability is $0$.
\end{itemize}

Finally, we evaluate the inner product to confirm the maximum bound:
\begin{align*}
    d^{\star\top} \val^{(i)} &= \sum_{j=1}^{2M^*} \left( \radius \frac{\operatorname{sign}(\val^{(i)}_j) |\val^{(i)}_j|^{\frac{1}{p-1}}}{\Vert \val^{(i)} \Vert_{\frac{p}{p-1}}^{\frac{1}{p-1}}} \right) \val^{(i)}_j 
    && \text{(Expanding the dot product)} \\
    &= \frac{\radius}{\Vert \val^{(i)} \Vert_{\frac{p}{p-1}}^{\frac{1}{p-1}}} \sum_{j=1}^{2M^*} |\val^{(i)}_j|^{\frac{p}{p-1}} 
    && \text{($\operatorname{sign}(x) \cdot x = |x|$ and $|x|^{\frac{1}{p-1}} |x| = |x|^{\frac{p}{p-1}}$)} \\
    &= \frac{\radius}{\Vert \val^{(i)} \Vert_{\frac{p}{p-1}}^{\frac{1}{p-1}}} \Vert \val^{(i)} \Vert_{\frac{p}{p-1}}^{\frac{p}{p-1}} 
    && \text{(Definition of the $L_{\frac{p}{p-1}}$ norm)} \\
    &= \radius \Vert \val^{(i)} \Vert_{\frac{p}{p-1}}.
    && \text{(Simplifying exponents)}
\end{align*}

Substituting this back yields $\optval(s_i) = \discount \radius \Vert \val^{(i)} \Vert_{\frac{p}{p-1}}$, completing the proof.
\end{proof}

To simplify the expression for the transient state values, we must evaluate the norm of the absorbing state values. The following lemma calculates this value.

\begin{lemma}\label{lem:reward_norm_q}
    For any transient state $s_i$, $L_{\frac{p}{p-1}}$ norm of its absorbing state values is $\Vert \val^{(i)} \Vert_{\frac{p}{p-1}} = b_i^{\frac{p-1}{p}}$.
\end{lemma}

\begin{proof}
    We evaluate the $L_{\frac{p}{p-1}}$ norm of the value vector $\val^{(i)}$ directly:
    \begin{align*}
        \Vert \val^{(i)} \Vert_{\frac{p}{p-1}} &= \left( \sum_{j=1}^{2M^*} |\val^{(i)}_j|^{\frac{p}{p-1}} \right)^{\frac{p-1}{p}} 
        && \text{(Definition of the $L_{\frac{p}{p-1}}$ norm)} \\
        &= \left( 2 \sum_{k=1}^{M^*} \left( u_{i,k}^{p-1} \right)^{\frac{p}{p-1}} \right)^{\frac{p-1}{p}} 
        && \text{(Grouping symmetric pairs $\absorbing_{i,k}^+$ and $\absorbing_{i,k}^-$)} \\
        &= \left( 2 \sum_{k=1}^{M^*} u_{i,k}^p \right)^{\frac{p-1}{p}} 
        && \text{(Simplifying the exponent: $(p-1) \cdot \frac{p}{p-1} = p$)} \\
        &= (2 x_i)^{\frac{p-1}{p}} 
        && \text{(Substituting the construction $x_i = \sum_{k=1}^{M^*} u_{i,k}^p$)} \\
        &= b_i^{\frac{p-1}{p}}. 
        && \text{(Substituting the definition $b_i = 2 x_i$)}
    \end{align*}
    This completes the proof.
\end{proof}

\begin{proof}[Proof of \Cref{thm:lp-comp-class}-Item~\ref{item:p-root-sum-hardness}]
    By \cref{lem:p-p/q-complexity}, and because $\gcd(p-1,p)=1$, the $\frac{p}{p-1}$-ROOT-SUM problem is as hard as the $p$-ROOT-SUM problem. Therefore, establishing that the problem $\discrmc$ is $\frac{p}{p-1}$-ROOT-SUM-hard proves that it is also $p$-ROOT-SUM-hard.
    
    We evaluate the Bellman equation at the initial state $s_0$:
    \begin{align*}
        \optval(s_0) &= \cost(s_0) + \discount \sum_{i=1}^n \nominal_{s_0}(s_i) \optval(s_i)
        && \text{(Because $\radius_{s_0}=0)$} \\
        &= 0 + \discount \sum_{i=1}^n \frac{1}{n} \optval(s_i) 
        && \text{(Substituting $\cost(s_0) = 0$ and $\nominal_{s_0}(s_i) = \frac{1}{n}$)} \\
        &= \frac{\discount}{n} \sum_{i=1}^n \left( \discount \radius b_i^{\frac{p-1}{p}} \right) 
        && \text{(Substituting results from Lemmas~\ref{lem:holder_lp} and \ref{lem:reward_norm_q})} \\
        &= \frac{\discount^2 \radius}{n} \sum_{i=1}^n b_i^{\frac{p-1}{p}}.
    \end{align*}

    We want to decide if $\optval(s_0) \ge \lambda$. Substituting our target threshold $\lambda = \frac{\discount^2 \radius K}{n}$ gives:
    \begin{align*}
        \optval(s_0) \ge \lambda &\iff \frac{\discount^2 \radius}{n} \sum_{i=1}^n b_i^{\frac{p-1}{p}} \ge \frac{\discount^2 \radius K}{n} \\
        &\iff \sum_{i=1}^n b_i^{\frac{p-1}{p}} \ge K.
        && \text{(Dividing both sides by the positive constant $\frac{\discount^2 \radius}{n}$)}
    \end{align*}

    Recall our scaling definitions $b_i := 2^p a_i$ and $K := 2^{p-1}\Threshold$. We substitute these into the inequality:
    \begin{align*}
        \sum_{i=1}^n (2^p a_i)^{\frac{p-1}{p}} \ge 2^{p-1}\Threshold 
        &\iff \sum_{i=1}^n 2^{p-1} a_i^{\frac{p-1}{p}} \ge 2^{p-1}\Threshold 
        && \text{($(2^p)^{\frac{p-1}{p}} = 2^{p-1}$)} \\
        &\iff 2^{p-1} \sum_{i=1}^n a_i^{\frac{p-1}{p}} \ge 2^{p-1}\Threshold 
        && \text{(Factoring out the constant $2^{p-1}$)} \\
        &\iff \sum_{i=1}^n a_i^{\frac{p-1}{p}} \ge \Threshold.
        && \text{(Dividing by $2^{p-1}$)}
    \end{align*}

    Because all numbers are bounded polynomially and Lemma~\ref{lem:greedy_powers} guarantees the gadget size $M^*$ is bounded, the reduction is valid and runs in polynomial time.
\end{proof}

\subsection{ETR Upper Bound}

In this section, we establish that the decision version of the $\discrmc$ problem with $L_p$ uncertainty sets belongs to the complexity class $\etr$ for any integer $p > 1$. Furthermore, we present an approach to solve the main computational problem by calling an $\etr$ solver.

The main $\discrmc$ problem asks to find the value vector $\optval_\rmc \in \R^{|\states|}$ and a best adversary policy $\envpolopt \colon \states \to \Delta(\states)$. The decision version of this problem asks, for a given state $s$ and a threshold $\lambda$, whether $\optval(s) \geq \lambda$. We will focus on formalizing the main problem as a polynomial-size $\etr$ formula, which can be solved via an $\etr$ solver. Once we find the exact answer using this approach, we can simply solve the decision problem by adding the threshold constraint to the formula.

\paragraph{Karush-Kuhn-Tucker (KKT).} 
Before analyzing the adversary's policy, we recall the Karush-Kuhn-Tucker (KKT) technique \cite{boyd2004convex}. This technique provides an equation to find the optimal solution in convex problems. Suppose we want to minimize a function $f(x)$, subject to inequality rules $g_i(x) \le 0$ and equality rules $h_j(x) = 0$. 

To apply KKT, we introduce a penalty variable for each rule. We use $\mu_i$ for each inequality rule and $\nu_j$ for each equality rule. We combine all conditions into a master equation called the Lagrangian:
$$L(x, \mu, \nu) = f(x) + \sum_i \mu_i g_i(x) + \sum_j \nu_j h_j(x)$$

The KKT conditions state that the optimal solution $x^{\star}$ must satisfy four conditions:
\begin{enumerate}
    \item The solution $x^{\star}$ must satisfy all the original rules ($g_i(x^{\star}) \le 0$ and $h_j(x^{\star}) = 0$).
    \item The penalty variables for the inequalities cannot be negative ($\mu_i \ge 0$).
    \item If $g_i(x^{\star}) < 0$, its penalty is ignored ($\mu_i = 0$). In other words,  $\mu_i g_i(x^{\star}) = 0$ for every rule.
    \item The derivative of Lagrangian must be zero ($\nabla f(x^{\star}) + \sum_i \mu_i \nabla g_i(x^{\star}) + \sum_j \nu_j \nabla h_j(x^{\star}) = 0$).
\end{enumerate}

\paragraph{Slater's Condition.}
The KKT conditions are necessary and sufficient for global optimality if strong duality holds. A standard constraint to establish strong duality is \textit{Slater's condition} \cite{boyd2004convex}.
Slater's condition requires the existence of a feasible point $x$ in the interior of the problem domain. Formally, $x$ must strictly satisfy all non-linear inequality constraints ($g_i(x) < 0$), but only needs to weakly satisfy linear inequality constraints ($g_i(x) \le 0$) and equality constraints ($h_j(x) = 0$). Satisfying this condition guarantees that the KKT equations are sufficient to find the optimal solution.

To prove that our problem can be formalized via an $\etr$ formula, we need to express all conditions using polynomials. We first show that the $L_p$ uncertainty set constraints can be written as a system of polynomial inequalities.

\begin{lemma}\label{lem:uncertainty_poly}
    For any state $s \in \states$, the condition $\envpol_s \in \uncert(s)$ can be expressed as a system of polynomial constraints.
\end{lemma}

\begin{proof}
    By definition, the uncertainty set $\uncert(s)$  contains the transition probabilities such that $\sum_{t \in \succ(s)} |\envpol_{s,t} - \nominal_{s,t}|^p \le \radius_s^p$. 
    
    Because the absolute value function is not a polynomial, we introduce auxiliary variables $y_{s,t} \in \R$ to bound the absolute differences. For each $t \in \states$, we add the following pair of linear inequalities:
    \begin{align*}
        \envpol_{s,t} - \nominal_{s,t} - y_{s,t} &\le 0, \\
        \nominal_{s,t} - \envpol_{s,t} - y_{s,t} &\le 0.
    \end{align*}
    These inequalities force $y_{s,t} \ge |\envpol_{s,t} - \nominal_{s,t}|$. 
    Note that the non-negativity of the auxiliary variables is logically implied. 
    We replace the $L_p$ constraint with the following inequality:
    \begin{align*}
        \sum_{t \in \succ(s)} y_{s,t}^p - \radius_s^p \le 0.
    \end{align*}
    
    Since $p > 1$ is a constant integer, $y_{s,t}^p$ is a valid polynomial term. Thus, the condition is expressed by a system of polynomial inequalities.
\end{proof}

Next, we formulate the adversary's optimization problem. By applying the KKT conditions, we can describe the optimal worst-case transition distribution using a polynomial system.

\begin{lemma}\label{lem:kkt_poly}
    For any fixed value vector $\val$, the optimal transition distribution $\envpol_s$ chosen by the adversary can be defined by a system of polynomial equations and inequalities using KKT conditions.
\end{lemma}

\begin{proof}
    The adversary wants to choose the worst-case transition probabilities $\envpol_s$ to maximize the expected cost $\sum_{t \in \states} \envpol_{s,t} \val_t$, which is equivalent to minimizing $-\sum_{t \in \states} \envpol_{s,t} \val_t$. 

    Using Lemma~\ref{lem:uncertainty_poly}, we can write the feasibility condition $\envpol_s \in \uncert(s)$ using auxiliary variables $y_{s,t}$. Thus, the adversary's goal can be written as the following optimization problem:
    \begin{align*}
        & \text{\textbf{Minimize:}} & & -\sum_{t \in \succ(s)} \envpol_{s,t} \val_t \\
        & \text{\textbf{Subject to:}} 
        & & \left(\sum_{t \in \succ(s)} \envpol_{s,t}\right) - 1 = 0 
        && \text{(Probabilities must sum to 1)} \\
        &&& -\envpol_{s,t} \le 0 
        && \text{(Probabilities cannot be negative)} \\
        &&& \envpol_{s,t} - \nominal_{s,t} - y_{s,t} \le 0 
        && \text{(Upper bound for distance)} \\
        &&& \nominal_{s,t} - \envpol_{s,t} - y_{s,t} \le 0 
        && \text{(Lower bound for distance)} \\
        &&& \left(\sum_{t \in \succ(s)} y_{s,t}^p\right) - \radius_s^p \le 0 
        && \text{(Total $L_p$ distance limit)}
    \end{align*}

    This is a convex optimization problem because the goal function is linear and all constraints are convex. To prove that the KKT technique finds the true optimum, we verify that the problem has Slater's condition. 
    We evaluate this in two cases based on $\radius_s$:

    \textbf{Case 1 ($\radius_s > 0$):} We must find a strictly feasible point where the inequality constraints do not touch the exact edge ($< 0$). We choose the default probabilities ($\envpol_{s,t} = \nominal_{s,t}$) and set the auxiliary variables to zero ($y_{s,t} = 0$).
    \begin{itemize}
        \item All of the linear constraints are satisfied by the definition as
        $\envpol=\nominal_{s,t}$ is a valid probability, and all lower bounds and upper bounds are zero.
        \item The non-linear distance limit evaluates to 
        $(\sum_{t \in \succ(s)} 0^p) - \radius_s^p = - \radius_s^p$.
        Because $\radius_s > 0$, this inequality strictly satisfies.
    \end{itemize}
    Because we found a valid choice safely inside the boundaries, Slater's condition is met.

    \textbf{Case 2 ($\radius_s = 0$)}: In this case, the uncertainty set $\uncert(s)$ is a single point with the nominal distribution. Therefore, the optimal transition distribution is the nominal probabilities ($\envpol_{s,t} = \nominal_{s,t}$). Because there is no optimization here, we bypass the KKT formulation in this case. Instead, we express this deterministic optimal policy using linear equality constraints ($\envpol_{s,t} - \nominal_{s,t} = 0$). Since linear equations are valid polynomial equations, the optimal policy is defined by a polynomial system.

    Because Slater's condition holds, the optimum is characterized by the KKT conditions. By assigning penalty multipliers to the constraints, we construct the Lagrangian $L$. The fourth KKT condition requires the partial derivatives of $L$ with respect to $\envpol_{s,t}$ and $y_{s,t}$ to be zero. 

    Because the objective and all but the last constraint are linear, their partial derivatives evaluate to simple constants. The only non-linear term in the Lagrangian is the distance limit, which contains $\sum_{t \in \succ(s)} y_{s,t}^p$. Its derivative with respect to $y_{s,t}$ is $p y_{s,t}^{p-1}$. 

    Since $p > 1$ is an integer, the exponent $(p-1)$ is a non-negative integer, meaning no fractional powers are introduced. Therefore, the KKT system consists of polynomial equations and inequalities.
\end{proof}

Finally, we combine the Bellman equations and the KKT conditions into a single system. The following lemma proves that this combined polynomial system identifies the optimal solution.

\begin{lemma}\label{lem:unique_system}
    Let $\Sigma$ be the system of polynomial equations and inequalities formed by combining the Bellman equations and the KKT conditions for all states:
    \begin{enumerate}
        \item For every state $s \in \states$, the Bellman equation holds: $\val_s = \cost_s + \discount \sum_{t \in \states} \envpol_{s,t} \val_t$.
        \item For every state $s \in \states$, the KKT polynomial system derived in Lemma~\ref{lem:kkt_poly} holds for $\envpol_s$ and $\val$.
    \end{enumerate}
    A tuple $(\val, \envpol, y, \mu, \nu)$ is a valid solution to $\Sigma$ if and only if $\val$ is the optimal value vector $\optval$ and $\envpol$ is an optimal adversary policy $\envpolopt$.
\end{lemma}

\begin{proof}    
    First, we show that the optimal value vector $\optval$ and the optimal worst-case policy $\envpolopt$ produce a valid solution to $\Sigma$. 
    By \cref{lemma:bellman}, the robust Bellman operator has a unique fixed point. Thus, the true optimal value vector $\optval$ exists. 
    For this $\optval$, there exists an optimal adversary policy $\envpolopt$ that achieves the worst-case transition, satisfying $\sum_{t \in \states} \envpolopt_{s,t} \optval_t = \max_{p \in \uncert(s)} p^\top \optval$. 
    By substituting this into the first part of $\Sigma$, we get $\optval_s = \cost_s + \discount \sum_{t \in \states} \envpolopt_{s,t} \optval_t$, which satisfies the Bellman equations.
    
    Moreover, because $\envpolopt_s$ maximizes the expected cost for $\optval$, it is the optimal solution to the inner minimization problem in Lemma~\ref{lem:kkt_poly}. Because this inner problem is convex and satisfies Slater's condition, strong duality holds. This guarantees that the KKT conditions are necessary for optimality. Therefore, there exist real auxiliary variables $y$ and KKT multipliers $\mu, \nu$ that satisfy the KKT system for $\optval$ and $\envpolopt_s$. Thus, the tuple $(\optval, \envpolopt, y, \mu, \nu)$ forms a valid solution to $\Sigma$.  
    
    Second, suppose we have a solution $(\val', \envpol', y', \mu', \nu')$ for $\Sigma$. Because $\envpol'_s$ satisfies the KKT conditions for $\val'_s$, and Slater's condition holds, $\envpol'_s$ is the true worst-case transition distribution for $\val'$. 
    This means the inner product reaches the maximum: $\sum_{t \in \states} \envpol'_{s,t} \val'_t = \max_{p \in \uncert(s)} p^\top \val'$.
    Substituting this into the Bellman constraint of $\Sigma$ gives:
    \begin{equation*}
        \val'_s = \cost_s + \discount \max_{p \in \uncert(s)} p^\top \val'
    \end{equation*}
    This means $\val'$ is a fixed point of the robust Bellman operator. By \cref{lemma:bellman}, this operator has a unique fixed point. Therefore, $\val'$ must be the optimal value vector $\optval$. Consequently, $\envpol'$ is an optimal adversary policy.
\end{proof}

We now have a complete polynomial system $\Sigma$ that finds the exact value vector and a best policy for the environment. Since the decision problem asks whether $\optval(s) \geq \lambda$, we can answer it by adding this linear constraint to the system $\Sigma$.

\begin{proof}[Proof of \Cref{thm:lp-comp-class}-Item~\ref{item:etr-easyness}]
    By \cref{lem:unique_system}, the polynomial system $\Sigma$ identifies the optimal value vector $\optval$ and the optimal policy $\envpolopt$ as its real solutions. 
    
    The system $\Sigma$ is constructed entirely using real variables, polynomial equations, and polynomial inequalities, where all coefficients are rational. By definition, finding a solution to a system of polynomial equations and inequalities over the real numbers is a problem in the existential theory of the reals. Thus, the main $\discrmc$ problem can be solved via an $\etr$ solver.
    Moreover, by adding the constraint $\optval(s) \geq \lambda$ to $\Sigma$, we can evaluate the threshold condition. Therefore, the decision problem belongs to the complexity class $\etr$.
\end{proof}

\section{Additional Experimental Details}

\label{appendix:experiments}

\paragraph{Benchmarks.}
We evaluate our algorithm on five benchmarks. For each benchmark, we scale the state-space size up to $n=|\States|=256$. Reward is represented as negated cost.

\begin{itemize}
\item \textbf{Gridworld} (Adapted from \cite{andrew2018reinforcement}). We consider a $k \times k$ grid navigation task containing $n = k^2$ states. The agent starts at the top-left $(0,0)$ and navigates toward an absorbing goal at the bottom-right $(k-1, k-1)$ to receive a continuous +1 reward. Actions consist of the four cardinal directions. The environment dynamics are stochastic: intended movements succeed with a probability of 0.8, while slipping perpendicularly occurs with a 0.1 probability in either direction. Movements into boundaries result in self-loops. An absorbing trap is placed at $x = \lfloor (k-1)/2 \rfloor, y = k - 1 - x$, which imposes a continuous -1 penalty. All other transitions carry a standard step penalty of -0.01.

\item \textbf{Inventory Management.} (Adapted from \cite{puterman94}). The state space $s \in \{0, \dots, n-1\}$ denotes the current inventory level. Let $d_{\max} = \max(1, \lfloor(n-1)/2\rfloor)$ represent the maximum expected demand. At each step, the agent selects an order quantity from $\mathcal{A} = \{0, \operatorname{round}(d_{\max}/2), d_{\max}\}$, where $\operatorname{round}(x)$ rounds to the nearest integer, with half-integers rounded to the nearest even integer. Customer demand is modeled as a discrete triangular distribution supported on $\{0, \dots, d_{\max}\}$ with a peak at $m = \lfloor d_{\max}/2 \rfloor$ and unnormalized weights $w(d) = m - |d - m| + 1$. The reward function evaluates the expected profit under this distribution, assuming a unit revenue of 1, a unit holding cost of 0.1, and a unit ordering cost of 0.5.

\item \textbf{Machine Replacement} (adapted by \cite{wiesemann2013}). The state space $s \in \{0, \dots, n-1\}$ tracks the machine's degradation level, starting pristine at $s=0$ and fully broken at the absorbing state $s=n-1$. The action space consists of 3 options: \texttt{operate}, \texttt{repair}, and \texttt{replace}. The \texttt{operate} action causes the machine to degrade to $s+1$ with a 1/3 probability or remain in $s$ with a 2/3 probability. The \texttt{repair} action improves the condition to $s-1$ with a 3/4 probability, failing and remaining in $s$ with a 1/4 probability. The \texttt{replace} action deterministically resets the machine to $s=0$. Operating provides a state-dependent reward defined as $\frac{n-1-s}{n-1}$, whereas the \texttt{repair} and \texttt{replace} actions impose fixed penalties of -0.25 and -0.5.

\item \textbf{GARNET} (adapted from \cite{archibald1995generation}). Generalized Artificial Random Networks serve as unstructured random MDP baselines. We consider a state space of size $n$, where each state has 4 actions and a fixed transition branching factor of $b = 3$. The successor set $\succ(s,a) \subseteq \mathcal{S}$ of size $b$ is drawn uniformly at random for every state-action pair. We assign normalized integer weights $w(s, a, s') \sim \mathcal{U}\{1, 1000\}$ to all $s' \in \succ(s,a)$. The nominal transition probabilities are calculated as $P(s' \mid s, a) = \frac{w(s, a, s')}{\sum_{s''} w(s, a, s'')}$. Immediate rewards are given by uniformly sampled integers $r(s, a) \sim \mathcal{U}\{0, 10\}$.

\item \textbf{Long Chain}. A construction introduced in this paper to show a $\Theta(n)$ lower bound on policy iteration under a fixed discount factor~$\gamma$. With chain length~$k$ (so $n = 2k+1$), the state space consists of $k$ \emph{path} states $s_0, \dots, s_{k-1}$, $k$ absorbing \emph{leaf} states $l_0, \dots, l_{k-1}$, and an absorbing sink $s_{\mathrm{end}}$. From each path state~$s_i$, action \texttt{path} deterministically advances to $s_{i+1}$ (or to $s_{\mathrm{end}}$ when $i=k-1$), while action \texttt{leaf} jumps to $l_i$. Per-step rewards are $0$ on path states, $1$ on leaves, and $\gamma^{-(k+1)}$ on the sink; the inflated sink reward is calibrated so that \texttt{path} is optimal at every state by a unit margin. Initialized with \texttt{leaf} everywhere, policy iteration flips one path state per outer step, requiring $k = \Theta(n)$ iterations to converge. Although the sink reward grows in $k$, its bit-size remains $\mathcal{O}(k\log(1/\gamma))$, which is polynomial in the input size for any constant $\gamma$.

\end{itemize}

\paragraph{Results}

\Cref{fig:main_results_g9,fig:main_results_g99,fig:main_results_g995} report the outer \texttt{RMDP-PI} and total inner \texttt{RMC-PI} iteration counts for different values of $\gamma,\delta$ under both $L_1$ and $L_\infty$ uncertainty sets.

\begin{figure}[h]
  \centering
  \includegraphics[width=0.8\linewidth]{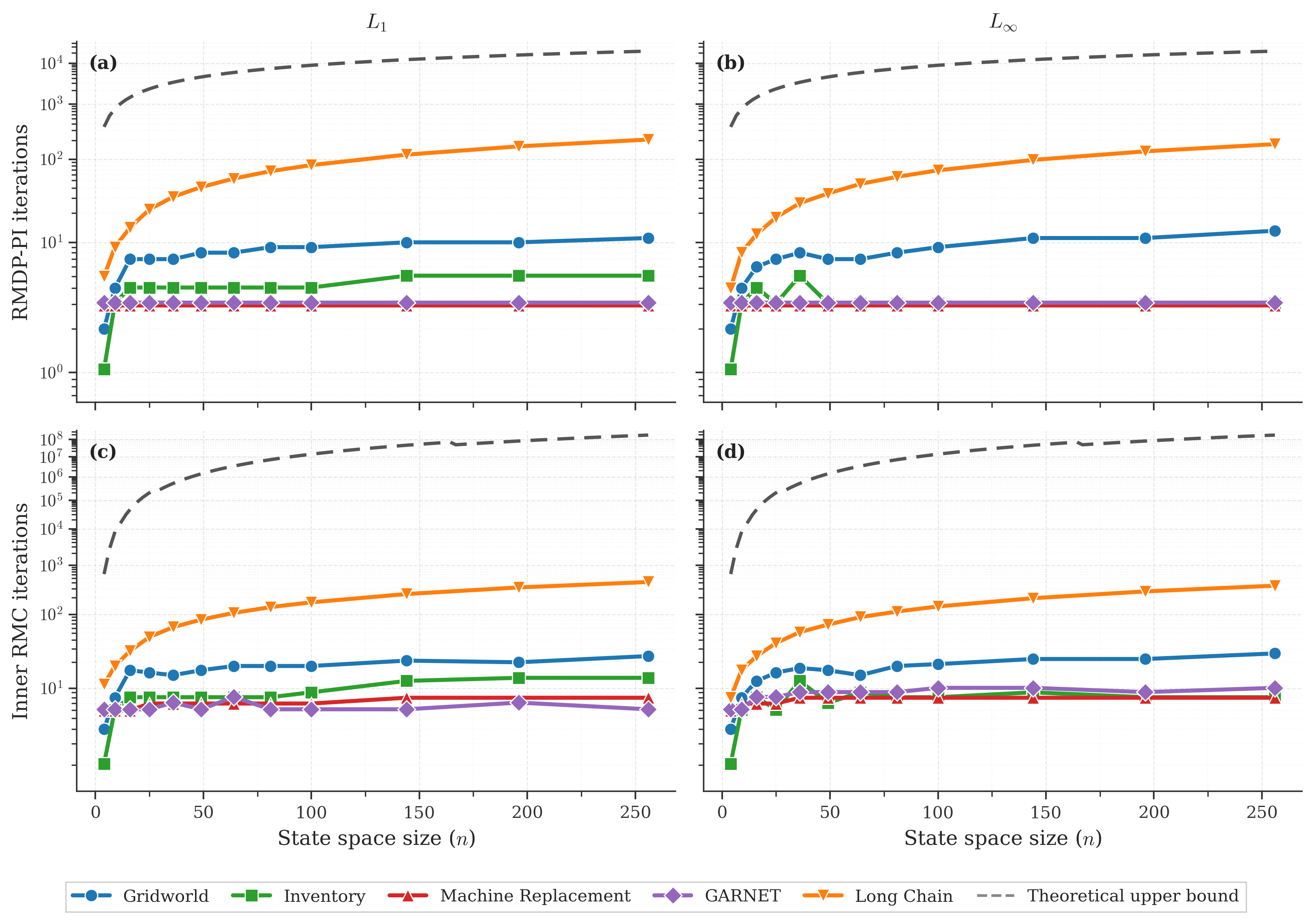}
  \caption{%
     Outer \texttt{RMDP-PI} (a,~b) and total inner \texttt{RMC-PI} (c,~d)
    iterations versus state-space size~$n$ on five benchmarks, under $L_1$
    (left) and $L_\infty$ (right) uncertainty sets, with $\gamma=0.9$,
    $\delta=0.05$. Dashed line: theoretical upper bound.%
  }
  \label{fig:main_results_g9}
\end{figure}

\begin{figure}[h]
  \centering
  \includegraphics[width=0.8\linewidth]{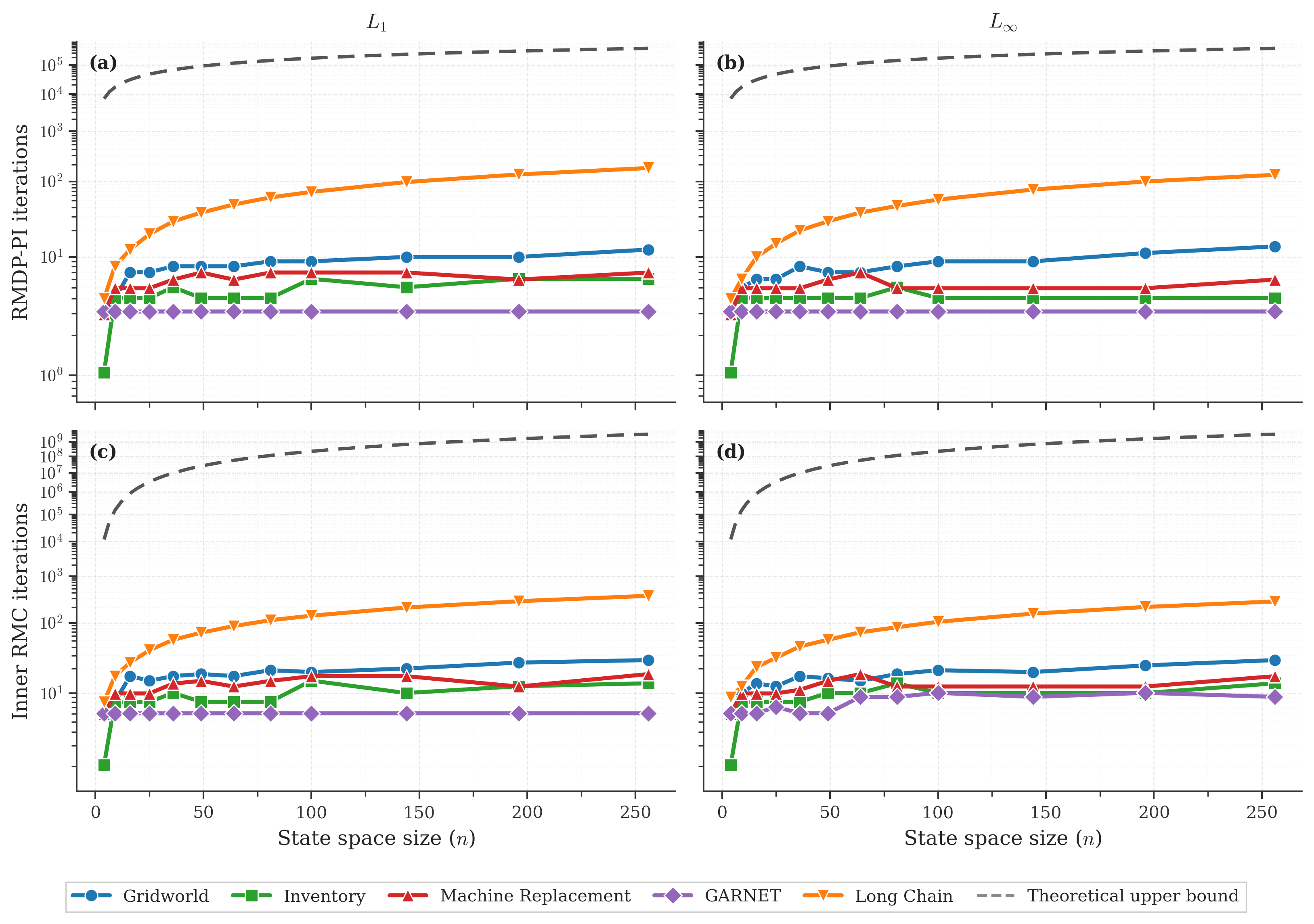}
  \caption{%
     Outer \texttt{RMDP-PI} (a,~b) and total inner \texttt{RMC-PI} (c,~d)
    iterations versus state-space size~$n$ on five benchmarks, under $L_1$
    (left) and $L_\infty$ (right) uncertainty sets, with $\gamma=0.99$,
    $\delta=0.01$. Dashed line: theoretical upper bound.%
  }
  \label{fig:main_results_g99}
\end{figure}

\begin{figure}[h]
  \centering
  \includegraphics[width=0.8\linewidth]{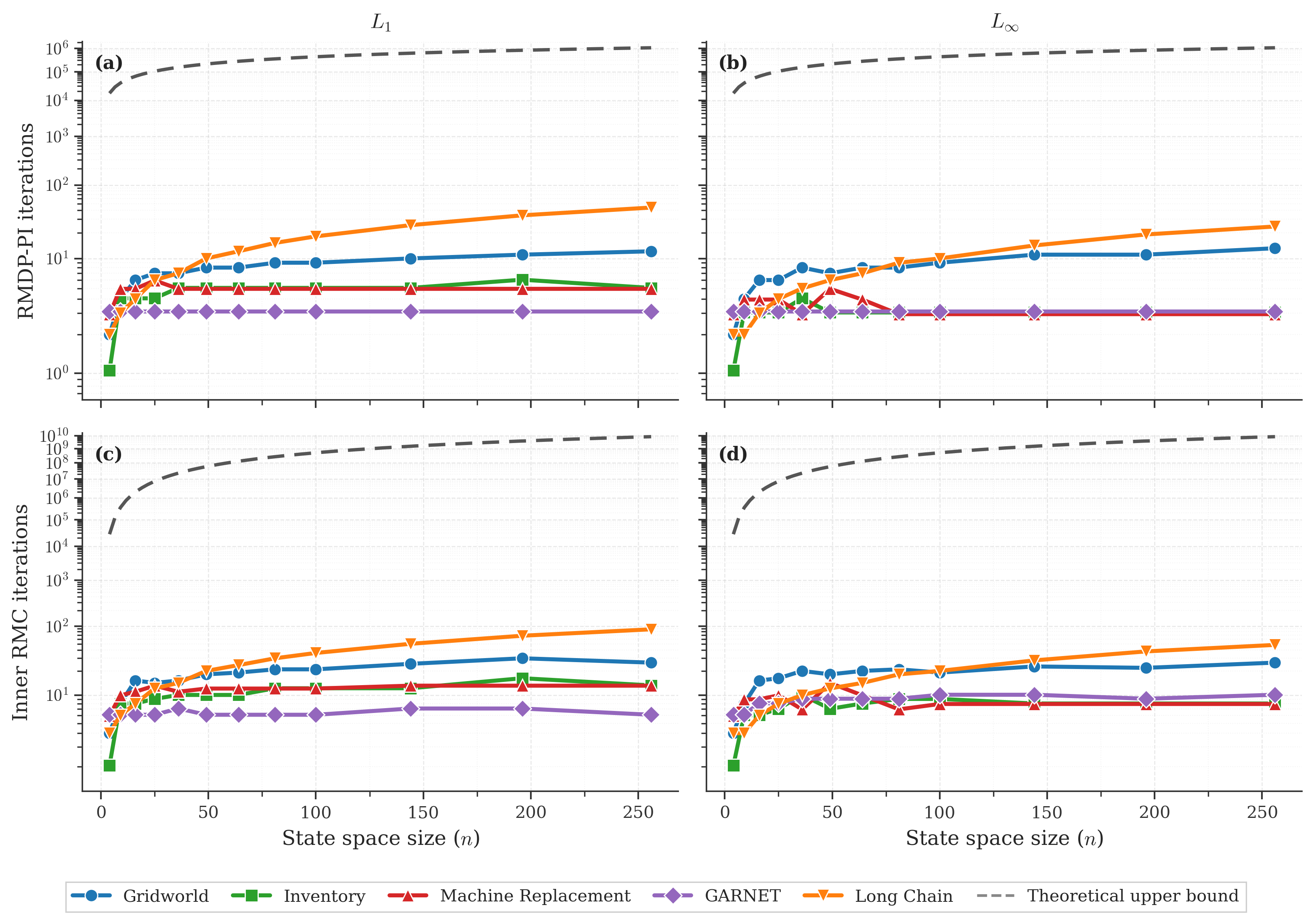}
  \caption{%
     Outer \texttt{RMDP-PI} (a,~b) and total inner \texttt{RMC-PI} (c,~d)
    iterations versus state-space size~$n$ on five benchmarks, under $L_1$
    (left) and $L_\infty$ (right) uncertainty sets, with $\gamma=0.995$,
    $\delta=0.05$. Dashed line: theoretical upper bound.%
  }
  \label{fig:main_results_g995}
\end{figure}

%%%%%%%%%%%%%%%%%%%%%%%%%%%%%%%%%%%%%%%%%%%%%%%%%%%%%%%%%%%%

\clearpage
\section*{NeurIPS Paper Checklist}

\begin{enumerate}

\item {\bf Claims}
    \item[] Question: Do the main claims made in the abstract and introduction accurately reflect the paper's contributions and scope?
    \item[] Answer: \answerYes{} % Replace by \answerYes{}, \answerNo{}, or \answerNA{}.
    \item[] Justification: 
    The abstract and introduction summarize our main results, such as the complexity of RMDPs and RMCs, which match the proofs and experiments in the paper.
    \item[] Guidelines:
    \begin{itemize}
        \item The answer \answerNA{} means that the abstract and introduction do not include the claims made in the paper.
        \item The abstract and/or introduction should clearly state the claims made, including the contributions made in the paper and important assumptions and limitations. A \answerNo{} or \answerNA{} answer to this question will not be perceived well by the reviewers. 
        \item The claims made should match theoretical and experimental results, and reflect how much the results can be expected to generalize to other settings. 
        \item It is fine to include aspirational goals as motivation as long as it is clear that these goals are not attained by the paper. 
    \end{itemize}

\item {\bf Limitations}
    \item[] Question: Does the paper discuss the limitations of the work performed by the authors?
    \item[] Answer: \answerYes{} % Replace by \answerYes{}, \answerNo{}, or \answerNA{}.
    \item[] Justification: We explain that we still need to find tighter bounds for $L_p$ uncertainty sets and lower bounds for $L_1$ and $L_{\infty}$ RMDPs.
    \item[] Guidelines:
    \begin{itemize}
        \item The answer \answerNA{} means that the paper has no limitation while the answer \answerNo{} means that the paper has limitations, but those are not discussed in the paper. 
        \item The authors are encouraged to create a separate ``Limitations'' section in their paper.
        \item The paper should point out any strong assumptions and how robust the results are to violations of these assumptions (e.g., independence assumptions, noiseless settings, model well-specification, asymptotic approximations only holding locally). The authors should reflect on how these assumptions might be violated in practice and what the implications would be.
        \item The authors should reflect on the scope of the claims made, e.g., if the approach was only tested on a few datasets or with a few runs. In general, empirical results often depend on implicit assumptions, which should be articulated.
        \item The authors should reflect on the factors that influence the performance of the approach. For example, a facial recognition algorithm may perform poorly when image resolution is low or images are taken in low lighting. Or a speech-to-text system might not be used reliably to provide closed captions for online lectures because it fails to handle technical jargon.
        \item The authors should discuss the computational efficiency of the proposed algorithms and how they scale with dataset size.
        \item If applicable, the authors should discuss possible limitations of their approach to address problems of privacy and fairness.
        \item While the authors might fear that complete honesty about limitations might be used by reviewers as grounds for rejection, a worse outcome might be that reviewers discover limitations that aren't acknowledged in the paper. The authors should use their best judgment and recognize that individual actions in favor of transparency play an important role in developing norms that preserve the integrity of the community. Reviewers will be specifically instructed to not penalize honesty concerning limitations.
    \end{itemize}

\item {\bf Theory assumptions and proofs}
    \item[] Question: For each theoretical result, does the paper provide the full set of assumptions and a complete (and correct) proof?
    \item[] Answer: \answerYes{} % Replace by \answerYes{}, \answerNo{}, or \answerNA{}.
    \item[] Justification: We state our assumptions, and we provide complete proofs for our theorems in the main sections and appendices.
    \item[] Guidelines:
    \begin{itemize}
        \item The answer \answerNA{} means that the paper does not include theoretical results. 
        \item All the theorems, formulas, and proofs in the paper should be numbered and cross-referenced.
        \item All assumptions should be clearly stated or referenced in the statement of any theorems.
        \item The proofs can either appear in the main paper or the supplemental material, but if they appear in the supplemental material, the authors are encouraged to provide a short proof sketch to provide intuition. 
        \item Inversely, any informal proof provided in the core of the paper should be complemented by formal proofs provided in appendix or supplemental material.
        \item Theorems and Lemmas that the proof relies upon should be properly referenced. 
    \end{itemize}

    \item {\bf Experimental result reproducibility}
    \item[] Question: Does the paper fully disclose all the information needed to reproduce the main experimental results of the paper to the extent that it affects the main claims and/or conclusions of the paper (regardless of whether the code and data are provided or not)?
    \item[] Answer: \answerYes{} % Replace by \answerYes{}, \answerNo{}, or \answerNA{}.
    \item[] Justification: We provide enough details about our benchmarks and methods in the main text and the appendices.
    \item[] Guidelines:
    \begin{itemize}
        \item The answer \answerNA{} means that the paper does not include experiments.
        \item If the paper includes experiments, a \answerNo{} answer to this question will not be perceived well by the reviewers: Making the paper reproducible is important, regardless of whether the code and data are provided or not.
        \item If the contribution is a dataset and\slash or model, the authors should describe the steps taken to make their results reproducible or verifiable. 
        \item Depending on the contribution, reproducibility can be accomplished in various ways. For example, if the contribution is a novel architecture, describing the architecture fully might suffice, or if the contribution is a specific model and empirical evaluation, it may be necessary to either make it possible for others to replicate the model with the same dataset, or provide access to the model. In general. releasing code and data is often one good way to accomplish this, but reproducibility can also be provided via detailed instructions for how to replicate the results, access to a hosted model (e.g., in the case of a large language model), releasing of a model checkpoint, or other means that are appropriate to the research performed.
        \item While NeurIPS does not require releasing code, the conference does require all submissions to provide some reasonable avenue for reproducibility, which may depend on the nature of the contribution. For example
        \begin{enumerate}
            \item If the contribution is primarily a new algorithm, the paper should make it clear how to reproduce that algorithm.
            \item If the contribution is primarily a new model architecture, the paper should describe the architecture clearly and fully.
            \item If the contribution is a new model (e.g., a large language model), then there should either be a way to access this model for reproducing the results or a way to reproduce the model (e.g., with an open-source dataset or instructions for how to construct the dataset).
            \item We recognize that reproducibility may be tricky in some cases, in which case authors are welcome to describe the particular way they provide for reproducibility. In the case of closed-source models, it may be that access to the model is limited in some way (e.g., to registered users), but it should be possible for other researchers to have some path to reproducing or verifying the results.
        \end{enumerate}
    \end{itemize}

\item {\bf Open access to data and code}
    \item[] Question: Does the paper provide open access to the data and code, with sufficient instructions to faithfully reproduce the main experimental results, as described in supplemental material?
    \item[] Answer: \answerYes{} % Replace by \answerYes{}, \answerNo{}, or \answerNA{}.
    \item[] Justification: 
    We describe the computational steps, libraries, and methods used in the readme file of the supplementary materials.
    \item[] Guidelines:
    \begin{itemize}
        \item The answer \answerNA{} means that paper does not include experiments requiring code.
        \item Please see the NeurIPS code and data submission guidelines (\url{https://neurips.cc/public/guides/CodeSubmissionPolicy}) for more details.
        \item While we encourage the release of code and data, we understand that this might not be possible, so \answerNo{} is an acceptable answer. Papers cannot be rejected simply for not including code, unless this is central to the contribution (e.g., for a new open-source benchmark).
        \item The instructions should contain the exact command and environment needed to run to reproduce the results. See the NeurIPS code and data submission guidelines (\url{https://neurips.cc/public/guides/CodeSubmissionPolicy}) for more details.
        \item The authors should provide instructions on data access and preparation, including how to access the raw data, preprocessed data, intermediate data, and generated data, etc.
        \item The authors should provide scripts to reproduce all experimental results for the new proposed method and baselines. If only a subset of experiments are reproducible, they should state which ones are omitted from the script and why.
        \item At submission time, to preserve anonymity, the authors should release anonymized versions (if applicable).
        \item Providing as much information as possible in supplemental material (appended to the paper) is recommended, but including URLs to data and code is permitted.
    \end{itemize}

\item {\bf Experimental setting/details}
    \item[] Question: Does the paper specify all the training and test details (e.g., data splits, hyperparameters, how they were chosen, type of optimizer) necessary to understand the results?
    \item[] Answer: \answerYes{} % Replace by \answerYes{}, \answerNo{}, or \answerNA{}.
    \item[] Justification: We list the parameters, such as state-space sizes and discount factors, in the experimental results Section and Appendix D.
    \item[] Guidelines:
    \begin{itemize}
        \item The answer \answerNA{} means that the paper does not include experiments.
        \item The experimental setting should be presented in the core of the paper to a level of detail that is necessary to appreciate the results and make sense of them.
        \item The full details can be provided either with the code, in appendix, or as supplemental material.
    \end{itemize}

\item {\bf Experiment statistical significance}
    \item[] Question: Does the paper report error bars suitably and correctly defined or other appropriate information about the statistical significance of the experiments?
    \item[] Answer: \answerYes{} % Replace by \answerYes{}, \answerNo{}, or \answerNA{}.
    \item[] Justification: 
    Most of the benchmarks do not involve randomization in their construction. The policy iteration algorithm tested is deterministic and uses exact math to find the answers. Since the algorithm does not use random sampling to solve the problem, there is no statistical variance, so error bars are not needed.
    \item[] Guidelines:
    \begin{itemize}
        \item The answer \answerNA{} means that the paper does not include experiments.
        \item The authors should answer \answerYes{} if the results are accompanied by error bars, confidence intervals, or statistical significance tests, at least for the experiments that support the main claims of the paper.
        \item The factors of variability that the error bars are capturing should be clearly stated (for example, train/test split, initialization, random drawing of some parameter, or overall run with given experimental conditions).
        \item The method for calculating the error bars should be explained (closed form formula, call to a library function, bootstrap, etc.)
        \item The assumptions made should be given (e.g., Normally distributed errors).
        \item It should be clear whether the error bar is the standard deviation or the standard error of the mean.
        \item It is OK to report 1-sigma error bars, but one should state it. The authors should preferably report a 2-sigma error bar than state that they have a 96\% CI, if the hypothesis of Normality of errors is not verified.
        \item For asymmetric distributions, the authors should be careful not to show in tables or figures symmetric error bars that would yield results that are out of range (e.g., negative error rates).
        \item If error bars are reported in tables or plots, the authors should explain in the text how they were calculated and reference the corresponding figures or tables in the text.
    \end{itemize}

\item {\bf Experiments compute resources}
    \item[] Question: For each experiment, does the paper provide sufficient information on the computer resources (type of compute workers, memory, time of execution) needed to reproduce the experiments?
    \item[] Answer: \answerYes{} % Replace by \answerYes{}, \answerNo{}, or \answerNA{}.
    \item[] Justification: We state the specific hardware used for our tests, which was a machine with an Intel Core Ultra 5 processor and 16GB of RAM.
    \item[] Guidelines:
    \begin{itemize}
        \item The answer \answerNA{} means that the paper does not include experiments.
        \item The paper should indicate the type of compute workers CPU or GPU, internal cluster, or cloud provider, including relevant memory and storage.
        \item The paper should provide the amount of compute required for each of the individual experimental runs as well as estimate the total compute. 
        \item The paper should disclose whether the full research project required more compute than the experiments reported in the paper (e.g., preliminary or failed experiments that didn't make it into the paper). 
    \end{itemize}
    
\item {\bf Code of ethics}
    \item[] Question: Does the research conducted in the paper conform, in every respect, with the NeurIPS Code of Ethics \url{https://neurips.cc/public/EthicsGuidelines}?
    \item[] Answer: \answerYes{} % Replace by \answerYes{}, \answerNo{}, or \answerNA{}.
    \item[] Justification: Our research follows standard ethical guidelines for Neurips.
    \item[] Guidelines:
    \begin{itemize}
        \item The answer \answerNA{} means that the authors have not reviewed the NeurIPS Code of Ethics.
        \item If the authors answer \answerNo, they should explain the special circumstances that require a deviation from the Code of Ethics.
        \item The authors should make sure to preserve anonymity (e.g., if there is a special consideration due to laws or regulations in their jurisdiction).
    \end{itemize}

\item {\bf Broader impacts}
    \item[] Question: Does the paper discuss both potential positive societal impacts and negative societal impacts of the work performed?
    \item[] Answer: \answerNA{} % Replace by \answerYes{}, \answerNo{}, or \answerNA{}.
    \item[] Justification: This is foundational theoretical research that does not have a direct societal impact or specific negative applications.
    \item[] Guidelines:
    \begin{itemize}
        \item The answer \answerNA{} means that there is no societal impact of the work performed.
        \item If the authors answer \answerNA{} or \answerNo, they should explain why their work has no societal impact or why the paper does not address societal impact.
        \item Examples of negative societal impacts include potential malicious or unintended uses (e.g., disinformation, generating fake profiles, surveillance), fairness considerations (e.g., deployment of technologies that could make decisions that unfairly impact specific groups), privacy considerations, and security considerations.
        \item The conference expects that many papers will be foundational research and not tied to particular applications, let alone deployments. However, if there is a direct path to any negative applications, the authors should point it out. For example, it is legitimate to point out that an improvement in the quality of generative models could be used to generate Deepfakes for disinformation. On the other hand, it is not needed to point out that a generic algorithm for optimizing neural networks could enable people to train models that generate Deepfakes faster.
        \item The authors should consider possible harms that could arise when the technology is being used as intended and functioning correctly, harms that could arise when the technology is being used as intended but gives incorrect results, and harms following from (intentional or unintentional) misuse of the technology.
        \item If there are negative societal impacts, the authors could also discuss possible mitigation strategies (e.g., gated release of models, providing defenses in addition to attacks, mechanisms for monitoring misuse, mechanisms to monitor how a system learns from feedback over time, improving the efficiency and accessibility of ML).
    \end{itemize}
    
\item {\bf Safeguards}
    \item[] Question: Does the paper describe safeguards that have been put in place for responsible release of data or models that have a high risk for misuse (e.g., pre-trained language models, image generators, or scraped datasets)?
    \item[] Answer: \answerNA{} % Replace by \answerYes{}, \answerNo{}, or \answerNA{}.
    \item[] Justification: The paper presents theoretical mathematics and algorithms, which poses no safety risks or danger of misuse.
    \item[] Guidelines:
    \begin{itemize}
        \item The answer \answerNA{} means that the paper poses no such risks.
        \item Released models that have a high risk for misuse or dual-use should be released with necessary safeguards to allow for controlled use of the model, for example by requiring that users adhere to usage guidelines or restrictions to access the model or implementing safety filters. 
        \item Datasets that have been scraped from the Internet could pose safety risks. The authors should describe how they avoided releasing unsafe images.
        \item We recognize that providing effective safeguards is challenging, and many papers do not require this, but we encourage authors to take this into account and make a best faith effort.
    \end{itemize}

\item {\bf Licenses for existing assets}
    \item[] Question: Are the creators or original owners of assets (e.g., code, data, models), used in the paper, properly credited and are the license and terms of use explicitly mentioned and properly respected?
    \item[] Answer: \answerYes{} % Replace by \answerYes{}, \answerNo{}, or \answerNA{}.
    \item[] Justification: All benchmark sources are cited in the references.
    \item[] Guidelines:
    \begin{itemize}
        \item The answer \answerNA{} means that the paper does not use existing assets.
        \item The authors should cite the original paper that produced the code package or dataset.
        \item The authors should state which version of the asset is used and, if possible, include a URL.
        \item The name of the license (e.g., CC-BY 4.0) should be included for each asset.
        \item For scraped data from a particular source (e.g., website), the copyright and terms of service of that source should be provided.
        \item If assets are released, the license, copyright information, and terms of use in the package should be provided. For popular datasets, \url{paperswithcode.com/datasets} has curated licenses for some datasets. Their licensing guide can help determine the license of a dataset.
        \item For existing datasets that are re-packaged, both the original license and the license of the derived asset (if it has changed) should be provided.
        \item If this information is not available online, the authors are encouraged to reach out to the asset's creators.
    \end{itemize}

\item {\bf New assets}
    \item[] Question: Are new assets introduced in the paper well documented and is the documentation provided alongside the assets?
    \item[] Answer: \answerNA{} % Replace by \answerYes{}, \answerNo{}, or \answerNA{}.
    \item[] Justification: We do not release new software or datasets.
    \item[] Guidelines:
    \begin{itemize}
        \item The answer \answerNA{} means that the paper does not release new assets.
        \item Researchers should communicate the details of the dataset\slash code\slash model as part of their submissions via structured templates. This includes details about training, license, limitations, etc. 
        \item The paper should discuss whether and how consent was obtained from people whose asset is used.
        \item At submission time, remember to anonymize your assets (if applicable). You can either create an anonymized URL or include an anonymized zip file.
    \end{itemize}

\item {\bf Crowdsourcing and research with human subjects}
    \item[] Question: For crowdsourcing experiments and research with human subjects, does the paper include the full text of instructions given to participants and screenshots, if applicable, as well as details about compensation (if any)? 
    \item[] Answer: \answerNA{} % Replace by \answerYes{}, \answerNo{}, or \answerNA{}.
    \item[] Justification: This study is purely theoretical and does not involve human subjects or crowdsourcing.
    \item[] Guidelines:
    \begin{itemize}
        \item The answer \answerNA{} means that the paper does not involve crowdsourcing nor research with human subjects.
        \item Including this information in the supplemental material is fine, but if the main contribution of the paper involves human subjects, then as much detail as possible should be included in the main paper. 
        \item According to the NeurIPS Code of Ethics, workers involved in data collection, curation, or other labor should be paid at least the minimum wage in the country of the data collector. 
    \end{itemize}

\item {\bf Institutional review board (IRB) approvals or equivalent for research with human subjects}
    \item[] Question: Does the paper describe potential risks incurred by study participants, whether such risks were disclosed to the subjects, and whether Institutional Review Board (IRB) approvals (or an equivalent approval/review based on the requirements of your country or institution) were obtained?
    \item[] Answer: \answerNA{} % Replace by \answerYes{}, \answerNo{}, or \answerNA{}.
    \item[] Justification: No human participants were involved, so approval is not required.
    \item[] Guidelines:
    \begin{itemize}
        \item The answer \answerNA{} means that the paper does not involve crowdsourcing nor research with human subjects.
        \item Depending on the country in which research is conducted, IRB approval (or equivalent) may be required for any human subjects research. If you obtained IRB approval, you should clearly state this in the paper. 
        \item We recognize that the procedures for this may vary significantly between institutions and locations, and we expect authors to adhere to the NeurIPS Code of Ethics and the guidelines for their institution. 
        \item For initial submissions, do not include any information that would break anonymity (if applicable), such as the institution conducting the review.
    \end{itemize}

\item {\bf Declaration of LLM usage}
    \item[] Question: Does the paper describe the usage of LLMs if it is an important, original, or non-standard component of the core methods in this research? Note that if the LLM is used only for writing, editing, or formatting purposes and does \emph{not} impact the core methodology, scientific rigor, or originality of the research, declaration is not required.
    %this research? 
    \item[] Answer: \answerNA{} % Replace by \answerYes{}, \answerNo{}, or \answerNA{}.
    \item[] Justification: Our methodology does not use LLMs as a core method.
    \item[] Guidelines:
    \begin{itemize}
        \item The answer \answerNA{} means that the core method development in this research does not involve LLMs as any important, original, or non-standard components.
        \item Please refer to our LLM policy in the NeurIPS handbook for what should or should not be described.
    \end{itemize}

\end{enumerate}

\end{document}